\documentclass[draftclsnofoot,onecolumn]{IEEEtran}
\usepackage{amsmath,amsfonts,amssymb,amsthm}
\usepackage{array}
\usepackage{textcomp}
\usepackage{stfloats}
\usepackage{url}
\usepackage{verbatim}
\usepackage{pdfpages}
\usepackage{cite}
\newtheorem{theorem}{Theorem}
\newtheorem{lemma}{Lemma}
\newtheorem{definition}{Definition}

\newtheorem{proposition}{Proposition}

\newtheorem{remark}{Remark}
\usepackage{algpseudocode}
\usepackage{algorithmicx,algorithm}

\usepackage{array}
\usepackage{subfigure} 

\usepackage{booktabs}
\usepackage{multirow}
\usepackage{makecell}
\usepackage{diagbox}
\usepackage{color}

\usepackage{url}

\usepackage{flushend}

\begin{document}

\title{Integer-Forcing Precoding: A Geometric Perspective and Low-Complexity Algorithm Design}

\author{Junren Qin, Fan Jiang, Tao Yang, Shanxiang Lyu, Rongke Liu, \IEEEmembership{Senior Member, IEEE}, Shi Jin, \IEEEmembership{Fellow, IEEE}
\thanks{
	This work was supported in part by Mobile Information Networks National Science and Technology Major Project under Grant No. 2025ZD1305200, in part by Low-Altitude Airspace Strategic Program Portfolio under Grant Z25306105.
	
	Junren Qin is with the School of Electronics Information Engineering, Beihang University, Beijing 100191, China, the Pengcheng Laboratory (PCL), Shenzhen 518066, China, and the Shenzhen Institute of Beihang University, Shenzhen, China (e-mail: junrenqin@buaa.edu.cn).
	
	Fan Jiang is with the Pengcheng Laboratory (PCL), Shenzhen 518066, China (e-mail: jiangf02@pcl.ac.cn).
	
	Shanxiang Lyu is with the College of Cyber Security, Jinan University, Guangzhou 510632, China (E-mail: lsx07@jnu.edu.cn).
	
	Tao Yang is with the School of Electronics Information Engineering, Beihang University, Beijing 100191, China (e-mail: tyang@buaa.edu.cn;).
	
	Rongke Liu is with the Shenzhen Institute of Beihang University, Shenzhen, China, and is also with the School of Electronics Information Engineering, Beihang University, Beijing 100191, China (e-mail: rongke\_liu@buaa.edu.cn).
	
	Shi Jin is with the School of Information Science and Engineering, Southeast University, Nanjing 210096, China (e-mail: jinshi@seu.edu.cn).
	
	Corresponding authors: Rongke Liu, Tao Yang, Fan Jiang.
	}
}



\maketitle

\begin{abstract}
\textcolor{blue}{The joint coefficient--scaling design of integer-forcing (IF) precoding couples a full-rank integer coefficient matrix with a positive diagonal scaling matrix, resulting in a mixed discrete--continuous optimization problem. We first distinguish the original full-rank formulation and its finite-field invertibility requirement from the unimodular restriction studied in the subsequent analysis. Under this restriction, we formulate coefficient selection as the unimodular minimum-sum basis problem (UMSBP). By parameterizing the scaling matrix with its normalized diagonal vector, we define a deterministic exact single-output UMSBP oracle whose preimages partition the scaling manifold into disjoint coefficient-labelled regions. The restricted objective is consequently represented as the lower envelope of fixed-coefficient quadratic functions. Based on this formulation, we derive an exact decomposition into fixed-coefficient scaling optimization, coefficient-consistency verification, cross-region exploration, and candidate comparison. We prove that exact UMSBP evaluation is NP-hard through a polynomial-time reduction from exact Euclidean CVP. For a fixed coefficient matrix, we derive an equivalent convex formulation, its normalized stationarity condition, and a matrix-balancing interpretation. Under the stated positivity and finite-projective-diameter conditions, we further prove that the associated reciprocal map is contractive in Hilbert's projective metric, yielding directional convergence and an explicit iteration bound. Guided by these results, we develop Multi-Cone Nested Stochastic Pattern Search (MCN-SPS), which combines sampled cross-region exploration with rate-safeguarded alternating refinement. Its practical implementation replaces the exact oracle with LLL as a low-complexity unimodular candidate generator, while the acceptance rule guarantees non-decreasing rates. Under the adopted arithmetic-complexity model, its complexity is $\mathcal O\left(K^4\log K\log_2(r_0)\right)$. Numerical results demonstrate a favorable empirical performance--complexity tradeoff, particularly in overloaded MIMO systems.}
\end{abstract}


\section{Introduction}
\IEEEPARstart{B}{enifited} from creating spatial diversity and performing spatial multiplexing, multiple-input multiple-output (MIMO) demonstrates powerful performance \cite{hampton2013MIMO}, and becomes the core technology of modern communication. Developed through the massive MIMO (m-MIMO) \cite{hassibi2003muchmassivemimo} in fifth-generation (5G) and the Extremely large-scale MIMO (XL-MIMO) \cite{wang2023XLMIMO} in sixth-generation (6G), the reliability of communication is enhanced by the channel hardening effect \cite{hochwald2004multiple} aroused when the number of antennas at base station (BS) $N$ is larger than the number of single antenna user equipments (UEs) $K$. Towards the massive communication and ubiquitous connectivity applications in 6G, the demand for increased connectivity and user density has become a key performance indicator, and the increasing trend of these parameters is expected to continue in near future \cite{rajatheva2020white}. 
This trend draws increasing attention to overloaded MIMO systems, characterized by \(K\geq N\), where limited spatial degrees of freedom weaken the channel-hardening effect.

Considering the downlink, the messages are delivered from BS to UEs, which enables all messages of UEs to broadcast with full bandwidth. The capacity region of MIMO broadcasting (BC) channel can be established from the capacity region of MIMO multiple-access channel via the uplink-downlink duality, and can be achieved so far by dirty-paper coding (DPC) \cite{costa1983writing,yu2004sum}. Unfortunately, since the high implementation complexity, DPC is widely regarded as a theoretical tool. To decrease the complexity under a tradeoff on performance, Tomlinson-Harashima Precoding (THP) \cite{sun2014tomlinson} employs the decision feedback equalization (DFE) into the transmitter side. Furthermore, vector-perturbation (VP) \cite{hochwald2005vector,avner2015vector} precoding regards the interference cancellation process of DFE as the solution of closest vector problem (CVP) in lattice theory \cite{micciancio2002complexity}. To minimize transmission power, VP employs an accurate sphere decoder (SD) combined Lenstra-Lenstra-Lov$\acute{a}$sz (LLL) algorithm \cite{nguyen2010lll} and sphere decoding algorithm \cite{vikalo2005sphere}. These methods are all realized under the inspiration of interference canceling, which cannot be directly applied to the overload MIMO. 
\textcolor{blue}{When the number of served UEs exceeds the available spatial dimensions, multi-user interference (MUI) cannot be completely eliminated, which limits interference-cancellation-based precoding \cite{KrunzS23}.}
Meanwhile, the lack of spatial dimension in the overload MIMO decrease the accuracy of successive interference cancellation (SIC) in THP, and lose the effectiveness of LLL algorithm in VP. These lead to a weak performance on THP and a high complexity on VP.

On the other hand, the linear precoding rectify MUI by multiplying a precoding matrix related to channel matrix with the signal or message, and becoming a large class due to its easy-to-implement characteristics. The popular methods in this category are zero-forcing (ZF) precoding and regularized ZF (RZF) precoding \cite{zhang2021local,chen2022generalized,wang2024efficient}. For the precoding matrix, ZF employs the inverse (or pseudo inverse) of channel matrix, and RZF conducted a normalization matrix by minimum mean square error (MMSE) filtering. When $N\gg K$ in MMIMO or XL-MIMO, both ZF and RZF can approach the channel capacity \cite{hampton2013MIMO,albreem2021overview}, but the situation is quite the opposite in the overload MIMO. Regarding the ZF, since the channel matrix is full column-rank, the multiplication between channel matrix and its pseudo inverse outputs the least norm and least square solution, and the gap between the outputs and identity matrix are enlarged during the ratio $K/N$ expansion. Meanwhile, in the pseudo inverse of channel matrix, the expansion of condition number under the overload MIMO cause the increased additive noise contained in the outputs of precoder. Although these flaws are attenuated in RZF by the MMSE filter, the effects caused by overload MIMO are still non-negligible and an adaptable solution for overload MIMO is required.

Recently, based on the lattice-reduction-aided (LRA) precoding \cite{stern2016advanced}, compute-and-forward (CF) \cite{hong2012reverse} and integer-forcing receiver \cite{SakzadHV13,HarshanSV17}, the integer-forcing (IF) precoding \cite{silva2017integer} provides an effective methods to restrain MUI without a significant increased complexity. 
Rather than separating users' codewords solely through equalization and beamforming, integer-forcing (IF) precoding applies an inverse linear transformation at the transmitter so that the integer-linear combination (ILC) decoded by each user corresponds to its intended message \cite{silva2017integer}. This construction enables precoded IF to universally achieve the MIMO capacity to within a constant gap \cite{ordentlich2014precoded}.
Regarding the overload MIMO, the ILC provides a degree of freedom to offset the effects caused by the dimensional flaw, which makes IF precoding exhibiting a good adaptability to overloaded MIMO. In the IF precoding, the performance relates to the selection of a proper linear precoding matrix $\mathbf{P}$. 
\textcolor{blue}{For general signal-to-noise ratio (SNR), the IF precoder design can be formulated in terms of a full-rank integer coefficient matrix $\mathbf A$ and a positive diagonal power-scaling matrix $\mathbf D$ such that $\mathbf H\mathbf P\approx c\mathbf D\mathbf A$ \cite{silva2017integer}. When the messages are defined over a fixed field $\mathbb F_p$, the coefficient matrix must also be invertible over $\mathbb F_p$ to enable message recovery.}

As a results, the research on IF precoder designing mainly focuses on the joint optimization of $\mathbf{A}$ and $\mathbf{D}$. For the special case that $K=2$, Silva \textit{et. al} \cite{silva2017integer} provide the optimal pair $(\mathbf{A},\mathbf{D})$. In high SNR. In general case, He \textit{et. al} \cite{he2018uplink} propose an iterative algorithm based on uplink-downlink duality to calculate the optimal $(\mathbf{A},\mathbf{D})$ accurately. 
Qiu \textit{et al.} \cite{qiu2024lattice} apply the Particle Swarm Optimization (PSO) algorithm \cite{LiangQSB06} to search for a coefficient--scaling pair.
Considering the requirements of low-complexity and sub-optimal optimization in general case, Venturelli \textit{et. al} \cite{venturelli2020optimization} propose an calculation-constraint-relaxed algorithm form integer-valued full-rank matrix $\mathbf{A}$ to complex-valued upper-unitriangular matrix $\mathbf{A}$.


Nevertheless, the aforementioned methods face three critical tradeoffs:
\begin{itemize}
	\item  \textbf{Optimality vs. Feasibility:} Iterative algorithms based on uplink-downlink duality may yield power allocations incompatible with practical constraints like a common shaping lattice \cite{he2018uplink}. 
	\item \textbf{Global vs. Local Search:} Heuristic methods like Particle Swarm Optimization (PSO) are prone to getting trapped in local optima and incur high computational costs \cite{trelea2003particle,qiu2024lattice}.
	\item \textbf{Performance vs. Complexity:} Relaxation-based methods simplify the problem by constraining $\mathbf{A}$ to a relaxed form (e.g., complex upper-unitriangular), but this often leads to performance loss \cite{venturelli2020optimization}.
\end{itemize}
Thus, designing an algorithm that efficiently navigates this joint optimization to strike a favorable complexity--performance tradeoff remains an open challenge. To address this challenge, we separate the scaling optimization within each coefficient-labelled region from the coefficient exploration across regions. This structure connects the geometric characterization to a low-complexity design combining projective scaling, approximate coefficient generation, and stochastic region search.
\color{blue}
The main contributions are summarized as follows.
\begin{itemize}
		\item \textbf{Geometric decomposition and NP-hardness:}
	We first distinguish the full-rank, finite-field-invertible, and unimodular coefficient domains. By restricting to unimodular coefficient matrices, we formulate the minimum-sum independent-vector problem (MSIVP) as the unimodular minimum-sum basis problem (UMSBP), which seeks a basis of the same lattice that minimizes the sum of squared basis-vector lengths. We then define a deterministic exact coefficient mapping whose preimage for each output forms a coefficient-labelled region, namely, the set of scaling vectors producing the same coefficient matrix. These regions partition the normalized scaling manifold and represent the restricted objective as a lower envelope, yielding an exact decomposition into fixed-coefficient scaling, coefficient-consistency verification, cross-region exploration, and candidate comparison. This formulation also characterizes the self-consistent region-wise candidates and the conditional local preservation of coefficient outputs. We further prove that exact UMSBP evaluation is NP-hard through a polynomial-time reduction from exact Euclidean CVP.
	﻿
	\item \textbf{Convex reformulation and projective contraction:}
	For the fixed-coefficient scaling subproblem identified by the geometric decomposition, we establish an equivalent convex formulation and derive its normalized stationarity condition. This condition reveals a unilateral matrix-balancing structure and motivates Reciprocal Approximation (RA), a normalized reciprocal fixed-point iteration for determining the scaling direction. Under the stated positivity condition and the assumption of finite projective diameter, we prove that the RA mapping is contractive in Hilbert's projective metric, thereby establishing directional convergence and an explicit iteration bound.
	﻿
	\item \textbf{Monotonicity and complexity:}
	Motivated by the NP-hard coefficient evaluation and the impracticality of exhaustive region exploration, we develop Multi-Cone Nested Stochastic Pattern Search (MCN-SPS) to approximate the preceding exact decomposition. MCN-SPS combines RA-based scaling, approximate coefficient generation, and stochastic cross-region exploration through rate-safeguarded alternating refinement. We prove that its accepted rates are nondecreasing and bounded above, and that the capped alternating refinement terminates finitely. The framework accommodates different coefficient-update routines, with LLL adopted as a low-complexity unimodular candidate generator in the practical implementation. We further derive its arithmetic complexity and demonstrate its empirical performance--complexity advantage in the considered MIMO configurations.
\end{itemize}
\color{black}

\noindent\textbf{Notation}: Matrices and column vectors are denoted by uppercase and lowercase boldface letters, respectively. 
\textcolor{blue}{For a matrix $\mathbf X=[\mathbf x_1,\ldots,\mathbf x_n]\in\mathbb R^{m\times n}$, $\mathbf x_j\in\mathbb R^{m}$ denotes its $j$-th column.}
We represent the set by uppercase squiggles letters. For any $x>0$, let $\mathrm{log}^{+}(x)\triangleq\mathrm{max}\{0,\mathrm{log}(x)\}$. Considering the multiplying for matrix $\mathbf{A}$ and $\mathbf{B}$, we denote the dot product as $\mathbf{A}\mathbf{B}$, and the Hadamard product as $\mathbf{A}\circ\mathbf{B}$. Regarding two vector $\mathbf{x}$ and $\mathbf{y}$, we define $\mathbf{x}\oslash\mathbf{y}$ the element-wise division which outputs a vector $[x_1/y_1,\cdots,x_n/y_n]^{\top}$. Considering multiplication in symbol calculation, the term $\otimes$ is represented as modulo-$p$ multiplication. The operation $\mathrm{diag}(\cdot)$ represents obtaining a vector contained the main diagonal elements of the inputted matrix.
%
%
%
%
%
%
%
%
%
%

\section{Preliminaries}
\subsection{System Model}
Regarding a single-cell system consisting of a BS with $N$ antenna and $K$ single-antenna UEs, the BS broadcasts signals to the UEs, and each UE acquires their correspond message in the received signal from BS. In each broadcasting, considering $k\in\{1,\cdots,K\}$, and time or subcarrier index $n\in\{1,\cdots,N\}$, the coded-modulated signal matrix $\mathbf{X}=[\mathbf{x}_1,\cdots,\mathbf{x}_N]$, where $\mathbf{x}_n=[x_{1,n},\cdots,x_{K,n}]^{\top}$, is modeled by the signal-level matrix $\mathbf{P}$ with the power constraint 
\begin{equation}\label{Eq. power constraint}
	\mathrm{Tr}(\mathbf{P}\mathbf{P}^{\top})\leq\rho, 
\end{equation}
where $\rho$ represents the transmit power at BS. Each precoded signals $\mathbf{s}[n]$ from the $N$ antennas of BS is given by
\begin{equation}
	\mathbf{s}[n]=\mathbf{P}[n]\mathbf{x}[n], \label{Eq. precoding}
\end{equation}
\textcolor{blue}{where $\mathbf{x}[n]=\mathbf{x}_n=[x_{1,n},\ldots,x_{K,n}]^{\top}$ denotes the coded-modulated symbol vector for the $K$ UEs at time (or subcarrier) index $n$.}

Regarding the transmission in channel, each $\mathbf{s}[n]$ is affected by \textcolor{blue}{the channel coefficient vector $\mathbf{f}_i[n]$ within a row-oriented channel coefficient matrix $\mathbf{F}[n]=[\mathbf{f}_1[n],\cdots,\mathbf{f}_K[n]]\in\mathbb{R}^{N\times K}$, }and the additive white Gaussian noise (AWGN) $z_i[n]$ with zero mean and unit variance $\sigma_{z}^2=1$, where $i$ represents the index of UE. For a real-valued model, the base-band equivalent received signal of the $i$-th UE $y_i[n]$ is calculated by 
\begin{equation}
	y_i[n]=\mathbf{f}_i^\top[n]\mathbf{s}[n]+z_i[n]. \label{Eq. received signal}
\end{equation}
\textcolor{blue}{Considering the channel coefficient matrix $\mathbf{H}[n]=\mathbf{F}^\top[n]\in\mathbb{R}^{K\times N}$}, the received signals of the $K$ UEs are
\begin{equation}
	\mathbf{y}[n]=\mathbf{H}[n]\mathbf{s}[n]+\mathbf{z}[n], 
\end{equation}
where $\mathbf{z}[n]$ denotes the AWGN vector. The transmission signal-to-noise ratio (SNR) is given by $\rho$. In this paper, we present a real-valued model, and the complex-valued model with $\check{\mathbf{H}}\in\mathbb{C}$ can be represented by a real-valued model of doubled dimension, i.e., 
\begin{equation}
	\left[\begin{matrix} \Re\{\mathbf{Y}\}\\\Im\{\mathbf{Y}\} \end{matrix}\right]
	=\left[\begin{matrix} \Re\{\check{\mathbf{H}}\}&-\Im\{\check{\mathbf{H}}\}\\\Im\{\check{\mathbf{H}}\}&\Re\{\check{\mathbf{H}}\} \end{matrix}\right]
	\left[\begin{matrix} \Re\{\mathbf{s}\}\\\Im\{\mathbf{s}\} \end{matrix}\right]
	+\left[\begin{matrix} \Re\{\mathbf{z}\}\\\Im\{\mathbf{z}\} \end{matrix}\right],
\end{equation}
as treated in \cite{hong2012reverse,silva2017integer}.


\subsection{Integer-Forcing (IF) Precoding}
\begin{figure*}[t!]
	\centering
	\includegraphics[width=.9\textwidth]{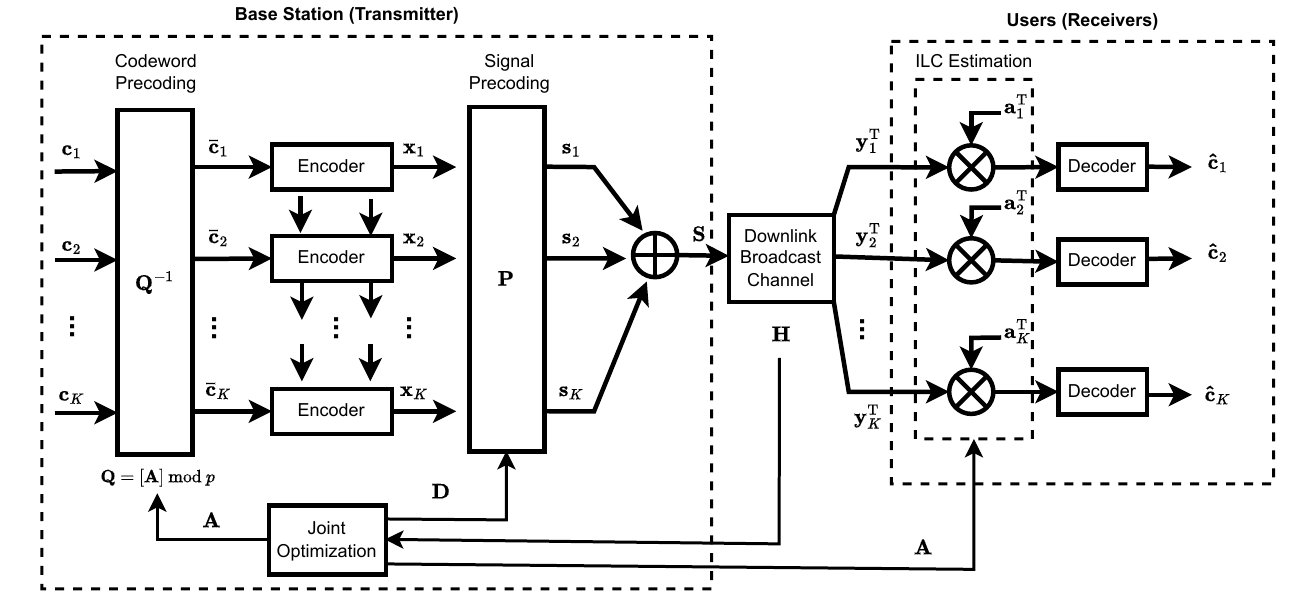}
	\caption{Block diagram of the IF precoding architecture \cite{he2018uplink,qiu2024lattice}}
	\label{Fig. IF Flowchart}
\end{figure*}

Differ to the conventional linear precoding, IF precoding employs the users' codeword to create an integer-valued effective channel matrix. By applying inverse linear transforming to correspond the ILC to the desired message of that user, IF precoder can universally achieve the MIMO capacity to within a constant gap \cite{silva2017integer,ordentlich2014precoded}. According to \cite{he2018uplink,qiu2024lattice}, the flowchart of IF precoding system is depicted in Fig. \ref{Fig. IF Flowchart}. The whole procedure consists of the transmitter, receivers, and joint optimization.

\subsubsection{Transmitter} 
Before transmitting, a full rank integer-valued matrix $\mathbf{A}\in\mathbb{Z}^{K\times K}$ and a positive diagonal matrix $\mathbf{D}=\mathrm{diag}(d_1,\cdots,d_K)$ are selected by the joint optimization. \textcolor{blue}{For convenience, we define $\mathbf{J}\triangleq\mathbf{A}^\top=[\mathbf{j}_1,\cdots,\mathbf{j}_K]$.} Moreover, $\mathbf{D}$ satisfies
\begin{equation}
	\det(\mathbf{D})=1.
	\label{Eq. D constraint}
\end{equation}
\textcolor{blue}{We distinguish the following feasible sets: 
	\begin{itemize}
		\item $\mathcal A_{\mathrm{FR}}=\left\{\mathbf A\in\mathbb Z^{K\times K}|~\det(\mathbf A)\neq0\right\}$
		\item $\mathcal A_{\mathrm{uni}}=\left\{\mathbf A\in\mathbb Z^{K\times K}|~|\det(\mathbf A)|=1\right\}$
		\item $\mathcal A_p=\left\{\mathbf A\in\mathbb Z^{K\times K}|\det(\mathbf A)\not\equiv0\pmod p\right\}$, for a fixed prime $p$
	\end{itemize}
They satisfy
\begin{equation}
\mathcal A_{\mathrm{uni}}\subseteq\mathcal A_p\subseteq\mathcal A_{\mathrm{FR}}.
\end{equation}
The original lattice-theoretic problem is defined over $\mathcal A_{\mathrm{FR}}$, while a fixed message field $\mathbb F_p$ additionally requires $\mathbf A\in\mathcal A_p$. The restricted formulation over $\mathcal A_{\mathrm{uni}}$ is introduced in Sect. III.}

At the BS, $K$ modulo-$p$ user messages $\mathbf{c}_1,\mathbf{c}_2,\cdots,\mathbf{c}_K\in\mathbb{Z}_p$ are modified by multiplying a matrix $\mathbf{Q}^{-1}$ satisfying $\mathbf{Q}^{-1}\otimes\mathbf{Q}=\mathbf{I}$, where
\begin{equation}
	\mathbf{Q}\equiv\mathbf{A}~(\mathrm{mod}~p)
	\label{Eq. CLP}
\end{equation}
over $\mathbb{Z}_p$. The outputs $\bar{\mathbf{c}}_k=\mathbf{Q}^{-1}\otimes \mathbf{c}_k$, $k\in\{1,\cdots,K\}$, are mapped to $K$ \textit{codeword-level precoded (CLP)} $p$-PAM symbols $\mathbf{x}_1,\cdots,\mathbf{x}_K$ by
\begin{align}
	\mathbf{x}_k&=\frac{1}{\gamma}\left(\bar{\mathbf{c}}_k-\frac{p-1}{2}[1,\cdots,1]^{\top}\right)
	\in\frac{1}{\gamma}\left\{\frac{1-p}{2},\cdots,\frac{p-1}{2}\right\}^K.
	\label{Eq. p-PAM}
\end{align} 

Then, a precoder $\mathbf{P}$ handles these CLP streams to the \textit{Signal-level precoded (SLP)} streams $\mathbf{s}_1,\cdots,\mathbf{s}_K$ by Eq. (\ref{Eq. precoding}), which are then transmitted via the $N$ antennas of BS. 
\textcolor{blue}{In this paper, both the diagonally-scaled exact IF (DIF) and the regularized IF (RIF) formulations are considered. Although RIF generally provides superior communication performance and is therefore preferred in practical implementations, the simpler optimization metric of DIF makes it significantly more amenable to theoretical analysis. We therefore establish the geometric framework based on the DIF formulation and subsequently demonstrate that the same framework extends naturally to RIF, thereby illustrating the generality of the proposed theory.}
According to \cite{silva2017integer},  $\mathbf{P}$ can be obtained by
\begin{align}
	&\mathbf{P}=\frac{1}{\eta}\mathbf{H}^{\top}(\mathbf{H}\mathbf{H}^{\top})^{-1}\mathbf{D}\mathbf{A} &\text{~(DIF)}\label{Eq. DIF}\\
	&\mathbf{P}=\frac{1}{\eta}\mathbf{H}^{\top}(\frac{K}{\rho}\mathbf{I}+\mathbf{H}\mathbf{H}^{\top})^{-1}\mathbf{D}\mathbf{A} &\text{~(RIF)}\label{Eq. RIF}
\end{align}
In Eq. (\ref{Eq. DIF}) and Eq. (\ref{Eq. RIF}), each entry $d_k,~k\in\{1,\cdots,K\}$ denotes the power allocated to the $k$-th CLP stream, and $\eta$ is a scalar to ensure that the power constraint Eq. (\ref{Eq. power constraint}) is met. 
\textcolor{blue}{While the subsequent theoretical development is first presented for DIF for clarity of exposition, the proposed framework applies directly to the RIF formulation, which is adopted in scenarios such as overloaded MIMO.}
Regarding the full column rank channel matrix $\mathbf{H}$ in overload MIMO, RIF is suitable since $\mathbf{H}\mathbf{H}^{\top}$ is not a full-rank matrix but $\frac{K}{\rho}\mathbf{I}+\mathbf{H}\mathbf{H}^{\top}$ is.

\subsubsection{Receivers} According to Eq. (\ref{Eq. p-PAM}), the received signal of UE $k$ in Eq. (\ref{Eq. received signal}) can be written as
\begin{align}
	y_k&=\mathbf{f}_k^\top\mathbf s_k+z_k
	=\mathbf{f}_k^\top\mathbf P\mathbf x_k+z_k
	=\frac{1}{\gamma}\tilde{\mathbf{h}}_k^{\top}\bar{\mathbf{c}}_k+z_k
	-\underbrace{\frac{p-1}{2\gamma}\sum_{l=1}^{K}\tilde{h}_{k,l}}
	_{\text{not related to the signal}},
	\label{Eq. receiver IF signal}
\end{align}
\textcolor{blue}{where $\tilde{\mathbf h}_k\triangleq\mathbf P^{\top}\mathbf{f}_i\in\mathbb R^{K}$, and hence $\tilde{\mathbf h}_k^{\top}=\mathbf{f}_i^\top\mathbf P=[\tilde h_{k,1},\ldots,\tilde h_{k,K}]$.}
Since $\bar{\mathbf{c}}_k$ is an integer-valued vector, $\tilde{\mathbf{h}}_k^{\top}\bar{\mathbf{c}}_k$ can be regarded as the $k$-th entry of lattice point with basis $[\tilde{\mathbf{h}}_1^{\top},\cdots,\tilde{\mathbf{h}}_K^{\top}]^{\top}$, 
\textcolor{blue}{and $y_k$ as a noisy observation of the \(k\)-th component of the lattice point after a known signal-independent offset.}

At the receivers, receiver $k$ computes the ILC of $\bar{\mathbf{c}}_k$, written as
\begin{equation}
	\hat{s}_{k,n}=\mathbf{j}_i^\top\otimes\bar{\mathbf{c}}_k,~k\in\{1,\cdots,K\},~n\in\{1,\cdots,N\},
	\label{Eq. ILC solve}
\end{equation}
based on the noisy observation of the lattice points in Eq. (\ref{Eq. receiver IF signal}). According to Eq. (\ref{Eq. CLP}), Eq. (\ref{Eq. ILC solve}) can be rewritten as
\begin{align}\label{Eq. receiver inversion}
	\hat{s}_{k,n}=\mathbf{j}_i^\top\otimes\bar{\mathbf{c}}_k
	=\mathbf{j}_i^\top\otimes\mathbf{A}^{-1}\otimes\mathbf{c}_k=\hat{\mathbf{c}}_k.
\end{align}
Then, the computed steam of UE-$k$ $\hat{\mathbf{s}}_k$ is shifted into the estimated message $\hat{\mathbf{c}}_1,\cdots,\hat{\mathbf{c}}_K$.
\textcolor{blue}{The inverse $\mathbf A^{-1}$ in Eq.~(\ref{Eq. receiver inversion}) is taken over $\mathbb F_p$. Therefore, a practical implementation with fixed $p$ additionally requires $\operatorname{rank}_{\mathbb F_p}(\mathbf A\bmod p)=K$, or equivalently, $\det(\mathbf A)\not\equiv0\pmod p$. Real full rank alone does not guarantee invertibility over $\mathbb F_p$.}

\subsubsection{Joint Optimization} The achievable rate of IF precoding can be obtained simply by
\begin{equation}
	R_i\left(\mathbf{A}\right)\triangleq\log_2^{+}\left(
	\frac{1}{\mathbf{j}_i^\top\left(\mathbf{I}-\frac{\rho}{\rho\|\tilde{\mathbf{h}}_i\|^2+1}\tilde{\mathbf h}_i\tilde{\mathbf h}_i^{\top}\right)\mathbf{j}_i}
	\right),
\end{equation}
where $\tilde{\mathbf h}_i=\mathbf P^{\top}\breve{\mathbf h}_i^{\top}$ and \textcolor{blue}{we use the function $\log^{+}(x)\triangleq\max\{0,\log(x)\}$ with $x>0$ to ensure that the rate remains non-negative.}

For the high SNR regime \cite{qiu2024lattice}, IF precoding is optimal and can achieve a sum rate given by
\begin{align}
	R(\mathbf{A},\mathbf{D})&=\frac{1}{2}\sum_{i=1}^{K}\log_2^{+}\left(\frac{d_i^2\rho}{\mathrm{Tr}\left(\mathbf{A}^{\top}\mathbf{D}^{\top}\mathbf{M}\mathbf{D}\mathbf{A}\right)}\right)
	=\frac{K}{2}\log_2^{+}\left(\frac{\rho\left(\det(\mathbf{D})\right)^{\frac{2}{K}}}{\mathrm{Tr}\left(\mathbf{A}^{\top}\mathbf{D}^{\top}\mathbf{M}\mathbf{D}\mathbf{A}\right)}\right),
	\label{Eq. high SNR sum rate}
\end{align}
where
\begin{align}
	&\mathbf{M}=\left(\mathbf{H}\mathbf{H}^{\top}\right)^{-1} &\text{~(DIF)}\label{Eq. DIF M}\\
	&\mathbf{M}=\left(\frac{K}{\rho}\mathbf{I}+\mathbf{H}\mathbf{H}^{\top}\right)^{-1} &\text{~(RIF)}\label{Eq. RIF M}
\end{align}
\textcolor{blue}{Since $\det(\mathbf D)=1$ and $\log_2^{+}(\cdot)$ is monotone, maximizing Eq.~(\ref{Eq. high SNR sum rate}) is equivalent to}
\textcolor{blue}{\begin{align}\label{Eq. optimization problem}
	\underset{\mathbf{A},\mathbf{D}}{\arg\min}~~~&\mathrm{Tr}\left(\mathbf{A}^{\top}\mathbf{D}^{\top}\mathbf{M}\mathbf{D}\mathbf{A}\right),\\
	\text{subject to}~~~&\mathbf{A}\in\mathcal{A}_{\mathrm{FR}}
	,~\mathbf{D}=\mathrm{diag}(d_1,\cdots,d_K)\in\mathbb{R}_{+}^{K\times K},~\det(\mathbf{D})=1. \notag
\end{align}}
\textcolor{blue}{When the message field $\mathbb F_p$ is fixed, $\mathcal A_{\mathrm{FR}}$ is replaced by $\mathcal A_p$.}
Considering a fixed $\mathbf{D}$, the optimization of $\mathbf{A}$ is equivalent to minimize
\begin{equation}
	\mathrm{Tr}\left(\mathbf{A}^{\top}\mathbf{D}^{\top}\mathbf{M}\mathbf{D}\mathbf{A}\right)
	=\sum_{l=1}^{K}\|\mathbf{M}^{\frac{1}{2}}\mathbf{D}\mathbf{a}_l\|^2,
	\label{Eq. A optimization}
\end{equation}
where $\mathbf{a}_l$ denotes the $l$-th column of $\mathbf{A}$. 
\textcolor{blue}{Regarding $\mathbf{M}^{\frac{1}{2}}\mathbf{D}$ as a lattice basis, to minimize Eq. (\ref{Eq. A optimization}) is a minimum-sum independent-vector problem (MSIVP) introduced in next subsection. }

\subsection{Lattice Basis Reduction}
This paper will use some concepts from lattices to simplify the algorithmic description of the algorithm and more elegant analysis. We therefore review some basic concepts of lattices and lattice reduction. An $M$-dimensional lattice in $\mathbb{R}^M$ is defined as  
\begin{equation}
	\Lambda=\left\lbrace \mathbf{Gz}|\mathbf{z}\in \mathbb{Z}^M\right\rbrace,
\end{equation}
where the full-rank matrix $\mathbf{G}\in\mathbb{R}^{M\times M}$ represents the generator matrix of $\Lambda$. In a lattice $\Lambda$, SIVP is a fundamental problem to describe the structure of lattice, which is defined in Definition \ref{Def. SIVP}.
\begin{definition}[SIVP \cite{nguyen2010lll}]\label{Def. SIVP}
	Given a lattice $\Lambda\in\mathbb{R}^M$, find $M$ linearly independent vectors $\mathbf{v}_1,\cdots,\mathbf{v}_M\in \Lambda$, such that
	\begin{equation}
		\|\mathbf{v}_m\|\leq\lambda_M(\Lambda),~m\in\{1,\cdots,M\}
	\end{equation}
\end{definition}
In Definition \ref{Def. SIVP}, $\lambda_M(\Lambda)$ represents the $M$-th successive minima which is given by Definition \ref{Def. Successive Minima}. 
\begin{definition}[Successive Minima \cite{nguyen2010lll}]\label{Def. Successive Minima}
	Considering a lattice $\Lambda\in\mathbb{R}^M$ and $\forall i\in\{1,\cdots,\mathrm{dim}(\Lambda)\}$, the $i$-th minimum $\lambda_i(\Lambda)$ is defined as the minimum of $\max_{1\leq j\leq i}\|\mathbf{v}_j\|$ over all $i$ linearly independent lattice vectors $\mathbf{v}_1,\cdots,\mathbf{v}_i\in \Lambda$
\end{definition}

\textcolor{blue}{LLL is a polynomial-time basis-reduction algorithm rather than an exact MSIVP solver. Given $\mathbf G$, it returns an LLL-reduced basis $\mathbf G\mathbf U_{\rm LLL}$, where $\mathbf U_{\rm LLL}$ is unimodular. Hence, $\mathbf A_{\rm LLL}=\mathbf U_{\rm LLL}$ provides a low-complexity feasible candidate, but is not generally optimal over $\mathcal A_{\rm FR}$ or $\mathcal A_p$.}
The corresponded metric of LLL algorithm, referred to as LLL-reduced basis, is defined in Definition \ref{Def. LLL-reduced basis}. 
\textcolor{blue}{For $\Gamma_i=\{1,\ldots,i-1\}$, let $\mathbf G_{\Gamma_i}=[\mathbf g_1,\ldots,\mathbf g_{i-1}]$. The operator $\pi_{\mathbf G_{\Gamma_i}}^{\perp}(\mathbf x)$ denotes the orthogonal projection of $\mathbf x$ onto the orthogonal complement of $\operatorname{span}(\mathbf G_{\Gamma_i})$.}
\color{blue}
\begin{definition}[LLL reduced \cite{Lyu2017boost}]\label{Def. LLL-reduced basis}
	Let $\mathbf G=[\mathbf g_1,\ldots,\mathbf g_M]$ and let $\mathbf R$ be the upper-triangular factor in the QR decomposition of $\mathbf G$, with entries $r_{i,j}$. The basis $\mathbf G$ is LLL-reduced if
	\begin{itemize}
		\item[i)]
		$\left|r_{i,j}/r_{i,i}\right|\leq\frac{1}{2}$,
		$1\leq i<j\leq M$.
		\item[ii)]
		$\delta
		\left\|
		\pi_{\mathbf G_{\Gamma_i}}^{\perp}(\mathbf g_i)
		\right\|^2
		\leq
		\left\|
		\pi_{\mathbf G_{\Gamma_i}}^{\perp}(\mathbf g_{i+1})
		\right\|^2$,
		$1\leq i\leq M-1$.
	\end{itemize}
\end{definition}
\color{black}
In Definition \ref{Def. LLL-reduced basis}, $\delta\in\left(\frac{1}{4},1\right]$ is called the Lov$\acute{a}$sz constant, and an equivalently rewritten form of Lov$\acute{a}$sz conditions is
\begin{equation}\label{Eq. equivalently lovasz condition}
	\|\mathbf{g}_{i+1}^{*}\|^2\geq\left(\delta-\mu^2_{i+1}\right)\|\mathbf{g}_{i}^{*}\|^2.
\end{equation}
If $\mathbf{G}$ is LLL-reduced, it satisfies \cite{LLL}
\begin{equation}
	\|\mathbf{g}_i\|\leq\beta^{N-1}\lambda_i\left(\Lambda_{\mathbf{G}}\right),~i\in\{1,\cdots,N\},
\end{equation}
in which 
\begin{equation}
	\beta=\frac{1}{\sqrt{\delta-\frac{1}{4}}}\in\left(\frac{2}{\sqrt{3}},+\infty\right).
\end{equation}

\begin{algorithm}[t]
	\caption{The LLL algorithm} 
	\label{Alg. LLL}

	\hspace*{0.02in} {\bf Input:} 
	Lattice basis matrix $\mathbf{G}$, Lov$\acute{a}$sz constant $\delta$\\
	\hspace*{0.02in} {\bf Output:} 
	LLL-reduced basis matrix $\mathbf{G}_{\text{LLL}}$, Transforming unimodular matrix $\mathbf{U}_{\text{LLL}}$
	\begin{algorithmic}[1]
		\State $[\mathbf{Q},\mathbf{R}]\gets\text{qr}(\mathbf{G})$ \Comment{The QR Decomposition of $\mathbf{G}$}
		\State $[M,N]\gets\text{size}(\mathbf{G})$
		\State $i\gets 2$
		\While{$i\leq N$}
		\For{$k=i-1:-1:1$} 
		\State $\mu_{k,i}=\lfloor{r_{k,i}}/{r_{k,k}}\rceil$\Comment{$r_{i,j}$ is the $(i,j)$ element of $\mathbf{R}$}
		\If{$\mu_{k,i}\neq0$}\Comment{Operation i}
		\State $\mathbf{r}_{i}\gets\mathbf{r}_{i}-\mu_{k,i}\mathbf{r}_{k}$\Comment{$\mathbf{r}_{i}$ is the $i$-th column of $\mathbf{R}$}
		\EndIf
		\EndFor
		\If{$\delta|r_{i-1,i-1}|^2>|r_{i,i}|^2+|r_{i-1,i}|^2$}
		\State $\alpha=r_{i-1,i}/\sqrt{r_{i-1,i}^2+r_{i,i}^2}$
		\State $\beta=r_{i,i}/\sqrt{r_{i-1,i}^2+r_{i,i}^2}$
		\State Swap $\mathbf{r}_{i-1}$ and $\mathbf{r}_{i}$\Comment{Operation ii}
		\State $\mathbf{V}\gets\left[\begin{matrix}
			\alpha &\beta\\
			-\beta &\alpha
		\end{matrix}\right]$
		\State \textcolor{blue}{$\left[\begin{matrix}
			 \mathbf{e}_{i-1}^\top\\
			 \mathbf{e}_{i}^\top
		\end{matrix}\right]\mathbf{R}\gets\mathbf{V}\left[\begin{matrix}
		\mathbf{e}_{i-1}^\top\\
		\mathbf{e}_{i}^\top
		\end{matrix}\right]\mathbf{R}$}
		\State $\left[\mathbf{q}_{i-1},\mathbf{q}_{i}\right]\gets\left[\mathbf{q}_{i-1},\mathbf{q}_{i}\right]\mathbf{V}^{\top}$\Comment{$\mathbf{q}_{i}$ is the $i$-th column of $\mathbf{Q}$}
		\State $i\gets\max(i-1,2)$
		\Else
		\State $i\gets i+1$
		\EndIf
		\EndWhile
		\State $\mathbf{G}_{\text{LLL}}\gets\mathbf{Q}\mathbf{R}$
		\State $\mathbf{U}_{\text{LLL}}\gets\lfloor\mathbf{G}^{-1}\mathbf{G}_{\text{LLL}}\rceil$ \Comment{$\lfloor\cdot\rceil$ round for all elements}
	\end{algorithmic}
\end{algorithm}

In the LLL algorithm, summarized in Alg. \ref{Alg. LLL}, the operation i and ii can be regarded as two transformation matrices $\mathbf{T}^{(i)}_{k,i}$ and $\mathbf{T}^{(ii)}_{i}$ respectively satisfy
\begin{align}
	\mathbf{T}^{(i)}_{k,i}=\left[
	\begin{matrix}
		1 &-\mu_{k,i}\\
		0 &1
	\end{matrix}
	\right],
	\mathbf{T}^{(ii)}_{i}=\left[
	\begin{matrix}
		0 &1\\
		1 &0
	\end{matrix}
	\right],
\end{align}
where $\mathbf{T}^{(i)}_{k,i}$ is implemented on a pair of vectors $[\mathbf{r}_{k},\mathbf{r}_{i}]$ and $\mathbf{T}^{(ii)}_{i}$ on $[\mathbf{r}_{i-1},\mathbf{r}_{i}]$ . In the LLL algorithm processing, the transforming matrix between the original and LLL-reduced basis is a unimodular matrix which is derived by multiplying all $\mathbf{T}^{(i)}_{k,i}$ and $\mathbf{T}^{(ii)}_{i}$ in the order in which they are applied in the process \cite{dirk2011latticereduction}.
\textcolor{blue}{
	\begin{remark}
		Let $\mathbf{G}$ be a lattice generator and $\mathbf A=[\mathbf a_1,\ldots,\mathbf a_K]\in\mathbb Z^{K\times K}$ satisfy $\det(\mathbf A)\neq0$. Standard SIVP minimizes
		\begin{equation}
			\max_{1\le j\le K}\|\mathbf G\mathbf a_j\|_2,
		\end{equation}
		whereas MSIVP in Eq. (\ref{Eq. A optimization}) minimizes
		\begin{equation}
			\sum_{j=1}^{K}\|\mathbf G\mathbf a_j\|_2^2 =\|\mathbf G\mathbf A\|_F^2.
		\end{equation}
		Thus, standard SIVP and MSIVP share the same full-rank feasible set but have different objectives. Neither problem requires $\mathbf A$ to be unimodular. Restricting Eq.~(\ref{Eq. A optimization}) to $\mathcal A_{\mathrm{uni}}$ yields the unimodular minimum-sum basis problem (UMSBP)
		\begin{equation}
			\label{Eq. UMSBP definition}
			\underset{\mathbf A\in\mathcal A_{\mathrm{uni}}}{\operatorname{minimize}}
			\quad
			\|\mathbf G\mathbf A\|_F^2
			=
			\sum_{j=1}^{K}\|\mathbf G\mathbf a_j\|_2^2.
		\end{equation}
		Here, $\mathbf G\mathbf A$ is constrained to be a basis of the same lattice generated by $\mathbf G$.
	\end{remark}
}

\subsection{The Contraction Mapping Approach to Perron-Frobenius Theory}
In this paper, we employ the Perron-Frobenius Theory in the Hilbert's projective metric to guide the algorithm design. Hereby, we review some about the contraction mapping approach \cite{birkhoff1957extensions,kohlberg1982contraction,carroll2004birkhoff} and nonlinear Perron-Frobenius theory \cite{lemmens2012nonlinear}. 

In a complete metric space $(\mathcal{X},d)$, a linear transformation $\mathbf{T}$ is a contraction if it satisfies
\begin{equation}
	\exists \gamma\in\left(0,1\right), d\left(\mathbf{T}\mathbf{x},\mathbf{T}\mathbf{y}\right)\leq\gamma d\left(\mathbf{x},\mathbf{y}\right),\forall\mathbf{x},\mathbf{y}\in\mathcal{X}. 
\end{equation}
For a contraction $\mathbf{T}$, there exists an $\mathbf{x}_0\in\mathcal{X}$ such that $\mathbf{T}^q\mathbf{x}\mapsto\mathbf{x}_0$ for all $\mathbf{x}\in\mathcal{X}$, where $\mathbf{x}_0$ is referred to as the fixed point for $\mathbf{T}$ in $(\mathcal{X},d)$. In direction, the well-known Perron Theorem \cite{kohlberg1982contraction} tells that if $\mathbf{T}$ is a positive linear transformation on $\mathbb{R}^N$, then it has
\begin{equation}
	\exists \mathbf{x}_0,\mathbf{x}>0,~\text{such that (s.t.)},~\frac{\mathbf{T}^{(q)}\mathbf{x}}{\|\mathbf{T}^{(q)}\mathbf{x}\|}\rightarrow\frac{\mathbf{x}_0}{\|\mathbf{x}_0\|},~q\rightarrow\infty.
\end{equation}
In \cite{birkhoff1957extensions,hopf1963inequality}, a Hilbert metric $d_H(\cdot)$ in non-Euclidean geometry is defined as
\begin{equation}
	d_H\left(\mathbf{x},\mathbf{y}\right)=\log\left(\frac{\max\left(\mathbf{x}\oslash\mathbf{y}\right)}{\min\left(\mathbf{x}\oslash\mathbf{y}\right)}\right).
\end{equation}
It has the following properties \cite{carroll2004birkhoff}.
\begin{itemize}
	\item[i)] $d_H\left(\mathbf{x},\mathbf{y}\right)\geq0$, and $d_H\left(\mathbf{x},\mathbf{y}\right)=0$ if and only if $\mathbf{x}=c\mathbf{y}$, $c>0$.
	\item[ii)] $d_H\left(\mathbf{x},\mathbf{y}\right)=d_H\left(\mathbf{y},\mathbf{x}\right)$.
	\item[iii)] $d_H\left(\mathbf{x},\mathbf{z}\right)\leq d_H\left(\mathbf{x},\mathbf{y}\right)d_H\left(\mathbf{y},\mathbf{z}\right)$.
	\item[iv)] $d_H\left(\mathbf{x},\mathbf{z}\right)\leq d_H\left(\mathbf{x},\mathbf{y}\right)+d_H\left(\mathbf{y},\mathbf{z}\right)$.
	\item[v)] $\forall a,b>0$, $d_H\left(a\mathbf{x},b\mathbf{y}\right)=d_H\left(\mathbf{x},\mathbf{y}\right)$.
\end{itemize}
The corresponding contraction result in Hilbert's projective metric is the Birkhoff--Hopf theorem, which is denoted by
\begin{align}
	&d_H\left(\mathbf{T}\mathbf{x},\mathbf{T}\mathbf{y}\right)\leq\gamma d_H\left(\mathbf{x},\mathbf{y}\right),\forall\mathbf{x},\mathbf{y}\in\mathcal{X},
\end{align}
\textcolor{blue}{where $\mathbf T$ is a positive linear operator with finite projective diameter and $\gamma\in(0,1)$.}
The contraction ratio $\kappa(\mathbf{T})$ of a linear operator $\mathbf{T}$ is defined in \cite{lemmens2012nonlinear} as
\begin{align}
	&\kappa(\mathbf{T})\triangleq
	\inf\left\{\gamma\geq0|d_H\left(\mathbf{T}\mathbf{x},\mathbf{T}\mathbf{y}\right)\leq\gamma d_H\left(\mathbf{x},\mathbf{y}\right),\forall\mathbf{x},\mathbf{y}\in\mathcal{X}\right\}.
\end{align}
In \cite{birkhoff1957extensions}, it is proved that
\begin{equation}
	\kappa(\mathbf{T})=\tanh\left(\frac{\Delta(\mathbf{T})}{4}\right),
\end{equation}
where the projective diameter $\Delta(\mathbf{T})$ of $\mathbf{T}$ is defined by
\begin{equation}
	\Delta(\mathbf{T})\triangleq\sup\left\{d_H\left(\mathbf{T}\mathbf{x},\mathbf{T}\mathbf{y}\right)|\mathbf{x},\mathbf{y}\in\mathcal{X}\right\}.
\end{equation}

\section{Geometric Structure under the Unimodular Constraint}

\textcolor{blue}{Taking Eq.~(\ref{Eq. optimization problem}), defined over $\mathcal A_{\mathrm{FR}}$ or over $\mathcal A_p$ for a fixed message field, as the starting point, we subsequently study its restriction to $\mathcal A_{\mathrm{uni}}$. This restriction remains feasible for the original problem but may exclude better full-rank solutions. Unless otherwise stated, the subsequent analysis and algorithm design refer to this restricted formulation, where the $\mathbf d$-space geometry organizes coefficient-labelled regions, links the fixed-$\mathbf A$ stationarity condition to a normalized directional solution, and identifies self-consistent candidates for the cross-region search.}
Since $\mathbf{D}=\mathrm{diag}(d_1,\cdots,d_K)$ is a positive diagonal matrix with $\det(\mathbf{D})=1$, the set of all possible $\mathbf{D}$ is the positive diagonal subgroup of special linear group $\mathrm{SLD}_{K}^{+}(\mathbb{R})$. 
For convenience, we represent the main diagonal of $\mathbf{D}$ by the vector $\mathbf{d}$, represented as
\begin{equation}\label{Eq. d definition}
	\mathbf{d}
	=\left[d_1,\cdots,d_K\right]^{\top},
\end{equation}
Since the elements $\{d_k\}_{k=1}^K$ are all positive, following the constraint in Eq. (\ref{Eq. D constraint}), $\mathbf{d}$ satisfies $\prod_{k=1}^{K}d_k=1$, or $\sum_{k=1}^{K}\log(d_k)=0$. 
\textcolor{blue}{Regarding $\{d_k\}_{k=1}^{K}$ as coordinates in the positive orthant, the feasible power-scaling matrices are represented by the normalized $(K-1)$-dimensional manifold}
\begin{equation}
	\varOmega=\left\{\mathbf{d}\in\mathbb R_{+}^{K}~|\prod_{k=1}^{K}d_k=1,d_k>0,k\in{1,\cdots,K}\right\}.
\end{equation}

\subsection{Coefficient-Labelled Geometric Reformulation}\label{Sec. general model}
\textcolor{blue}{For the unimodularly restricted problem, $\mathbf A(\mathbf d)\in\mathcal A_{\mathrm{uni}}$ denotes the coefficient matrix selected for a given $\mathbf d$.}
To correspond $\mathcal{D}$, the expression in Eq. (\ref{Eq. optimization problem}) can be rewritten as 
\begin{equation}
	\mathrm{Tr}\left(\mathbf{A}(\mathbf{d})^{\top}\mathbf{D}^{\top}\mathbf{M}\mathbf{D}\mathbf{A}(\mathbf{d})\right)
	=\mathbf{d}^{\top}\left(\left(\mathbf{A}(\mathbf{d})\mathbf{A}(\mathbf{d})^{\top}\right)\circ\mathbf{M}\right)\mathbf{d}.
	\label{Eq. generalized optimization}
\end{equation}
Thus, the unimodularly restricted form of (\ref{Eq. optimization problem}) can be written as
\begin{align}\label{Eq. new optimization problem}
	\underset{\mathbf{d}}{\arg\min}~~~&\mathbf{d}^{\top}\left(\left(\mathbf{A}(\mathbf{d})\mathbf{A}(\mathbf{d})^{\top}\right)\circ\mathbf{M}\right)\mathbf{d},\\
	\text{subject to}~~~&\mathbf{d}\in\varOmega. \notag
\end{align}

\textcolor{blue}{The coefficient dependence in Eq.~(\ref{Eq. new optimization problem}) is described by the exact restricted mapping}
\begin{equation} \label{Eq. d2A mapping}
	\mathbf{A}(\mathbf d)=\left(\mathbf{M}^{\frac{1}{2}}\mathrm{diag}(\mathbf{d})\right)^{-1}\Psi\left(\mathbf{M}^{\frac{1}{2}}\mathrm{diag}(\mathbf{d})\right),
\end{equation}
\textcolor{blue}{where $\Psi(\cdot)$ denotes an ideal exact single-output oracle for the UMSBP in Eq.~(\ref{Eq. UMSBP definition}). A deterministic tie-breaking rule is applied when multiple exact minimizers exist, so $\mathbf A(\mathbf d)$ is single-valued.}

\textcolor{blue}{
\begin{remark}
	In the practical implementation, LLL returns $\mathbf G'_{\mathrm{LLL}}(\mathbf d)=\mathbf M^{\frac{1}{2}}\operatorname{diag}(\mathbf d)\mathbf U_{\mathrm{LLL}}(\mathbf d)$, where $\mathbf U_{\mathrm{LLL}}(\mathbf d)$ is unimodular. Hence, the practical coefficient update is
	\begin{equation}
		\label{Eq. d2A mapping LLL}
		\mathbf A_{\mathrm{LLL}}(\mathbf d)
		=
		\left[\mathbf M^{\frac{1}{2}}\operatorname{diag}(\mathbf d)\right]^{-1}
		\mathbf G'_{\mathrm{LLL}}(\mathbf d)
		=
		\mathbf U_{\mathrm{LLL}}(\mathbf d)
		\in\mathcal A_{\mathrm{uni}}.
	\end{equation}
	This LLL output is a practical candidate and need not equal the exact restricted minimizer $\mathbf A(\mathbf d)$. The theoretical regions below are defined by Eq.~(\ref{Eq. d2A mapping}).
\end{remark}
For a fixed $\mathbf M$, each $\mathbf d\in\varOmega$ uniquely determines $\mathbf D=\operatorname{diag}(\mathbf d)$ and the corresponding scaled lattice generator  $\mathbf M^{\frac{1}{2}}\operatorname{diag}(\mathbf d)$. The exact single-output mapping in Eq.~(\ref{Eq. d2A mapping}) then assigns one coefficient matrix $\mathbf A(\mathbf d)$ to each $\mathbf d$. For each $\mathbf A\in\mathcal A_{\mathrm{uni}}$, define its coefficient-dependent region as
\begin{equation}
	\label{Eq. coefficient dependent region}
	\mathcal P(\mathbf A)=\left\{\mathbf d\in\varOmega~|\mathbf A(\mathbf d)=\mathbf A\right\}.
\end{equation}
Because the mapping $\mathbf A(\mathbf d)$ is single-valued, its nonempty preimages $\mathcal P(\mathbf A)$ form a disjoint partition of $\varOmega$. Within each $\mathcal P(\mathbf A)$, the coefficient matrix is fixed and the restricted joint problem reduces to a continuous scaling problem. This representation therefore separates fixed-$\mathbf A$ scaling optimization from cross-region coefficient search.
To characterize the local stability of the exact coefficient output under a diagonal perturbation of the lattice generator, we use the spectral norm
\begin{equation}
	\|\mathbf D\|_2\triangleq\sup_{\mathbf x\neq\mathbf 0_K}\frac{\|\mathbf D\mathbf x\|_2}{\|\mathbf x\|_2}.
\end{equation}
The following theorem relates the exact coefficient outputs of two lattice generators connected by a determinant-one diagonal perturbation and characterizes their local preservation within a coefficient-dependent region.}


\begin{theorem}\label{The. similar lattice reduction}
	Considering a non-singular generator matrix $\mathbf{G}_{A}$ of lattice $\Lambda_A$ and a positive diagonal matrix $\mathbf{D}_{A}$ with constraint $\det(\mathbf{D}_{A})=1$, the product $\mathbf{G}_{B}=\mathbf{G}_{A}\mathbf{D}_{A}$ is the generator matrix of lattice $\Lambda_B$, and the solution set of $\Psi(\mathbf{G}_{A})\in\mathcal A_{\mathrm{uni}}$ and $\Psi(\mathbf{G}_{B})\in\mathcal A_{\mathrm{uni}}$ are denoted as $\mathcal{S}_A$ and $\mathcal{S}_B$ respectively. When $\mathbf{D}_{A}$ is similar to $\mathbf{I}$, i.e., for a minimum value $\epsilon>0$, $\|\mathbf{D}_{A}-\mathbf{I}\|_{2}$ satisfies
	\begin{equation}\label{Eq. The. 1 condition 1}
		\|\mathbf{D}_{A}-\mathbf{I}\|_{2}<\epsilon,
	\end{equation}
	and
	\begin{equation}\label{Eq. The. 1 condition 2}
		|\lambda_K(\Lambda_A)-\lambda_K(\Lambda_B)|<\epsilon,
	\end{equation}
	so we have
	\begin{equation}
		\forall \mathbf{R}_{B}=\mathbf{G}_{B}\mathbf{U}\in \mathcal{S}_B,~\exists\mathbf{R}_{A}=\mathbf{G}_{A}\mathbf{V}\in \mathcal{S}_A,~\text{s.t.},~\mathbf{U}=\mathbf{V}.
	\end{equation}
\end{theorem}
\begin{proof}
	See Appendix \ref{App. similar lattice reduction proof}.
\end{proof}

\textcolor{blue}{
	\begin{remark}
		The positive $\epsilon$ in Theorem~\ref{The. similar lattice reduction} is understood as a local stability radius within a coefficient-dependent region. At a coefficient-selection boundary, such a positive radius may not exist and an arbitrarily small perturbation may produce a different coefficient matrix, so the local-preservation conclusion is not asserted. Nevertheless, the deterministic tie-breaking rule in Eq.~(\ref{Eq. d2A mapping}) assigns the boundary point itself to exactly one $\mathcal P(\mathbf A)$, and the coefficient-dependent regions remain a well-defined partition of $\varOmega$.
	\end{remark}
}
\textcolor{blue}{The following two-dimensional example illustrates how the practical LLL mapping changes within the same unimodular coefficient domain.}
\begin{itemize}
	\item \textit{Size reduction condition: }Considering the QR decomposition
	\begin{equation}
		\mathbf{M}_2^{'}=\mathbf{Q}_2\mathbf{R}_2=\mathbf{Q}_2
		\left[
		\begin{matrix}
				r_{1,1}&r_{1,2}\\
				0&r_{2,2}
		\end{matrix}
		\right].
	\end{equation}
	According to Alg. \ref{Alg. LLL}, the coefficient of size reduction $\mu=\lfloor\frac{r_{1,2}}{r_{1,1}}\rceil$ when $\mathbf{D}_2=\mathbf{I}_2$. For a changed $\mathbf{D}_2=\mathrm{diag}(d_1,d_2)$, the changed coefficient of size reduction $\mu^{'}$ is given by
	\begin{equation}
		\mu^{'}=\left\lfloor\frac{d_2 r_{1,2}}{d_1 r_{1,1}}\right\rceil
		\overset{d_1=\frac{1}{d_2}}{=}\left\lfloor\mu+\left\lfloor(d_2^2-1)\frac{r_{1,2}}{r_{1,1}}\right\rceil\right\rceil.
	\end{equation}
	When $\left\lfloor(d_2^2-1)\frac{r_{1,2}}{r_{1,1}}\right\rceil\neq0$, $\mu^{'}\neq\mu$, such that $\mathbf{A}_2$ has been changed.
	\item \textit{Lov$\acute{a}$sz condition: }According to Eq. (\ref{Eq. equivalently lovasz condition}), the swapping between two vector in $\mathbf{M}_2^{'}$ is related to $\mu$. Since $\mu$ has been changed under different $\mathbf{D}_2$ in size reduction, the number of iterations becomes unknown for the algorithm process in practical, which triggers a sudden changes on $\Psi_{\text{LLL}}(\mathbf{M}_2^{'}\mathbf{D}_2)$ and leads to a different $\mathbf{A}_2$.
\end{itemize}

\textcolor{blue}{This example illustrates how the practical LLL coefficient mapping may change with $\mathbf d$. It is not used to characterize the exact preimages $\mathcal P(\mathbf A)$.}
\textcolor{blue}{The coefficient-labelled geometric decomposition further separates the exact joint problem into four subproblems and reveals its computational bottleneck.}
\textcolor{blue}{
\begin{theorem}
	\label{Pro. NP-hard}
	Under the coefficient domain $\mathcal A_{\mathrm{uni}}$, solving Eq.~(\ref{Eq. new optimization problem}) exactly through the coefficient mapping in Eq.~(\ref{Eq. d2A mapping}) involves the following four subproblems:
	\begin{itemize}
		\item[SP1] For a fixed coefficient matrix
		$\mathbf A_i\in\mathcal A_{\mathrm{uni}}$, determine its optimal scaling vector over $\varOmega$ without imposing $\mathbf d\in\mathcal P(\mathbf A_i)$.
		\item[SP2] Evaluate the exact coefficient mapping at the solution obtained in SP1 and determine whether $\mathbf A(\mathbf d)=\mathbf A_i$, equivalently, $\mathbf d\in\mathcal P(\mathbf A_i)$.
		\item[SP3] Identify the coefficient-dependent regions whose fixed-coefficient solutions satisfy the consistency condition in SP2.
		\item[SP4] Compare the resulting region-wise candidates and select the best candidate for the restricted joint problem.
	\end{itemize}
	The exact coefficient evaluation in SP2 requires solving the UMSBP and is NP-hard. Consequently, solving the restricted problem exactly through Eq.~(\ref{Eq. d2A mapping}) contains an NP-hard computational bottleneck.
\end{theorem}
}

\textcolor{blue}{
\begin{proof}
	See Appendix~\ref{App. unimodular NP-hardness proof}.
\end{proof}
}

\textcolor{blue}{
\begin{remark}
	The four subproblems reveal a hierarchical solution strategy for the coefficient-labelled geometric formulation. SP1 performs the continuous scaling optimization associated with a fixed coefficient region, while SP2 verifies whether the resulting solution is consistent with that region through the exact coefficient mapping. SP3 explores the self-consistent regions identified by SP2, and SP4 compares and refines the resulting region-wise candidates. This hierarchy separates fixed-region optimization, coefficient consistency verification, cross-region exploration, and candidate selection, thereby motivating the practical algorithm design in the subsequent section. Since exact coefficient evaluation and exhaustive region exploration are computationally demanding, the subsequent practical design employs approximate coefficient updates and sampled cross-region exploration rather than solving the exact formulation.
\end{remark}
}

\subsection{Region-Wise Optimization for Fixed $\mathbf A$}\label{SSect. Model on Fixed A}
\textcolor{blue}{Within a coefficient-dependent region $\mathcal P(\mathbf A)$, the coefficient matrix is fixed. The corresponding region-wise scaling problem is}
\begin{align}\label{Eq. fixed A optimization problem}
	\underset{\mathbf{d}}{\mathrm{\arg\min}}~~~&\mathbf{d}^{\top}\left(\left(\mathbf{A}\mathbf{A}^{\top}\right)\circ\mathbf{M}\right)\mathbf{d},\\
	\text{subject to}~~~&\mathbf{d}\in\varOmega, \notag
\end{align}
where the interference coupling matrix $\mathbf{G}=\left(\mathbf{A}\mathbf{A}^{\top}\right)\circ\mathbf{M}$ is constant in this situation. 
The following lemma states the conditions under which Eq.~(\ref{Eq. fixed A optimization problem}) is convex.
\begin{lemma}\label{Lem. extreme existence and uniqueness}
	Given a fixed $\mathbf{A}$,
 	Eq. (\ref{Eq. fixed A optimization problem}) is a convex problem if and only if any of the following conditions are satisfied:
 	\begin{itemize}
 		\item[i)] $N\geq K$ in DIF,
 		\item[ii)] any $\mathbf{H}^{K\times N}$ in RIF.
 	\end{itemize}
\end{lemma}
\begin{proof}
	Since the difference of $\mathbf{M}$ in DIF and RIF, we will elaborate these two situations respectively.
	\begin{itemize}
		\item \textit{DIF: }According to Eq. (\ref{Eq. DIF M}), $\mathbf{M}=(\mathbf{H}\mathbf{H}^{\top})^{-1}$.  According to the features of Gram matrix, $\mathbf{H}\mathbf{H}^{\top}$ is semi-positive definite matrix (SPDM), and $\mathbf{A}\mathbf{A}^{\top}$ is positive definite matrix (PDM) since $\mathbf{A}$ is full-rank. When $N\geq K$, $(\mathbf{H}\mathbf{H}^{\top})^{-1}$ exists and is also SPDM. Based on Schur Product Theorem, $\left(\mathbf{A}\mathbf{A}^{\top}\right)\circ\mathbf{M}$ is SPDM, which reveals that the target function in DIF is convex.
		\item \textit{RIF: }Regarding the singular value decomposition (SVD) of $\mathbf{H}\mathbf{H}^{\top}=\mathbf{Q}\mathbf{\Lambda}\mathbf{Q}^{\top}$, we have
		\begin{equation}
			\frac{K}{\rho}\mathbf{I}+\mathbf{H}\mathbf{H}^{\top}=\mathbf{Q}\left(\frac{K}{\rho}\mathbf{I}+\mathbf{\Lambda}\right)\mathbf{Q}^{\top}.
		\end{equation}
		Since $\frac{K}{\rho}>0$, $\frac{K}{\rho}\mathbf{I}+\mathbf{H}\mathbf{H}^{\top}$ is a PDM. Similar to DIF, $\mathbf{M}=\left(\frac{K}{\rho}\mathbf{I}+\mathbf{H}\mathbf{H}^{\top}\right)^{-1}$ is the PDM, and $\left(\mathbf{A}\mathbf{A}^{\top}\right)\circ\mathbf{M}$ is also the PDM. Thus, the target function in RIF is convex.
	\end{itemize}
	Combining the above two points, $\mathbf{d}^{\top}\left(\left(\mathbf{A}\mathbf{A}^{\top}\right)\circ\mathbf{M}\right)\mathbf{d}$ is convex when $N\geq K$ in DIF and any $\mathbf{H}^{K\times N}$ in RIF. The lemma has been proved.
\end{proof}

\textcolor{blue}{The Lagrangian of the fixed-$\mathbf A$ scaling problem is}
	\begin{equation}
		L_{\mathrm{fix}}(\mathbf d,\alpha)=\mathbf d^{\top}\mathbf G\mathbf d+\alpha\sum_{i=1}^{K}\log d_i.
	\end{equation}
\textcolor{blue}{Its stationarity conditions are}
\begin{equation}\label{Eq. fixed A equation system 1}
	\left\{
	\begin{aligned}
		&\triangledown L_{\mathrm{fix}}(\mathbf{d},\alpha)=2\mathbf{G}\mathbf{d}
		+\alpha(\mathbf{1}_K\oslash\mathbf{d})=\mathbf{0}_K;\\
		&\sum_{i=1}^{K}\log(d_i)=0.
	\end{aligned}
	\right.
\end{equation}
\textcolor{blue}{Geometrically, the first stationarity condition states that the gradient of the fixed-$\mathbf A$ quadratic objective is aligned with the normal direction of the scaling manifold $\varOmega$. Therefore, the fixed-$\mathbf A$ extremum can be interpreted as a tangency point between $\varOmega$ and a level set of the quadratic objective.}
For $\triangledown L_{\mathrm{fix}}(\mathbf{d},\alpha)$, since $\mathbf{d}^{\top}\mathbf{0}_K=0$ and $\mathbf{d}^{\top}(\mathbf{1}_K\oslash\mathbf{d})=K$, Eq. (\ref{Eq. fixed A equation system 1}) is equivalent to 
\begin{equation}\label{Eq. fixed A equation system 2}
	\left\{
	\begin{aligned}
		&\mathbf{d}^{\top}\mathbf{G}\mathbf{d}=-\frac{\alpha K}{2};\\
		&\sum_{i=1}^{K}\log(d_i)=0.
	\end{aligned}
	\right.
\end{equation}
\textcolor{blue}{Since $\mathbf G$ is positive definite according to Lemma~\ref{Lem. extreme existence and uniqueness}, there exists a nonsingular matrix $\mathbf T$ such that $\mathbf G=\mathbf T^{\top}\mathbf T$. The fixed-$\mathbf A$ objective can therefore be written as
	\begin{equation}
		\label{Eq. generalized model shifting}
		\mathbf d^{\top}\mathbf G\mathbf d
		=
		\|\mathbf T\mathbf d\|_2^2.
	\end{equation}
	Define the transformed feasible manifold as
	\begin{equation}
		\mathbf T\varOmega
		\triangleq
		\left\{
		\mathbf y=\mathbf T\mathbf d
		\mid
		\mathbf d\in\varOmega
		\right\}.
	\end{equation}
	Then the fixed-$\mathbf A$ scaling problem in Eq.~(\ref{Eq. fixed A optimization problem}) is equivalently represented by
	\begin{equation}
		\label{Eq. fixed A optimization model}
		\mathbf y_{\mathrm{opt}}
		=
		\underset{\mathbf y\in\mathbf T\varOmega}{\arg\min}
		\|\mathbf y\|_2^2,
		\qquad
		\mathbf d_{\mathrm{opt}}
		=
		\mathbf T^{-1}\mathbf y_{\mathrm{opt}}.
	\end{equation}
	Hence, the fixed-$\mathbf A$ extremum corresponds to the point of the transformed manifold $\mathbf T\varOmega$ that is closest to the origin. This geometric representation shows how the quadratic metric $\mathbf G$ changes the feasible scaling geometry and determines the direction of the fixed-$\mathbf A$ extremum.}

\begin{figure*}[t!]
	\centering
	\subfigure[$\mathbf{T}=\mathbf{\Lambda}^{\frac{1}{2}}\mathbf{Q}^{\top}$ in SVD]{\includegraphics[width=.32\textwidth]{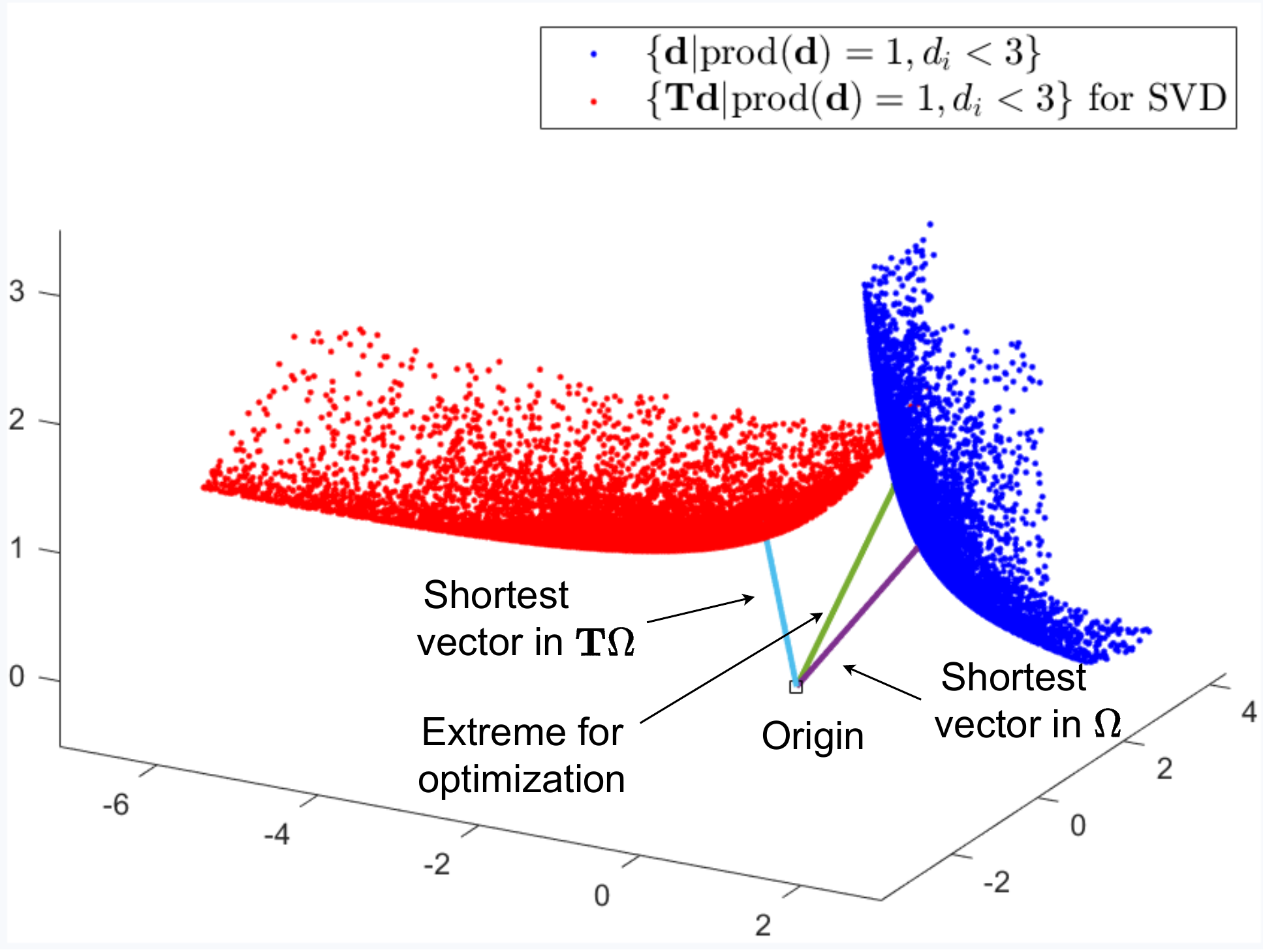}}
	\subfigure[$\mathbf{T}=\mathbf{\Sigma}^{\frac{1}{2}}\mathbf{L}^{\top}\mathbf{O}$ in LDL]{\includegraphics[width=.32\textwidth]{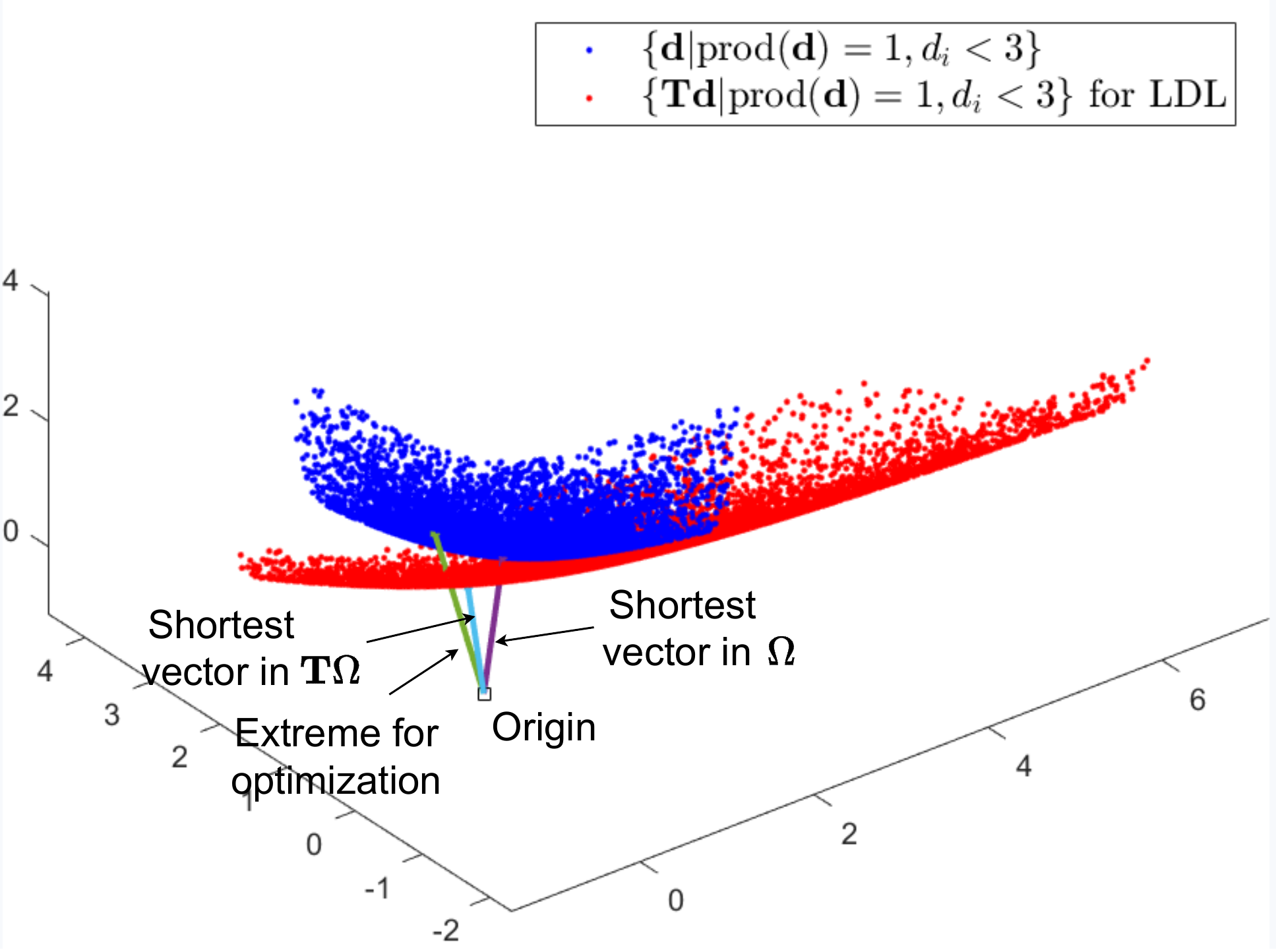}}
	\subfigure[$\mathbf{T}=\mathbf{\Sigma}^{\frac{1}{2}}$ in \cite{venturelli2020optimization}]{\includegraphics[width=.32\textwidth]{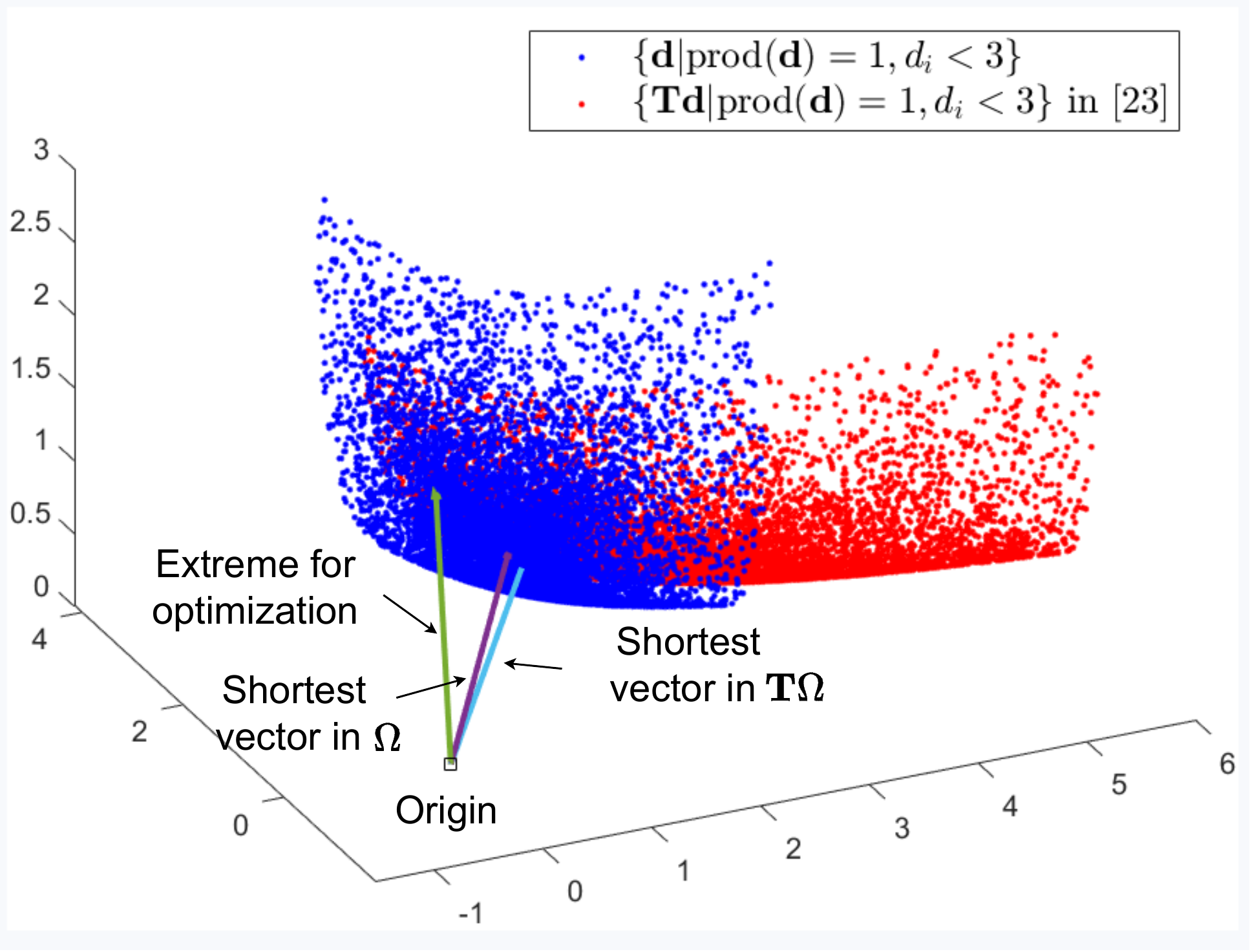}}
	\caption{\textcolor{blue}{Geometric interpretation of the fixed-$\mathbf A$ scaling problem. The transformation $\mathbf y=\mathbf T\mathbf d$ maps the normalized scaling manifold $\varOmega$ to $\mathbf T\varOmega$. The point of $\mathbf T\varOmega$ closest to the origin corresponds to the fixed-$\mathbf A$ stationary solution.}}
	\label{Fig. generalized model}
\end{figure*}

Considering SVD $\mathbf{G}=\mathbf{T}^{\top}\mathbf{T}=\mathbf{Q}\mathbf{\Lambda}\mathbf{Q}^{\top}$ and Lower-Diagonal-Lower Transpose (LDL) $\mathbf{O}\mathbf{G}\mathbf{O}^{\top}=\mathbf{O}\mathbf{T}^{\top}\mathbf{T}\mathbf{O}^{\top}=\mathbf{L}\mathbf{\Sigma}\mathbf{L}^{\top}$ \cite{venturelli2020optimization}, Fig. \ref{Fig. generalized model} depicts some samples in $3$-dimension on Eq. (\ref{Eq. fixed A optimization model}), in which SNR=$15$dB and
\begin{equation}
	\mathbf{H}=\left[
	\begin{matrix}
		0.7826	&0.6097	&0.7154\\
		1.8776	&1.8774	&0.1899\\
		1.8507	&0.0694	&0.3630
	\end{matrix}
	\right].\notag
\end{equation} 
\textcolor{blue}{The geometric observations from Fig.~\ref{Fig. generalized model} are summarized as follows:
\begin{itemize}
	\item[i)] In Figs.~\ref{Fig. generalized model}(a) and (b), the SVD and LDL based factors both satisfy $\mathbf G=\mathbf T^{\top}\mathbf T$ and therefore preserve the fixed-$\mathbf A$ quadratic objective exactly. Although the transformed manifolds may have different orientations, they represent the same fixed-$\mathbf A$ scaling problem when mapped back to $\varOmega$.	
	\item[ii)] Fig.~\ref{Fig. generalized model}(c) shows the diagonal-only transformation used in \cite{venturelli2020optimization}. Omitting the non-diagonal transformation changes the geometry of the transformed feasible manifold and may alter the recovered stationary direction.
\end{itemize}
Thus, the visualization illustrates how the complete quadratic metric affects the feasible scaling geometry and the direction of the fixed-$\mathbf A$ stationary solution.}

\begin{remark}\label{Rem. convex problem in fixed A}
	The fixed-$\mathbf A$ problem is geometrically equivalent to finding the point of $\mathbf T\varOmega$ closest to the origin. Hence, its solution lies on the contact surface of the smallest origin-centered sphere that intersects $\mathbf T\varOmega$, rather than in the interior of that sphere. Moreover, by the positive homogeneity of the quadratic objective, the constraint $\prod_{i=1}^{K}d_i=1$ can be equivalently extended to $\prod_{i=1}^{K}d_i\geq1$, since an optimum of the extended problem must lie on its boundary. The latter constraint is equivalently written as $\sum_{i=1}^{K}\log d_i\geq0$ and defines a convex set. Combined with Lemma~\ref{Lem. extreme existence and uniqueness}, this gives an equivalent convex formulation of the fixed-$\mathbf A$ problem.
\end{remark}

Comprehensive the above contents, we have Theorem \ref{Pro. solution set} for the characteristic of solution set in Eq. (\ref{Eq. new optimization problem})
\begin{theorem}\label{Pro. solution set}
	The solution set $\mathcal{S}$ of Eq. (\ref{Eq. new optimization problem}) is a finite nonempty set in $\mathcal{D}$.
\end{theorem}
\begin{proof}
	See Appendix \ref{App. the charteristic of solution set}.
\end{proof}
\begin{remark}
	If we regard $\varOmega$ as a set of $\mathcal{P}(\mathbf{A})$, then searching for the optimal extremum on $\varOmega$ can be interpreted as a search over the elements of $\mathcal{S}$. According to Theorem \ref{Pro. solution set}, $\mathcal{S}$ is a proper subset of the real number space $\mathbb{R}^K$, which implies that the search space can be reduced and the time complexity lowered when $\mathcal{S}$ is known, either completely or in part. If $\mathcal{S}$ is fully known, the optimal extremum can be selected from the extrema corresponding to all $\mathcal{P}(\mathbf{A})$ associated with $\mathcal{S}$; if only partially known, a sub-optimal extremum can be obtained.
\end{remark}

\section{Practical IF Precoder Design}
\textcolor{blue}{Guided by the preceding geometric structure, we develop Multi-Cone Nested Stochastic Pattern Search (MCN-SPS) to search for a high-rate feasible coefficient-scaling pair. The method consists of fixed-$\mathbf A$ scaling optimization (Solving SP1), rate-safeguarded alternating refinement (Solving the single region searching problem in SP3), and stochastic cross-region exploration (Solving the multiple region searching problem in SP3 and SP4).}
\subsection{Optimization for Fixed $\mathbf{A}$}
To solve Eq. (\ref{Eq. fixed A optimization model}), from Fig. \ref{Fig. generalized model}, we find that the shape of $\varOmega$ forms a cone in the positive orthant of $\mathcal{D}$. We have the following lemma for this hypothesis.
\begin{lemma} \label{Lem. Uniqueness of intersection}
	To $\varOmega$, any ray started from the origin and towards $\mathcal{D}$'s positive orthant $\mathcal{D}^{+}$ have and only have one intersection. For the space $\mathcal{B}$ transformed by $\mathbf{B}$ from $\mathcal{D}$, this characteristic holds true for $\varOmega_2=\left\{\mathbf{B}(\mathbf{1}_K\oslash\mathbf{d})|\prod_{i=1}^{K}d_i=1,d_i>0,i\in {1,\cdots,K}\right\}$, where $\mathbf{B}$ represents arbitrarily non-negative matrix.
\end{lemma} 
\begin{proof}
	See Appendix \ref{App. Uniqueness of intersection proof}.
\end{proof}
\textcolor{blue}{Lemma~\ref{Lem. Uniqueness of intersection} shows that each positive ray has a unique intersection with $\varOmega$. Hence, a feasible scaling vector on $\varOmega$ can be identified by its direction, while its magnitude is uniquely determined by the normalization constraint. To characterize this direction, we rewrite the stationarity condition in Eq.~(\ref{Eq. fixed A equation system 1}). Multiplying its $i$-th component by $d_i$ gives}
\begin{align}\label{Eq. fixed A process 1}
	&\left\{
	\begin{aligned}
		&2\mathbf{G}\mathbf{d}
		+\alpha(\mathbf{1}_K\oslash\mathbf{d})=\mathbf{0}_K\\
		&\sum_{i=1}^{K}\log(d_i)=0
	\end{aligned}
	\right.
	\Rightarrow
	\left\{
	\begin{aligned}
		&d_1[\mathbf{G}\mathbf{d}]_1=-\frac{\alpha}{2}\\
		&\cdots\\
		&d_K[\mathbf{G}\mathbf{d}]_K=-\frac{\alpha}{2}\\
		&\sum_{i=1}^{K}\log(d_i)=0
	\end{aligned}
	\right..
\end{align}
Eq. (\ref{Eq. fixed A process 1}) shows that $d_1\breve{\mathbf g}_1\mathbf{d}=\cdots=d_K\breve{\mathbf g}_K\mathbf{d}=-\alpha/2$.
\textcolor{blue}{Thus, all component-wise products $d_i\breve{\mathbf g}_i\mathbf d$ are balanced at the fixed-$\mathbf A$ extremum.}
\textcolor{blue}{Since Eq.~(\ref{Eq. fixed A process 1}) gives}
\begin{equation}
	\frac{d_i}{d_j}=\frac{[\mathbf{G}\mathbf{d}]_j}{[\mathbf{G}\mathbf{d}]_i},~i\neq j,
\end{equation}
\textcolor{blue}{we obtain $\mathbf G\mathbf d\propto\mathbf 1_K\oslash\mathbf d$.
Together with the normalization on $\varOmega$, the fixed-$\mathbf A$ directional system is}
\begin{align}\label{Eq. fixed A optimization on direction}
	&\left\{
	\begin{aligned}
		&\mathbf{G}\mathbf{d}\propto\mathbf{1}_K\oslash\mathbf{d};\\
		&\prod_{k=1}^{K}d_k=1.
	\end{aligned}
	\right.
\end{align}
\textcolor{blue}{The proportional relation in Eq.~(\ref{Eq. fixed A optimization on direction}) determines the direction of the fixed-$\mathbf A$ solution. By Lemma~\ref{Lem. Uniqueness of intersection}, this direction has a unique feasible representative on $\varOmega$. Hence, the geometric ray property connects the stationarity condition to the normalized directional formulation used in the subsequent algorithm.}

To solve Eq. (\ref{Eq. fixed A optimization on direction}), we can recognize it as a special case of matrix balancing problem. Here we have the following theorem about this equivalence.
\begin{theorem}\label{The. unilateral matrix balancing problem equivalence}
	The constraint $\prod_{i=1}^{K}d_i=1$ can be expressed by the scaling function
	\begin{equation}
		\Gamma(\mathbf{x})=\vartheta(\mathbf{x})\mathbf{x},~\vartheta(\mathbf{x})\triangleq\frac{1}{\prod_{i=1}^{K}x_i^{\frac{1}{K}}},~\mathbf{x}=[x_1,\cdots,x_K]^{\top}.
	\end{equation}
	Let $\mathbf{G}$ be an M matrix \cite{berman1994nonnegative}. The problem in fixed $\mathbf{A}$ is equivalent to a unilateral matrix balancing problem with row sum constraints. 
	
	That is, we seek a diagonal matrix $\mathbf{C}$ such that the matrix $\mathbf{B}=\vartheta(\mathbf{G}^{-1}\mathbf{C}\cdot\mathbf{1}_K)\mathbf{C}\mathbf{G}^{-1}\mathbf{C}$ is row stochastic, i.e., it satisfy \cite{sayed2014adaptation}:
	\begin{equation}
		\mathbf{B}\cdot\mathbf{1}_K=\mathbf{1}_K.
	\end{equation}
\end{theorem}
\begin{proof}
	Due to the constrain $\prod_{i=1}^{K}d_i=1$, Eq. (\ref{Eq. fixed A optimization on direction}) can be merged into
	\begin{equation}
		\mathbf{d}=\Gamma(\mathbf{G}^{-1}(\mathbf{1}_K\oslash\mathbf{d}))=\vartheta(\mathbf{G}^{-1}(\mathbf{1}_K\oslash\mathbf{d}))\mathbf{G}^{-1}\left(\mathbf{1}_K\oslash\mathbf{d}\right).
	\end{equation}
	The above equation can be reformulated as 
	\begin{align}
		&\mathbf{d}=\vartheta(\mathbf{G}^{-1}(\mathbf{1}_K\oslash\mathbf{d}))\mathbf{G}^{-1}\left(\mathbf{1}_K\oslash\mathbf{d}\right)\notag\\
		\Rightarrow&\mathbf{D}\cdot\mathbf{1}_K=\vartheta(\mathbf{G}^{-1}\mathbf{D}^{-1}\cdot\mathbf{1}_K)\mathbf{G}^{-1}\mathbf{D}^{-1}\cdot\mathbf{1}_K\notag\\
		\Rightarrow&\vartheta(\mathbf{G}^{-1}\mathbf{D}^{-1}\cdot\mathbf{1}_K)\mathbf{D}^{-1}\mathbf{G}^{-1}\mathbf{D}^{-1}\cdot\mathbf{1}_K=\mathbf{1}_K.
	\end{align}
	Considering $\mathbf{C}=\mathbf{D}^{-1}$, we have 
	\begin{equation}\label{Eq. unilateral matrix balancing problem equivalence 1}
		\vartheta(\mathbf{G}^{-1}\mathbf{C}\cdot\mathbf{1}_K)\mathbf{C}\mathbf{G}^{-1}\mathbf{C}\cdot\mathbf{1}_K=\mathbf{1}_K.
	\end{equation}
	Since $\mathbf{G}$ is an M matrix and is PDM in Lemma \ref{Lem. extreme existence and uniqueness}, $\mathbf{G}^{-1}$ must be a non-negative matrix according to \cite{berman1994nonnegative}. Considering the matrix $\mathbf{B}=\vartheta(\mathbf{G}^{-1}\mathbf{C}\cdot\mathbf{1}_K)\mathbf{C}\mathbf{G}^{-1}\mathbf{C}$, Eq. (\ref{Eq. unilateral matrix balancing problem equivalence 1}) can be reformulated as $\mathbf{B}\cdot\mathbf{1}_K=\mathbf{1}_K$. As the non-negative characteristic of $\mathbf{G}^{-1}$ and $\mathbf{D}$, $\mathbf{B}$ is a non-negative matrix, which also be a row stochastic matrix because the definition in \cite{sayed2014adaptation}. 
\end{proof}
Theorem \ref{The. unilateral matrix balancing problem equivalence} shows that the $\mathbf{D}$ optimization in fixed $\mathbf{A}$ can be regarded as an unilateral matrix balancing problem with row sum constraints. Inspired by the Sinkhorn iteration \cite{NIPS2017_491442df}, the function is formulated as
\begin{equation}\label{Eq. iterative function}
	\mathbf{d}^{(t+1)}=\Gamma\left(\mathbf{G}^{-1}\left(\mathbf{1}_K\oslash\mathbf{d}^{(t)}\right)\right),
\end{equation}
where $\Gamma(\mathbf{x})=\frac{\mathbf{x}}{\prod_{i=1}^{K}x_i^{\frac{1}{K}}}$, $\mathbf{x}=[x_1,\cdots,x_K]^{\top}$ is a normalized function to satisfy the constraint $\prod_{i=1}^{K}d_i=1$ and $t$ denotes the index of iteration. The convergence of Eq. (\ref{Eq. iterative function}) is proved in the next theorem. 
\begin{theorem}\label{The. Convergation of proposed method}
	Towards the $p$-PAM symbols in Rayleigh channel, assuming $\mathbf{A}\mathbf{A}^{\top}$ is non-negative, 
	Eq. (\ref{Eq. iterative function}) is converged in direction, i.e.,
	\begin{equation}\label{Eq. iterative terminate}
		\mathbf{d}^{(t)}\overset{t\rightarrow\infty}{\propto}\Gamma\left(\mathbf{G}^{-1}\left(\mathbf{1}_K\oslash\mathbf{d}^{(t)}\right)\right). 
	\end{equation}
	The aforementioned conclusion holds true in RIF, and in DIF when $N\geq K$.
\end{theorem}
\begin{proof}
	See Appendix \ref{App. Convergation of proposed method proof}.
\end{proof}
Besides the convergence, Theorem \ref{The. unilateral matrix balancing problem equivalence} and Theorem \ref{The. Convergation of proposed method} reveal some features of Eq. (\ref{Eq. iterative function}):
\begin{itemize}
	\item[i)] When $\mathbf{G}$ is a $M$ matrix, its inverse $\mathbf{G}^{-1}$ is a non-negative matrix, ensuring that the $\mathbf{d}^{(t)}$ in $t$-th iteration is belong to $\varOmega$ without any modifying.
	\item[ii)] When $\mathbf G$ is not an $M$-matrix, the contraction factor $\gamma$ may increase toward $1$ \cite{berman1994nonnegative}, and the iteration may exhibit a period-two oscillation between two limiting directions. These directions are retained as candidates for subsequent rate evaluation.
\end{itemize}

\begin{algorithm}[t]
	\caption{Reciprocal Approximation (RA)}
	\label{Alg. Reciprocal Approximation}
	
	\textbf{Input:} 
	Channel coefficient matrix $\mathbf{H}$, SNR $\rho$, integer coefficient matrix $\mathbf{A}$, tolerance $\epsilon$ \\
	\textbf{Output:} 
	Optimal power allocation matrix $\mathbf{D}$
	
	\begin{algorithmic}[1]
		\State $[N, K] \gets \text{size}(\mathbf{H})$
		\State $\mathbf{M} \gets \left(\frac{K}{\rho}\mathbf{I} + \mathbf{H}\mathbf{H}^{\top}\right)^{-1}$ 
		\State $\mathbf{G} \gets (\mathbf{A}\mathbf{A}^{\top}) \circ \mathbf{M}$ 
		\State $\mathbf{D} \gets \text{Init}(\mathbf{G})$ \Comment{Initialize diagonal matrix}
		\State \textcolor{blue}{$r_H \gets +\infty$	}	
		\While{\textcolor{blue}{$r_H > \epsilon$}} 
		\State $\mathbf{d}_{\text{old}} \gets \mathbf{D}\mathbf{1}_K$ 
		\State $\mathbf{d}_{\text{new}} \gets \Gamma(\mathbf{G}^{-1}(\mathbf{1}_K \oslash \mathbf{d}_{\text{old}}))$ 
		\State $\mathbf{D} \gets \text{diag}(\mathbf{d}_{\text{new}})$ 
		\State \textcolor{blue}{$r_H \gets \log\left(\dfrac{\max(\mathbf{d}_{\text{new}} \oslash \mathbf{d}_{\text{old}})}{\min(\mathbf{d}_{\text{new}} \oslash \mathbf{d}_{\text{old}})}\right)$ }
		\EndWhile
	\end{algorithmic}
\end{algorithm}


\textcolor{blue}{The fixed-$\mathbf A$ optimization is summarized in Alg.~\ref{Alg. Reciprocal Approximation}, referred to as Reciprocal Approximation (RA). The Hilbert metric is adopted because it is invariant to positive scaling and therefore measures convergence in direction. This agrees with Eq.~(\ref{Eq. fixed A optimization on direction}), while Lemma~\ref{Lem. Uniqueness of intersection} ensures that the converged direction has a unique representative on $\varOmega$. 
The Euclidean distance can be used only after fixing the scale of each iterate, since it is not invariant to independent positive scalings and therefore does not directly measure convergence of rays.
In contrast, the Hilbert metric directly measures convergence in direction and provides the Birkhoff--Hopf contraction used in Theorem~\ref{The. Convergation of proposed method} and the iteration bound in Theorem~\ref{The. upper bound of RA iterative number}.
Hence, Perron--Frobenius theory is required for the directional convergence analysis rather than for computing each RA update. Moreover, $r_H\leq\epsilon$ means that the ratio between the largest and smallest componentwise changes is no larger than $\text{exp}\{\epsilon\}$.}

Regarding Eq. (\ref{Eq. iterative function}), based on Theorem \ref{The. unilateral matrix balancing problem equivalence}, there are several conditions to ensure the effectiveness of non-negative matrix constraint for matrix balancing problem: 

\noindent \textbf{i) channel hardening effect happening: } When $N\gg K$, according to random matrix theory \cite{Tao2012}, the Gram matrix $\mathbf{H}\mathbf{H}^{\top}$ can be regarded as a diagonal matrix. Since the positive definiteness proved in Lemma \ref{Lem. extreme existence and uniqueness}, the main diagonal elements of $\mathbf{M}=\left[m_{i,j}\right]_{i,j\in\{1,\cdots,K\}}$ is positive and the others are zero, which is represented by 
	\begin{equation}\label{Eq. M elements}
			m_{i,j}
		=\left\{
		\begin{aligned}
			&m_{i,j}>0,&i=j;\\
			&0,&i\neq j.
		\end{aligned}
		\right.
	\end{equation}
	Then, since $\mathbf{A}\mathbf{A}^{\top}$ is SPDM, we have
	\begin{align}\label{Eq. G elements}
		g_{i,j}
		&=\left\{
		\begin{aligned}
			&a^{'}_{i,j}\times m_{i,j},&i=j;\\
			&0,&i\neq j.
		\end{aligned}
		\right.
	\end{align}
	where $a^{'}_{i,j}$ denotes the $\{i,j\}$ elements of $\mathbf{A}\mathbf{A}^{\top}$ and $g_{i,j}$ the $\{i,j\}$ elements of $\mathbf{G}$. 
	Eq. (\ref{Eq. G elements}) reflects that $\mathbf{G}$ is a M matrix \cite{berman1994nonnegative} whose main diagonal elements is non-negative and the others non-positive.
	Combining the properties of M Matrices \cite{berman1994nonnegative} and $\mathbf{G}$'s positive definiteness proved in Lemma \ref{Lem. extreme existence and uniqueness}, we learn that $\mathbf{G}^{-1}$ is a positive transformation and conforms to the famous Perron Theorem \cite{kohlberg1982contraction}.  

We have the following theorem for the approximate solution of Eq. (\ref{Eq. fixed A optimization on direction}) when $\mathbf{M}$ is a diagonal matrix.
\begin{theorem}\label{The. approximate solution for diagonal matrix}
	When $\mathbf{M}$ is a diagonal matrix, a solution of Eq. (\ref{Eq. fixed A optimization on direction}) is given by
	\begin{align}
		\mathbf{D}_{\mathrm{opt}}=\frac{\mathbf{L}^{\frac{1}{2}}}{\det(\mathbf{L})^{\frac{1}{2K}}},
	\end{align}
	where
	\begin{equation}
		\mathbf{L}=\mathrm{diag}\left(\left[\left(\mathbf{A}\mathbf{A}^{\top}\right)\circ\mathbf{M}\right]^{-1}\cdot\mathbf{1}_K\right).
	\end{equation}
\end{theorem}
\begin{proof}
	Since $\mathbf{M}$ is a diagonal matrix, it follows Eq. (\ref{Eq. M elements}). By substituting Eq. (\ref{Eq. G elements}) into Eq. (\ref{Eq. fixed A optimization on direction}), we have 
	\begin{equation}\label{Eq. Special case equations system}
		\left\{
		\begin{aligned}
			\left[
			\begin{matrix}
				d_1g_{1,1}&	\cdots&	0\\
				\cdots&	\ddots&	\cdots\\
				0&	\cdots&	d_Kg_{K,K}
			\end{matrix}
			\right]\mathbf{1}_K&\propto
			\left[
			\begin{matrix}
				d_1^{-1}&	\cdots&	0\\
				\cdots&	\ddots&	\cdots\\
				0&	\cdots&	d_K^{-1}
			\end{matrix}
			\right]\mathbf{1}_K;\\
			\prod_{i=1}^{K}d_i=1&,~i\in\{1,\cdots,K\}.
		\end{aligned}
		\right.
	\end{equation}
	After solving Eq. (\ref{Eq. Special case equations system}), we can derive that
	\begin{align}
		&\left\{
		\begin{aligned}
			d_1&=\left(\prod_{i=1}^{K}g_{i,i}^{-\frac{1}{2K}}\right)g_{1,1}^{-\frac{1}{2}};\\
			&\cdots\\
			d_1&=\left(\prod_{i=1}^{K}g_{i,i}^{-\frac{1}{2K}}\right)g_{K,K}^{-\frac{1}{2}},
		\end{aligned}
		\right.
		\Rightarrow
		\mathbf{D}_{\mathrm{opt}}=\frac{\left(\mathrm{diag}\left(\mathbf{G}^{-1}\cdot\mathbf{1}_K\right)\right)^{\frac{1}{2}}}{\det(\mathrm{diag}\left(\mathbf{G}^{-1}\cdot\mathbf{1}_K)\right)^{\frac{1}{2K}}}.
	\end{align}
	Since $\mathbf{G}=\left(\mathbf{A}\mathbf{A}^{\top}\right)\circ\mathbf{M}$, the theorem is proved.
\end{proof}

\noindent \textbf{ii) DIF with LLL algorithm: }When $N\geq K$, DIF employs $\mathbf{M}=(\mathbf{H}\mathbf{H}^{\top})^{-1}$ which is a M matrix according to Appendix \ref{App. Convergation of proposed method proof}. Since $\mathbf{H}$ is with Rayleigh distribution, $\mathbf{M}$ is not a sparse matrix which causes period-2 oscillation in iterative. According to Lemma \ref{Lem. LLL mu in M matrix}, the transforming matrix to reduce the basis $\mathbf{M}^{\frac{1}{2}}\mathbf{D}$ in size reduction condition is non-negative. Since the swapping matrix in Lov$\acute{a}$sz condition is non-negative too, the unimodular matrix $\mathbf{A}$ outputted by LLL algorithm in this situation is non-negative, which meets to the condition that $\mathbf{A}\mathbf{A}^{\top}$ is non-negative in Theorem \ref{The. Convergation of proposed method}.
\begin{lemma}\label{Lem. LLL mu in M matrix}
	In LLL algorithm, the coefficients $\mu_{i,j}$ of size reduction is always non-positive in iteration when the inputted basis $\mathbf{B}=\left[\mathbf{b}_1,\cdots,\mathbf{b}_K\right]$ satisfies that $\mathbf{B}^{\top}\mathbf{B}$ is a dense M matrix. 
\end{lemma}
\begin{proof}
	We prove this lemma by induction. Since $\mathbf{B}\mathbf{B}^{\top}$ is a dense M matrix, its non-major diagonal elements are non-positive, which means
	\begin{equation}
		\mathbf{b}_i^{\top}\mathbf{b}_j\leq0,~~i\neq j.
	\end{equation} 
	After Gram-Schmidt Orthogonalization step in LLL algorithm, we obtain the orthogonal vector $\mathbf{b}^{*}_1,\cdots,\mathbf{b}^{*}_K$, and $\mathbf{b}^{*}_1=\mathbf{b}_1$. When $i=2$, $\mu_{2,1}$ is given by
	\begin{equation}
		\mu_{2,1}=\frac{\mathbf{b}_2^{\top}\mathbf{b}^{*}_1}{(\mathbf{b}^{*}_1)^{\top}\mathbf{b}^{*}_1}
		=\frac{\mathbf{b}_2^{\top}\mathbf{b}_1}{\mathbf{b}_1^{\top}\mathbf{b}_1}\leq 0.
	\end{equation}
	By assuming $\forall j<p<i,~\mu_{p,j}\leq 0$, for $i>j$, we have
	\begin{equation}
		\mathbf{b}_i^{\top}\mathbf{b}_j^{*}=\mathbf{b}_i^{\top}\mathbf{b}_j-\sum_{p=1}^{j-1}\mu_{p,j}\left(\mathbf{b}_i^{\top}\mathbf{b}_p^{*}\right).
	\end{equation}
	Since $\mu_{p,j}\leq 0$ and $\mathbf{b}_i^{\top}\mathbf{b}_p^{*}\leq0$, $\sum_{p=1}^{j-1}\mu_{p,j}\left(\mathbf{b}_i^{\top}\mathbf{b}_p^{*}\right)\geq0$, so $\mathbf{b}_i^{\top}\mathbf{b}_j^{*}\leq\mathbf{b}_i^{\top}\mathbf{b}_j\leq0$. Thus, we have
	\begin{equation}
		\mu_{i,j}=\frac{\mathbf{b}_i^{\top}\mathbf{b}^{*}_j}{(\mathbf{b}^{*}_j)^{\top}\mathbf{b}^{*}_j}\leq0.
	\end{equation}
	The lemma is proved.
\end{proof}

\begin{figure*}[t!]
	\centering
	\subfigure[SNR=$0$dB]{\includegraphics[width=.4\textwidth]{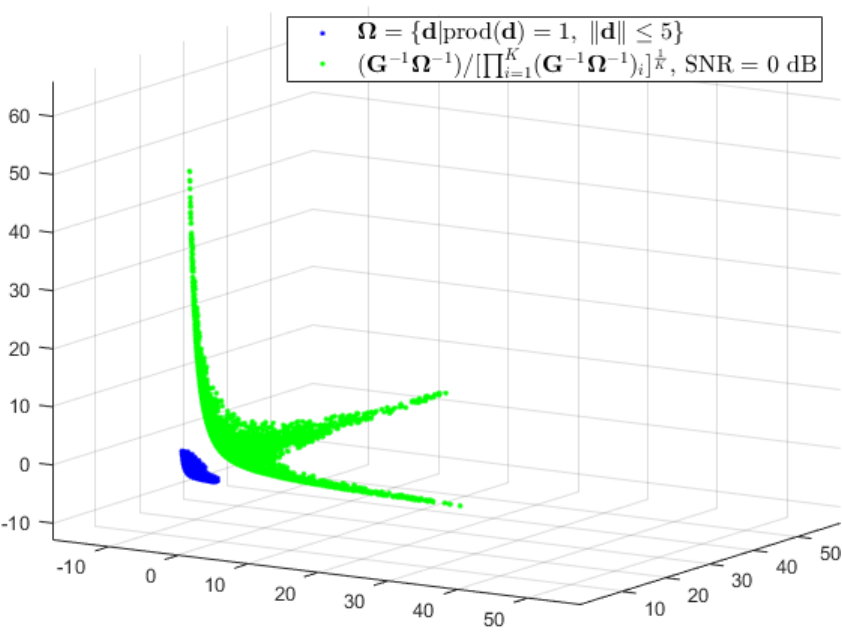}}\quad
	\subfigure[SNR=$5$dB]{\includegraphics[width=.4\textwidth]{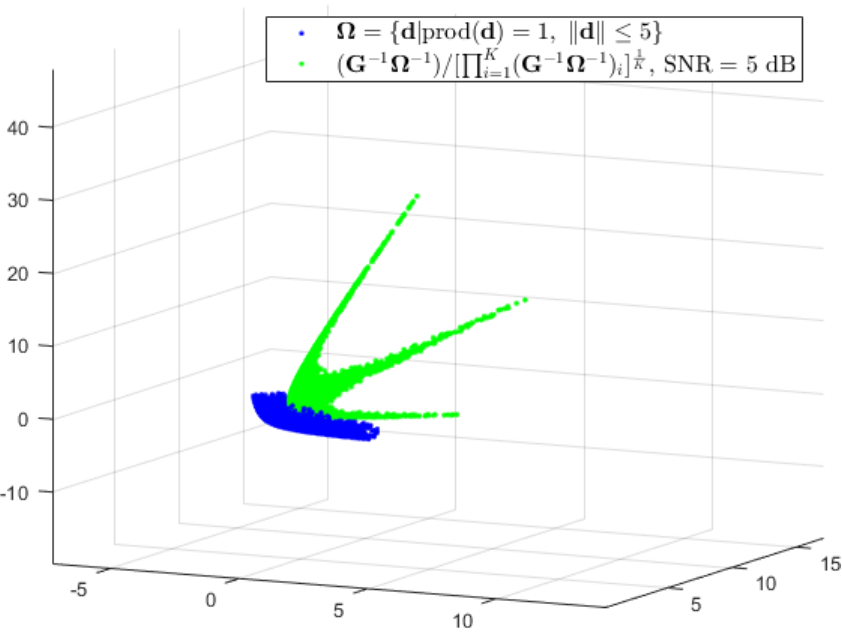}}
	\subfigure[SNR=$10$dB]{\includegraphics[width=.4\textwidth]{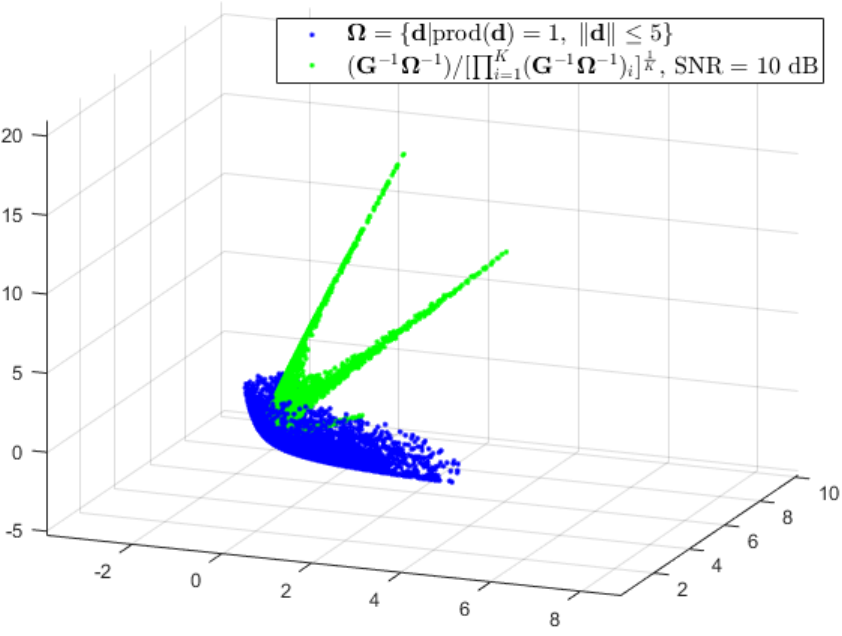}}\quad
	\subfigure[SNR=$15$dB]{\includegraphics[width=.4\textwidth]{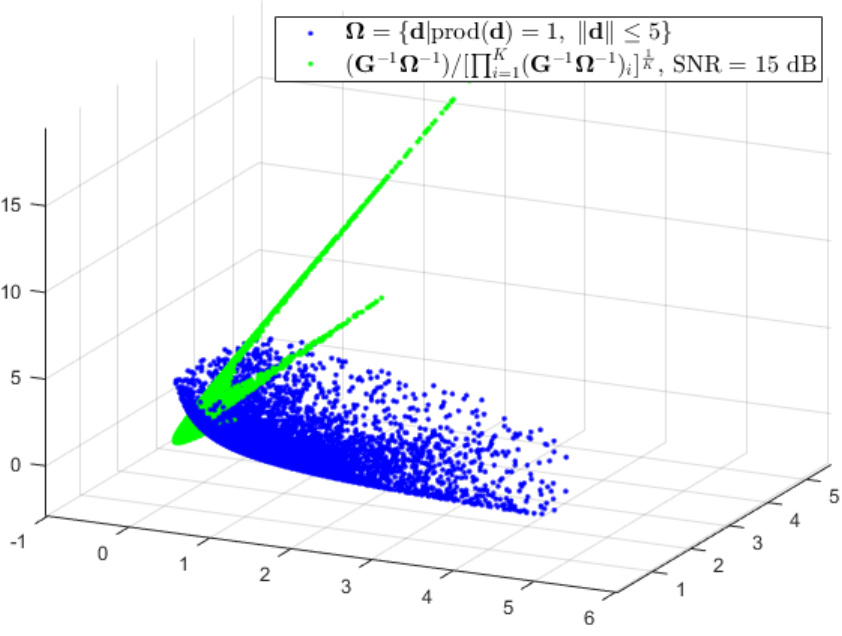}}
	\caption{$3$-dimensional sample for the convergence of Eq. (\ref{Eq. iterative function}) versus SNR}
	\label{Fig. SNR2converage}
\end{figure*}

\noindent \textbf{iii) RIF with LLL algorithm in high SNR: }At sufficiently high SNR, $\frac{K}{\rho}\mathbf{I}$ becomes negligible, which reveals
\begin{equation}
	\mathbf{M}=\left(\frac{K}{\rho}\mathbf{I}+\mathbf{H}\mathbf{H}^{\top}\right)^{-1}
	\overset{\rho\rightarrow+\infty}{\rightarrow}\left(\mathbf{H}\mathbf{H}^{\top}\right)^{-1}.
\end{equation}
As the same reason in case ii, we learn that the unimodular matrix $\mathbf{A}$ outputted by LLL algorithm in this situation is non-negative, which meets to the condition in Theorem \ref{The. Convergation of proposed method}. To explore how SNR affects the convergence, we present a $3$-dimensional sample in
\begin{equation}
	\mathbf{H}=\left[
	\begin{matrix}
		0.5472&	0.1643&	0.4058\\
		1.2850&	0.6481&	1.1394\\
		0.5685&	1.1173&	0.7982
	\end{matrix}
	\right].\notag
\end{equation} 
The observations on Fig. \ref{Fig. SNR2converage} reveal the fact that the point distribution affected by Eq. (\ref{Eq. iterative function}) shrinks during the SNR enlarging, and Eq. (\ref{Eq. iterative function}) mostly converge in direction when SNR$\geq5$dB.
\subsection{Alternating Joint Optimization for $(\mathbf{A},\mathbf{D})$}
\textcolor{blue}{For each accepted pair, RA updates $\mathbf D$ with $\mathbf A$ fixed, after which LLL generates a unimodular coefficient candidate for the updated $\mathbf D$. Since LLL is not an exact optimizer of Eq.~(\ref{Eq. A optimization}), the resulting pair is accepted only if its achievable rate is no smaller than that of the last accepted pair. Otherwise, the candidate is rejected and AJO terminates. The same safeguard applies if LLL is replaced by another unimodular coefficient-update routine.}

\color{blue}
\begin{algorithm}[t]
	\caption{Alternating Joint Optimization (AJO)}
	\label{Alg. Alternating Optimization}
	\hspace{0.02in} {\bf Input:}
	Initial power-scaling matrix $\mathbf D_{\mathrm{init}}$,
	channel matrix $\mathbf H$, SNR $\rho$, and $\mathrm{max\_iter}$.\\
	\hspace{0.02in} {\bf Output:}
	The last accepted pair $(\mathbf A,\mathbf D)$.
	\begin{algorithmic}[1]
		\State $\mathbf M\gets
		\left(\frac{K}{\rho}\mathbf I+\mathbf H\mathbf H^{\top}\right)^{-1}$
		\State $\mathbf D^{(0)}\gets\mathbf D_{\mathrm{init}}$
		\State \textcolor{blue}{Obtain $\mathbf A^{(0)}$ from $\mathbf D^{(0)}$ using the LLL update in Eq.~(\ref{Eq. d2A mapping LLL})}
		\State $R^{(0)}\gets
		R\!\left(\mathbf A^{(0)},\mathbf D^{(0)}\right)$
		\State $t\gets0$
		\While{$t<\mathrm{max\_iter}$}
		\State $\widetilde{\mathbf D}\gets
		\mathrm{RA}\!\left(\mathbf H,\mathbf A^{(t)},\rho\right)$
		\Comment{Update $\mathbf D$ with $\mathbf A^{(t)}$ fixed}
		\State \textcolor{blue}{Obtain $\widetilde{\mathbf A}$ from $\widetilde{\mathbf D}$ using the LLL update in Eq.~(\ref{Eq. d2A mapping LLL})}
		\State $\widetilde R\gets
		R\!\left(\widetilde{\mathbf A},\widetilde{\mathbf D}\right)$
		\If{$\widetilde R<R^{(t)}$}
		\State \textbf{break}
		\Comment{Reject the decreasing candidate}
		\EndIf
		\State $t\gets t+1$
		\State $\left(\mathbf A^{(t)},\mathbf D^{(t)},R^{(t)}\right)
		\gets
		\left(\widetilde{\mathbf A},
		\widetilde{\mathbf D},\widetilde R\right)$
		\If{$\mathbf A^{(t)}==\mathbf A^{(t-1)}$}
		\State \textbf{break}
		\EndIf
		\EndWhile
		\State \Return $\left(\mathbf A^{(t)},\mathbf D^{(t)}\right)$
	\end{algorithmic}
\end{algorithm}
\color{black}
\textcolor{blue}{The AJO loop terminates when the new candidate decreases the achievable rate, when the coefficient matrix remains unchanged, or when $\mathrm{max\_iter}$ is reached. In the first case, the candidate is discarded and the last accepted pair is returned.}
The overall procedure is summarized in Alg. \ref{Alg. Alternating Optimization}. 
\textcolor{blue}{Let $\{(\mathbf A^{(t)},\mathbf D^{(t)})\}$ denote the sequence of accepted AJO pairs. The LLL update is used only to generate a candidate and is not assumed to solve the fixed-$\mathbf D$ subproblem optimally. By the acceptance test in Alg.~\ref{Alg. Alternating Optimization}, every accepted iteration satisfies
\begin{equation}
	R\left(\mathbf A^{(t+1)},\mathbf D^{(t+1)}\right)
	\geq
	R\left(\mathbf A^{(t)},\mathbf D^{(t)}\right).
	\label{Eq. accepted AO monotonicity}
\end{equation}}
\color{blue}
	For $\rho>0$, the matrix $\mathbf M=\left(\frac{K}{\rho}\mathbf I+\mathbf H\mathbf H^{\top}\right)^{-1}$ in Alg.~\ref{Alg. Alternating Optimization} is positive definite. We have Proposition \ref{Pro. AO monotonicity} for the monotonicity of Alg.~\ref{Alg. Alternating Optimization}.
	\begin{proposition}
		\label{Pro. AO monotonicity}
		Assume that $\mathbf M\succ0$. Let $\{(\mathbf A^{(t)},\mathbf D^{(t)})\}$ denote the accepted iterates generated by Alg.~\ref{Alg. Alternating Optimization}, where $\mathbf A^{(t)}\in\mathcal A_{\mathrm{uni}}$ and $\det(\mathbf D^{(t)})=1$. Before termination, their achievable rates satisfy Eq. (\ref{Eq. accepted AO monotonicity}).
		If infinitely many candidates were accepted, the corresponding achievable-rate sequence would converge to a finite limit. The implemented algorithm terminates after at most $\mathrm{max\_iter}$ iterations and returns the last accepted pair.
	\end{proposition}
	\begin{proof}
		The acceptance test in Alg.~\ref{Alg. Alternating Optimization} rejects any candidate whose achievable rate is smaller than that of the last accepted pair. Therefore, every accepted update satisfies Eq.~(\ref{Eq. accepted AO monotonicity}).
		It remains to show that the accepted rate sequence is bounded above. Let
		\begin{equation}
			\mathbf C^{(t)}=\mathbf M^{\frac12}\mathbf D^{(t)}\mathbf A^{(t)}.
		\end{equation}
		Since $\mathbf M\succ0$, $\det(\mathbf D^{(t)})=1$, and $|\det(\mathbf A^{(t)})|=1$, the matrix $\mathbf C^{(t)}$ is nonsingular. Let $\sigma_1^{(t)},\ldots,\sigma_K^{(t)}$ denote its singular values. The arithmetic--geometric mean inequality gives
		\begin{align}
			&\operatorname{Tr}\left(
			(\mathbf A^{(t)})^{\top}
			(\mathbf D^{(t)})^{\top}
			\mathbf M
			\mathbf D^{(t)}
			\mathbf A^{(t)}
			\right)
			=
			\left\|\mathbf C^{(t)}\right\|_F^2
			=
			\sum_{i=1}^{K}\left(\sigma_i^{(t)}\right)^2
			\geq
			K\left(
			\prod_{i=1}^{K}
			\left(\sigma_i^{(t)}\right)^2
			\right)^{\frac1K}
			=
			K\left|\det\left(\mathbf C^{(t)}\right)\right|^{\frac{2}{K}}
			=
			K\det(\mathbf M)^{\frac1K}.
			\label{Eq. AO objective lower bound}
		\end{align}
		The last equality follows from
		\begin{align}
			\left|\det\left(\mathbf C^{(t)}\right)\right|
			=
			\left|\det\left(\mathbf M^{\frac12}\right)\right|
			\det\left(\mathbf D^{(t)}\right)
			\left|\det\left(\mathbf A^{(t)}\right)\right|
			=
			\det(\mathbf M)^{\frac12}.
		\end{align}
		Using Eq.~(\ref{Eq. high SNR sum rate}) and Eq.~(\ref{Eq. AO objective lower bound}), we obtain
		\begin{equation}
			R\left(\mathbf A^{(t)},\mathbf D^{(t)}\right)
			\leq
			\frac K2
			\log_2^+\left(
			\frac{\rho}
			{K\det(\mathbf M)^{1/K}}
			\right).
			\label{Eq. AO rate upper bound}
		\end{equation}
		Hence, any infinite sequence of accepted achievable rates would be non-decreasing and bounded above, and would therefore converge to a finite limit. In the practical implementation, the iteration counter is bounded by $\mathrm{max\_iter}$, so the algorithm terminates in finite time and returns the last accepted pair.
	\end{proof}
\color{black}

\subsection{Multi-Cone Nested Stochastic Pattern Search (MCN-SPS)}
\begin{figure*}[t!]
	\centering
	\includegraphics[width=.9\textwidth]{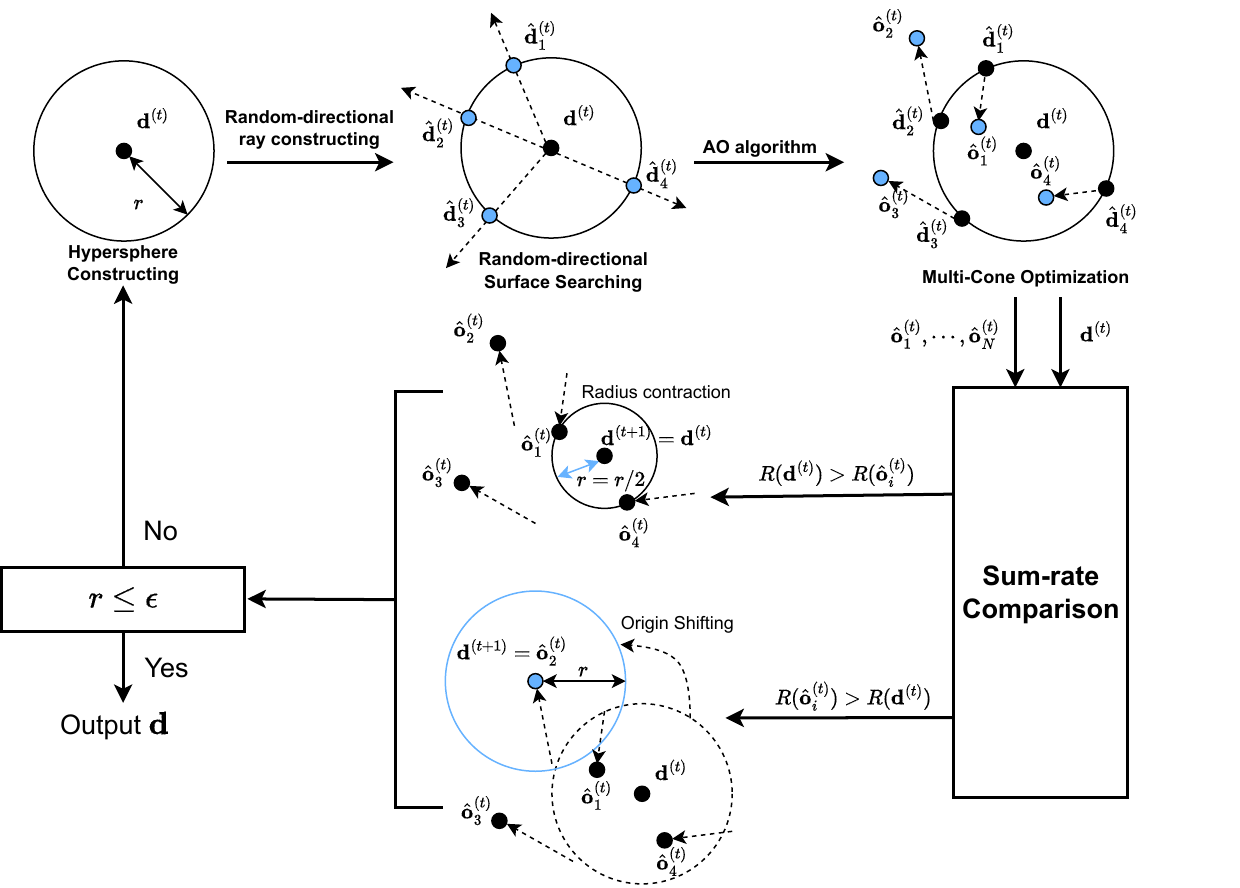}
	\caption{Flowchart of the MCN-SPS method}
	\label{Fig. MCNSPS Flowchart}
\end{figure*}

\textcolor{blue}{Alg.~\ref{Alg. Alternating Optimization} provides a rate-safeguarded refinement from each initial scaling vector. Since its practical coefficient update is generated by LLL, different initial points may lead to different accepted candidates, and no local or global optimality is implied. MCN-SPS therefore explores multiple initial points and compares the corresponding AJO-refined candidates across coefficient regions.
To explore and compare the candidates obtained from different initial points, we propose the Multi-Cone Nested Stochastic Pattern Search (MCN-SPS) method.}
Fig. \ref{Fig. MCNSPS Flowchart} illustrates the flowchart of the proposed method, which consists of three core processes: search preparation, search and optimization execution, and result integration and adjustment. These three processes are described in detail below.

\noindent \textbf{i) Search preparation: }By centering at a $\mathbf{d}^{(t)}\in\varOmega$ in the $t$-th iteration, a hypersphere $\mathcal{B}(\mathbf{d}^{(t)},r)$ is constructed with a search radius $r$. Starting from $\mathbf{d}^{(t)}$, $q$ rays are emitted along random directions, and cross the surface of $\mathcal{B}(\mathbf{d}^{(t)},r)$ at $\hat{\mathbf{d}}^{(t)}_1,\cdots,\hat{\mathbf{d}}^{(t)}_Q$. The $i$-th intersection $\hat{\mathbf{d}}^{(t)}_i$ can be calculated by
\begin{equation}
	\hat{\mathbf{d}}^{(t)}_i=\mathbf{d}^{(t)}+r\hat{\mathbf{v}}_i,
\end{equation}
where $\hat{\mathbf{v}}_i$ represents the unit vector in the direction of $i$-th ray. After the movement of $r\hat{\mathbf{v}}_i$, some $\hat{\mathbf{d}}^{(t)}_i$ may leave the positive orthant of $\mathcal{D}$. To limit in $d_i>0, i\in\{1,\cdots,K\}$, the escaped $\hat{\mathbf{d}}^{(t)}_i$ will be abandoned, and we choose the new random search directions until all vector are limited in the positive orthant of $\mathcal{D}$.

\noindent \textbf{ii) Search and Optimization execution: }
The sampled points are normalized by $\Gamma(\cdot)$ and processed in parallel by Alg.~\ref{Alg. Alternating Optimization}. The last accepted outputs are denoted by $\hat{\mathbf o}^{(t)}_i$, $i\in\{1,\ldots,Q\}$. By the rate-based safeguard, each returned candidate has an achievable rate no smaller than that of its corresponding input point. Coincident candidates are merged before the subsequent comparison.

\noindent \textbf{iii) Result integration and adjustment: }
The current center and the AJO-refined candidates are compared in terms of achievable rate.
\begin{itemize}
	\item[i)] If an AJO-refined candidate achieves the largest rate among the evaluated points, the best candidate is selected as the center of the next iteration.
	\item[ii)] Otherwise, the current center is retained and the search radius is reduced to $r/2$.
\end{itemize}

\begin{algorithm}[t]
	\caption{MCN-SPS method}
	\label{Alg. MCN-SPS}
	\hspace{0.02in} {\bf Input:}
	Initial power allocation matrix $\mathbf{D}_{\mathrm{init}}$, channel matrix $\mathbf{H}$, SNR $\rho$, Initial step size $r_0$ \\
	\hspace{0.02in} {\bf Output:}
	The best evaluated coefficient--scaling pair $(\mathbf A,\mathbf D)$.
	\begin{algorithmic}[1]
		\State \textcolor{blue}{$(\mathbf A,\mathbf D)\gets\mathrm{AJO}(\mathbf H,\mathbf D_{\mathrm{init}},\rho)$}
		\State $r\gets r_0$
		\State $(\mathbf A_{\mathrm{tmp}},\mathbf D_{\mathrm{tmp}})\gets(\mathbf A,\mathbf D)$
		\While{$r>\epsilon$}
		\State \textcolor{blue}{$(\mathbf A,\mathbf D)\gets\mathrm{AJO}(\mathbf H,\mathbf D,\rho)$}
		\State $(\mathbf A_{\mathrm{tmp}},\mathbf D_{\mathrm{tmp}})\gets(\mathbf A,\mathbf D)$
		\For{$i=1:\text{max\_router}$} \Comment{parallel processing}
		\State $\mathbf{v}\gets\mathrm{randn}(1,K)$
		\State $\mathbf{v}\gets\frac{\mathbf{v}}{\|\mathbf{v}\|}$
		\State $\hat{\mathbf{D}}_i\gets\mathrm{diag}(\mathrm{diag}+r\mathbf{v})$
		\State \textcolor{blue}{$(\hat{\mathbf A}_i,\hat{\mathbf O}_i)\gets\mathrm{AJO}(\mathbf H,\hat{\mathbf D}_i,\rho)$}
		\If{$R(\hat{\mathbf A}_i,\hat{\mathbf O}_i)>R(\mathbf A_{\mathrm{tmp}},\mathbf D_{\mathrm{tmp}})$}
		\State $\mathbf{D}_{\mathrm{tmp}}\gets\hat{\mathbf{O}}_i$
		\State $\mathbf{A}_{\mathrm{tmp}}\gets\hat{\mathbf{A}}_i$
		\EndIf
		\EndFor
		\If{$\mathbf{D}_{\mathrm{tmp}}==\mathbf{D}$}
		\State $r\gets\frac{r}{2}$
		\EndIf
		\State $\mathbf{D}\gets\mathbf{D}_{\mathrm{tmp}}$
		\State $\mathbf{A}\gets\mathbf{A}_{\mathrm{tmp}}$
		\EndWhile
	\end{algorithmic}
\end{algorithm}

\textcolor{blue}{The iteration terminates when the radius is smaller than the predefined threshold $\epsilon$. The overall procedure is summarized in Alg.~\ref{Alg. MCN-SPS}. MCN-SPS compares the rate-safeguarded AJO outputs obtained from the evaluated sample points and does not assume that these candidates exhaust the solution set characterized in Theorem~\ref{Pro. solution set}.}

\section{Computational Complexity Analysis}
\textcolor{blue}{Let $T_{\Psi}$ denote the complexity of the adopted routine for the unimodular coefficient update. The remaining complexity of MCN-SPS consists of the nested AJO and RA procedures and the outer radius refinement. Since the radius is halved after an unsuccessful search, the number of radius levels is $\mathcal O(\log_2(r_0))$. When LLL is adopted, $T_{\Psi}$ is replaced by the corresponding LLL complexity.}

\noindent \textbf{i) MCN-SPS method (Alg. \ref{Alg. MCN-SPS}): }The whole steps of MCN-SPS method consist of the initialization and iterative steps, which is given by
\begin{equation}
	T_{\mathrm{MCN-SPS}}=T_{\mathrm{init},\ref{Alg. MCN-SPS}}+T_{\mathrm{iter},\ref{Alg. MCN-SPS}}.
\end{equation}
In initialization step, the each assignment is with complexity $\mathcal{O}(1)$, except the initial $\mathbf{A}$ obtaining and the initial $\mathbf{D}$ obtaining with $\mathcal{O}(T_{\mathrm{AJO}})$. The whole complexity in initialization step is given by 
\begin{equation}
	T_{\mathrm{init},\ref{Alg. MCN-SPS}}=\mathcal{O}(T_{\Psi})+\mathcal{O}(T_{\mathrm{AJO}}).
\end{equation}
Regarding the iterative steps, since the radius is reduced doubly, the iterative number of the shrinkage step is approach to $\log_2(r_0)$. And we denote the iterative number of center update as $\mathcal{O}(T_{\mathrm{os}})$ which can be regarded as a constant. Inside the iterations, since the operations of each random-directional ray execute in parallel, we ignore the iterative number of the parallel processing. To Eq. (\ref{Eq. high SNR sum rate}), it consists of a inverse calculation for $\mathbf{M}$ obtaining. Thus, the whole complexity in iterative steps is given by
\begin{equation}
	T_{\mathrm{iter},\ref{Alg. MCN-SPS}}=\left(\mathcal{O}(T_{\mathrm{AJO}})+\mathcal{O}(T_{\Psi})+\mathcal{O}(K^3)\right)\left[\log_2(r_0)+\mathcal{O}(T_{\mathrm{os}})\right].
\end{equation}

\noindent \textbf{ii) AJO (Alg. \ref{Alg. Alternating Optimization}): }
The whole steps of AJO algorithm consist of the initialization and iterative steps, which is given by
\begin{equation}
	T_{\mathrm{AJO}}=T_{\mathrm{init},\ref{Alg. Alternating Optimization}}+T_{\mathrm{iter},\ref{Alg. Alternating Optimization}}.
\end{equation}
In initialization step, the each assignment is with complexity $\mathcal{O}(1)$, besides once operation for $\mathbf{M}$ and $\mathbf{A}$ obtaining. Thus, we have
\begin{equation}
	T_{\mathrm{init},\ref{Alg. Alternating Optimization}}=\mathcal{O}(K^3)+\mathcal{O}(T_{\Psi}).
\end{equation}
Regarding the iterative steps, the whole iterative number is constrained under a set constant value $\text{max\_iter}$. 
\textcolor{blue}{Each AJO iteration consists of one RA update, one coefficient update, and one achievable-rate evaluation.}
So $T_{\mathrm{iter},\ref{Alg. Alternating Optimization}}$ can be calculated by
\textcolor{blue}{
\begin{align}
	T_{\mathrm{iter},\ref{Alg. Alternating Optimization}}&=T_{\ref{Alg. Alternating Optimization}}\left(\mathcal{O}(T_{\mathrm{RA}}+T_{\Psi}+K^{3})\right)
	\leq\text{max\_iter}\times\,\mathcal O\!\left(T_{\mathrm{RA}}+T_{\Psi}+K^{3}\right).
\end{align}}
\noindent \textbf{iii) RA (Alg. \ref{Alg. Reciprocal Approximation}): }Regarding the initialization step in Alg. \ref{Alg. Reciprocal Approximation}, it contains once operation on $\mathbf{M}$ obtaining and $\mathbf{D}$ initialing, where the $\mathbf{D}$ initialing consist of the inverse calculation with complexity $\mathcal{O}(K^3)$ in average. Thus, the complexity of initialization step can be calculated by
\begin{equation}
	T_{\mathrm{init},\ref{Alg. Reciprocal Approximation}}=2\mathcal{O}(K^3).
\end{equation}
Outside the iterative steps, there is a theorem on the upper bound of iterative number.
\begin{theorem}\label{The. upper bound of RA iterative number}
	Given a tolerance \(\epsilon\), an interference-coupling matrix \(\mathbf G\), and an initialization \(\mathbf d^{(0)}\), the number of RA iterations \(t_{\mathrm{RA}}\) of Alg. \ref{Alg. Reciprocal Approximation} satisfies
	\begin{equation}\label{Eq. upper bound of tRA}
		t_{\mathrm{RA}}\leq\frac{\log\epsilon-\log\left(d_{H}(\mathbf{G}^{-1}(\mathbf{1}_K\oslash\mathbf{d}^{(0)}),\mathbf{d}^{(0)})\right)}
		{\log \left(\tanh\left(\frac{\Delta(\mathbf{G}^{-1})}{4}\right)\right)},
	\end{equation}
	where
	\begin{equation}\label{Eq. projective diameter}
		\Delta(\mathbf{G}^{-1})=\log\left(\max_{i,j,p,q}\frac{g^{'}_{p,i}g^{'}_{q,j}}{g^{'}_{p,j}g^{'}_{q,i}}\right),
	\end{equation}
	and $g^{'}_{i,j},~i,j\in\{1,\cdots,K\}$ represents the $\{i,j\}$ element of $\mathbf{G}^{-1}$.
\end{theorem}
\begin{proof}
	See Appendix \ref{App. upper bound of RA iterative number proof}.
\end{proof}
Regarding the iterative steps, the whole process in each iteration consists a inverse calculation with complexity $\mathcal{O}(K^3)$ to apply Eq. (\ref{Eq. iterative function}) and the others with complexity $\mathcal{O}(K)$. Combining with Theorem \ref{The. upper bound of RA iterative number}, the whole complexity of RA algorithm can be calculated by
\begin{align}
	T_{\mathrm{RA}}&=2\mathcal{O}(K^3)+t_{\mathrm{RA}}\left(\mathcal{O}(K^3)+4\mathcal{O}(K)\right).
\end{align}

In combination with the above, since $T_{\ref{Alg. Alternating Optimization}}$ and $T_{\mathrm{os}}$ can be regarded as the constant, we can derive the complexity of our proposed method from
\begin{align}
	&T_{\mathrm{MCN-SPS}}\notag\\
	=&\mathcal{O}\left((T_{\mathrm{AJO}}+T_{\Psi}+K^3+1)(\log_2(r_0)+T_{\mathrm{os}})\right)\notag\\
	=&\mathcal{O}\left((T_{\ref{Alg. Alternating Optimization}}(T_{\mathrm{RA}}+T_{\Psi})+T_{\Psi}+K^3+1)(\log_2(r_0)+T_{\mathrm{os}})\right)\notag\\
	=&\mathcal{O}\left(T_{\ref{Alg. Alternating Optimization}}(T_{\Psi}+t_{\mathrm{RA}}K^3)(\log_2(r_0)+T_{\mathrm{os}})\right)\notag\\
	=&\mathcal{O}\left((T_{\Psi}+t_{\mathrm{RA}}K^3)\log_2(r_0)\right).
\end{align}

\textcolor{blue}{When LLL is adopted as the practical coefficient-update routine with complexity $\mathcal O\left(K^4\log K\right)$ \cite{nguyen2010lll}, the proposed method has complexity $\mathcal O\left(K^4\log K\log_2(r_0)\right)$.}
When $K$ is sufficiently large 
such that $K\gg r_0$, $\mathcal{O}\left(\log_2(r_0)\right)$ is treated as a constant, and the complexity of MCN-SPS method is given by $\mathcal{O}\left(K^4\log K\right)$. In Tab. \ref{Tab. Complexity comparison}, we compare the MCN-SPS method, Venturelli's method \cite{venturelli2020optimization} and PSO-based method \cite{qiu2024lattice} in three cases: the average complexity, average Complexity with LLL algorithm, and $K\gg r_0$ with LLL algorithm. The observations reveal that the MCN-SPS method has a lower complexity in all three cases compared to PSO-based method \cite{qiu2024lattice}, and matches the stated complexity of Venturelli's method when $K\gg r_0$.

\begin{table*}[t!]
	\centering
	\renewcommand\arraystretch{1.5}
	\caption{Theoretical Time Complexity under Different Settings}
	\label{Tab. Complexity comparison}
	\begin{tabular}{cccc}
		\hline
		Method              & Average Complexity & \begin{tabular}[c]{@{}c@{}}Average Complexity\\ with LLL algorithm\end{tabular} & \begin{tabular}[c]{@{}c@{}}$K\gg r_0$\\ with LLL algorithm\end{tabular}  \\
		\hline
		MCN-SPS method      &  $\mathcal{O}\left(\log_2(r_0)T_{\Psi}\right)$                  &  $\mathcal{O}\left(K^4\log K\log_2(r_0)\right)$        &        $\mathcal{O}\left(K^4\log K\right)$\\
		
		Venturelli's method \cite{venturelli2020optimization} &  $\mathcal{O}\left(T_{\Psi}\right)$                  &  $\mathcal{O}\left(K^4\log K\right) $      &  $\mathcal{O}\left(K^4\log K\right) $   \\
		
		PSO-based method \cite{qiu2024lattice}   & $\mathcal{O}\left(K^2T_{\Psi}\right)$ &  $\mathcal{O}\left(K^6\log K\right)$    &   $\mathcal{O}\left(K^6\log K\right)$ \\   
		\hline
	\end{tabular}
\end{table*}

\section{Numerical Evaluation}
This section evaluates the feasibility, computational cost, and empirical achievable-rate performance of MCN-SPS.

\begin{figure*}[t!]
	\centering
	\subfigure[]{\includegraphics[width=.45\textwidth]{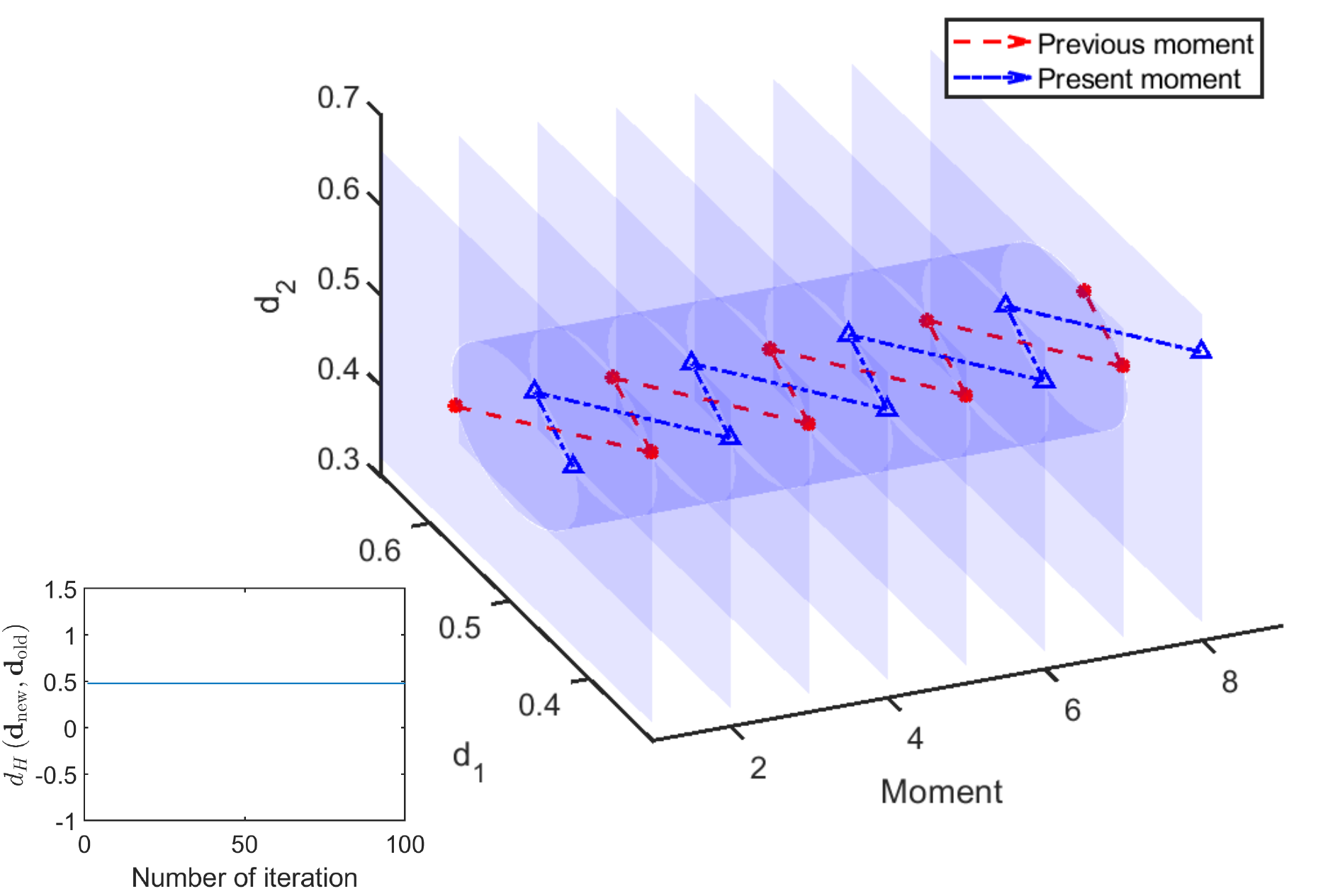}}\quad
	\subfigure[]{\includegraphics[width=.45\textwidth]{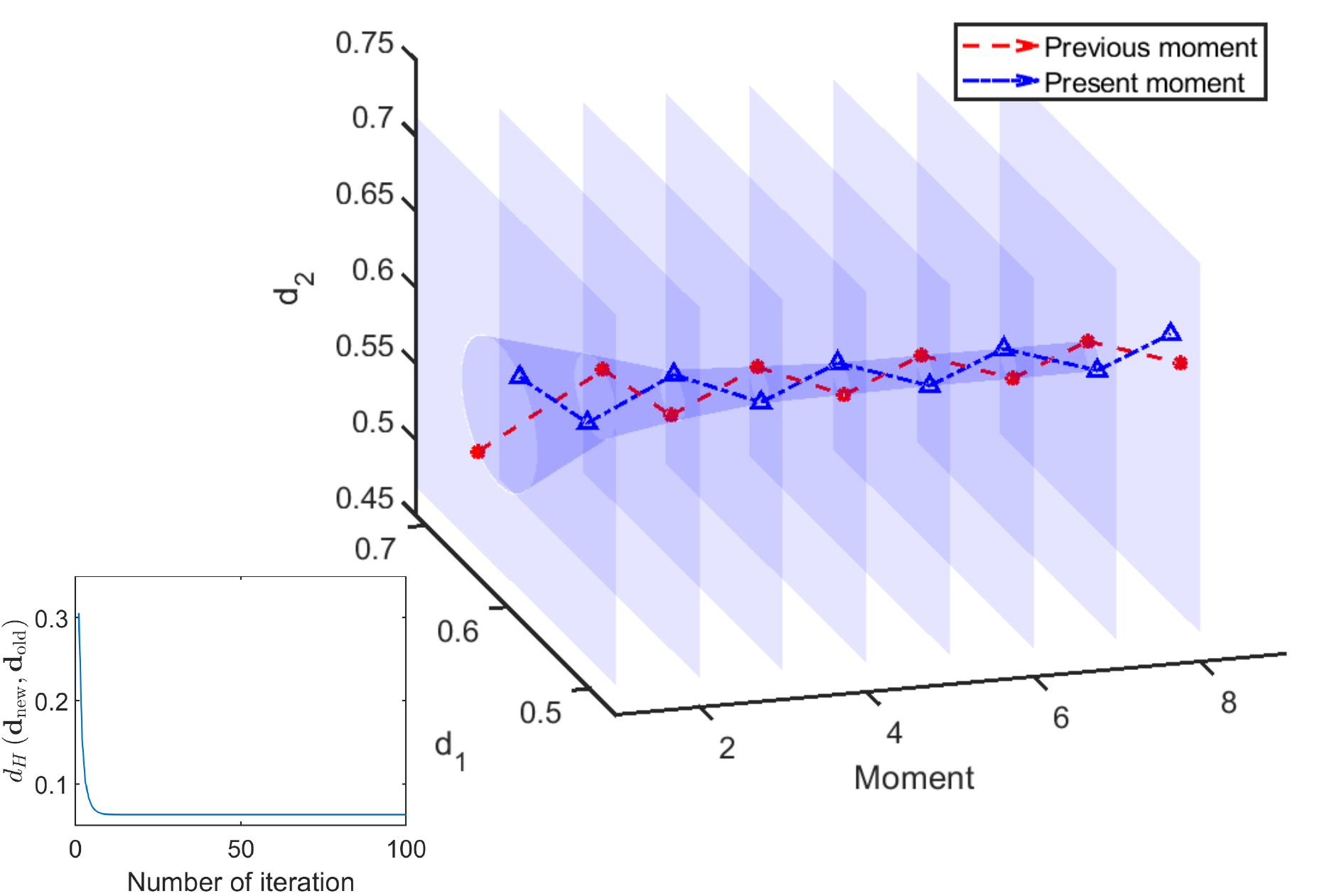}}\quad
	\subfigure[]{\includegraphics[width=.45\textwidth]{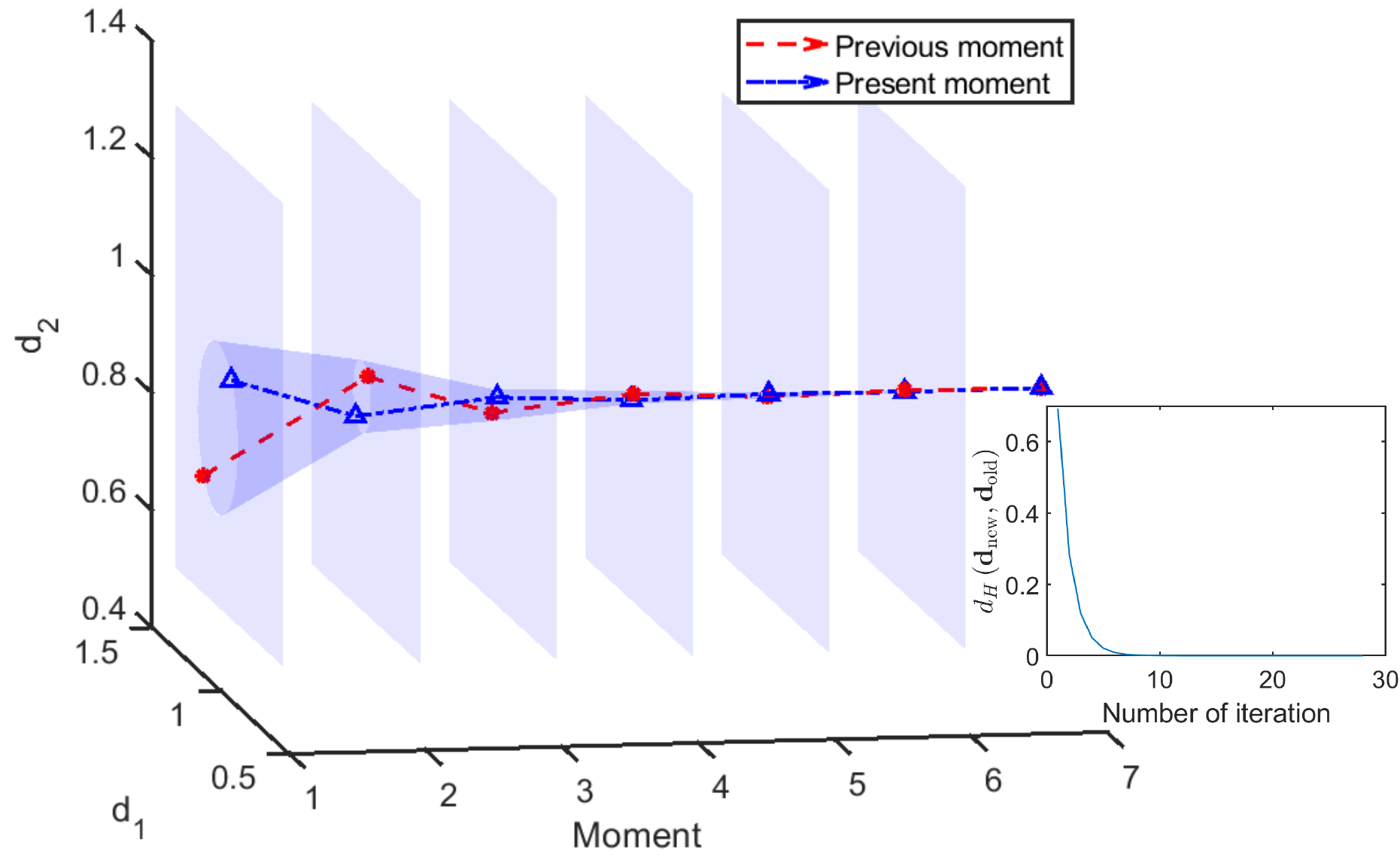}}
	
	\caption{\textcolor{blue}{Trajectories of $\mathbf d$ and the corresponding Hilbert distances across iterations for an $N=K=3$ sample channel at (a) SNR$=0$ dB, (b) SNR$=5$ dB, and (c) SNR$=10$ dB.}}
	\label{Fig. Convergence verification}
\end{figure*}

\begin{figure*}[htbp]
	\centering
	\subfigure[]{\includegraphics[width=.4\textwidth]{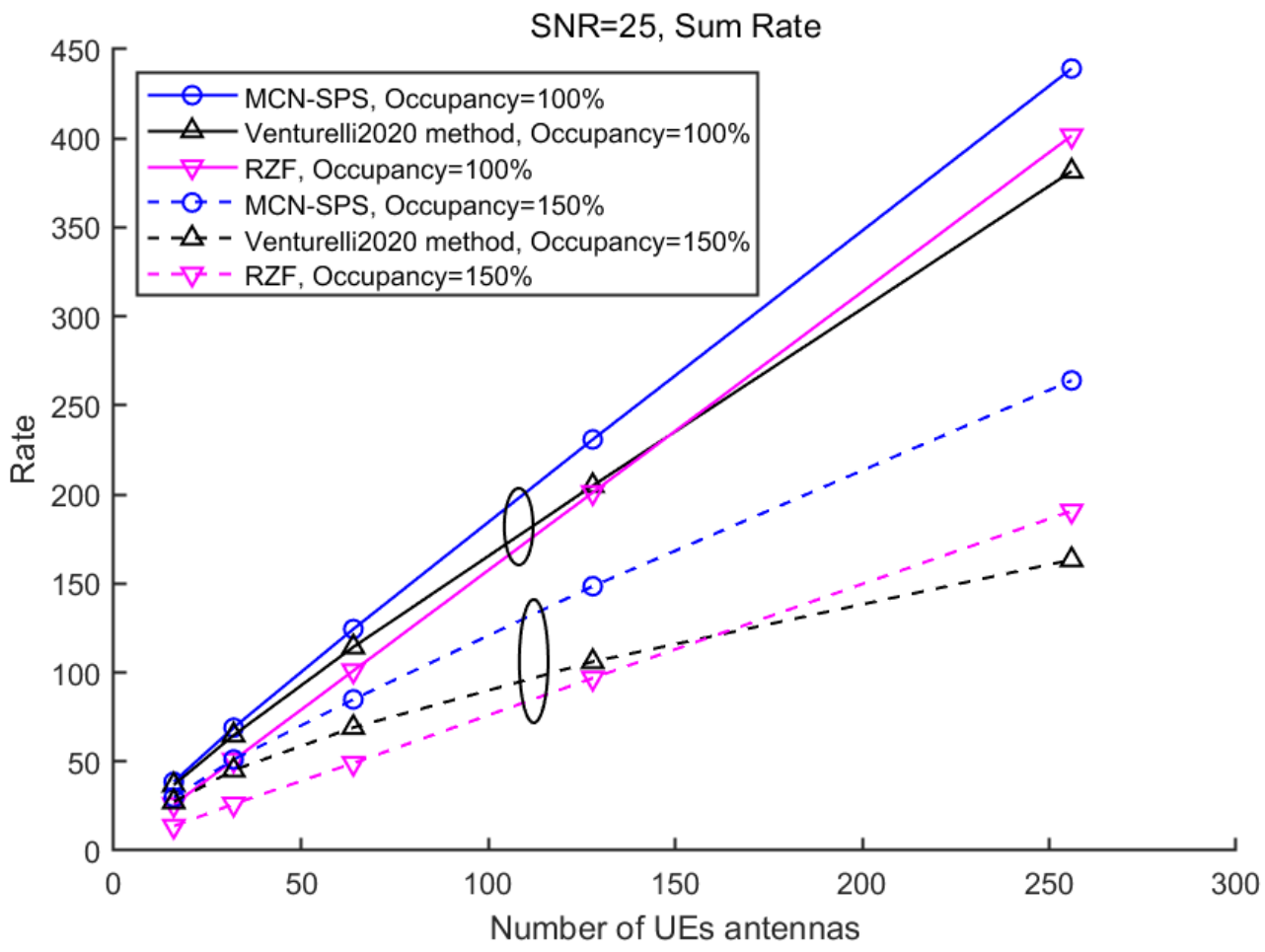}}
	\subfigure[]{\includegraphics[width=.4\textwidth]{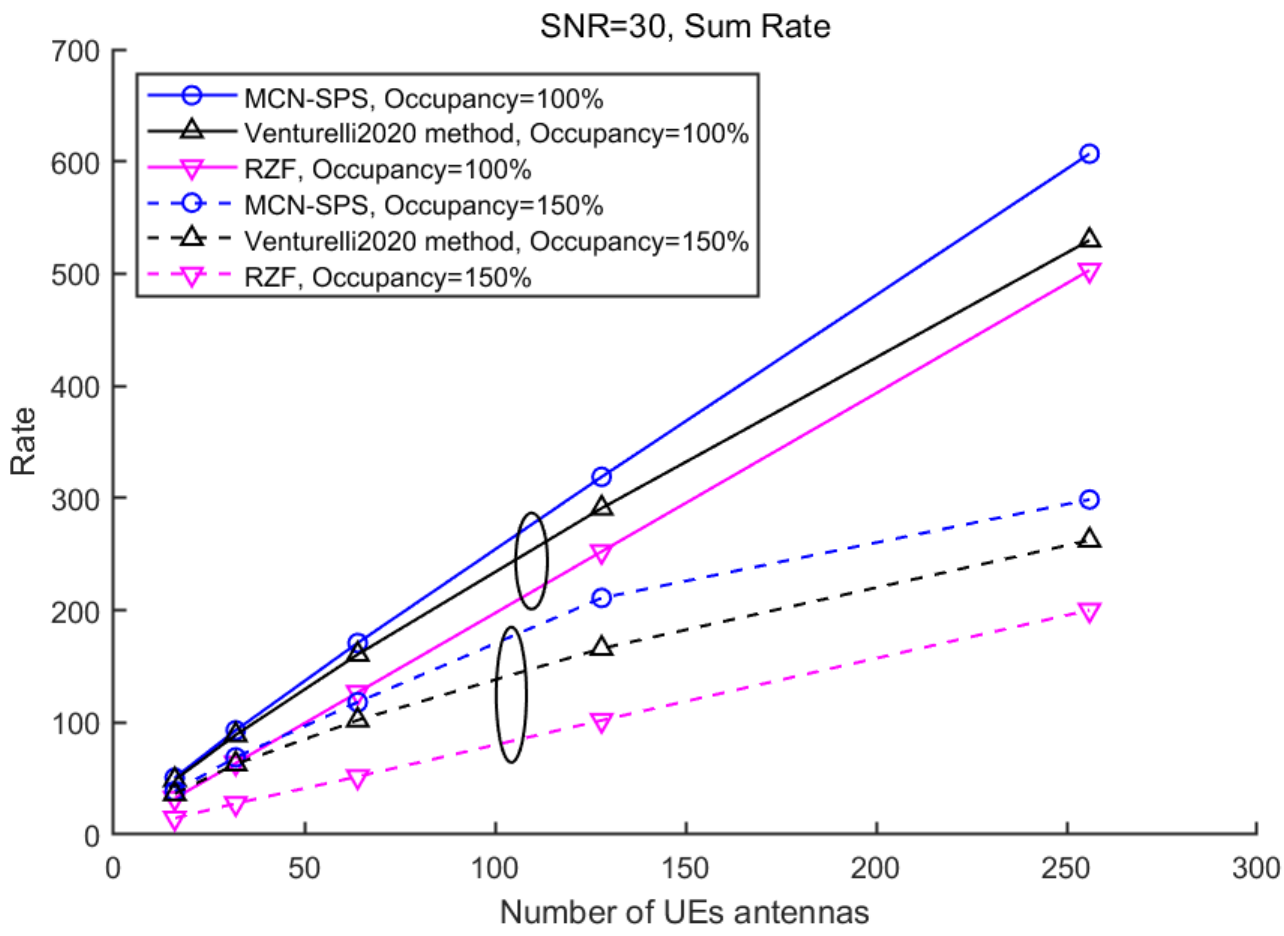}}
	\caption{Sum rate comparisons in different occupancy ($100\%$ and $150\%$) and $K$ ($16$ to $256$) with SNR: (a) $25$dB, (b) $30$dB.}
	\label{Fig. N v.s. Rate_256}
\end{figure*}

\begin{figure*}[htbp]
	\centering
	\subfigure[]{\includegraphics[width=.35\textwidth]{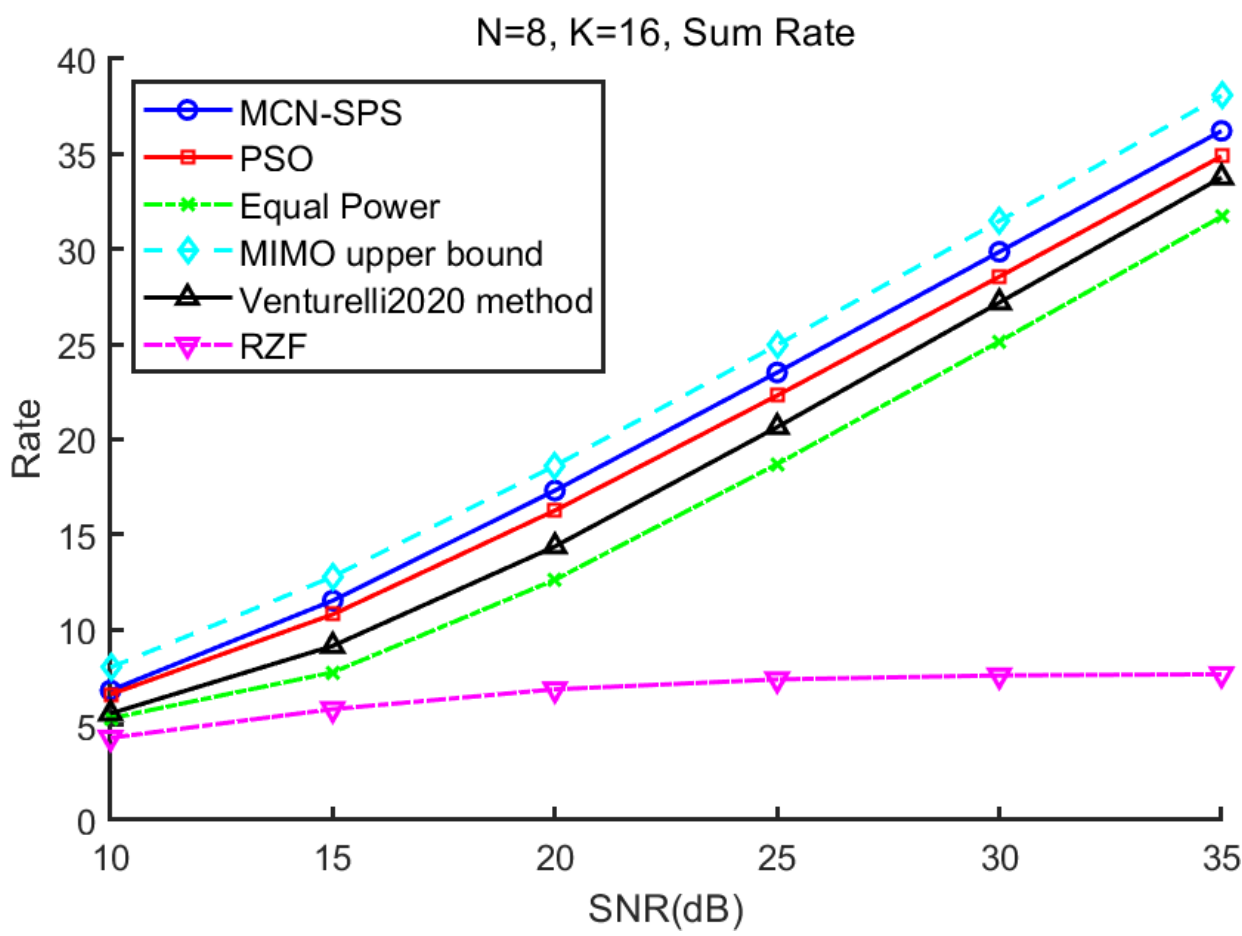}}
	\subfigure[]{\includegraphics[width=.35\textwidth]{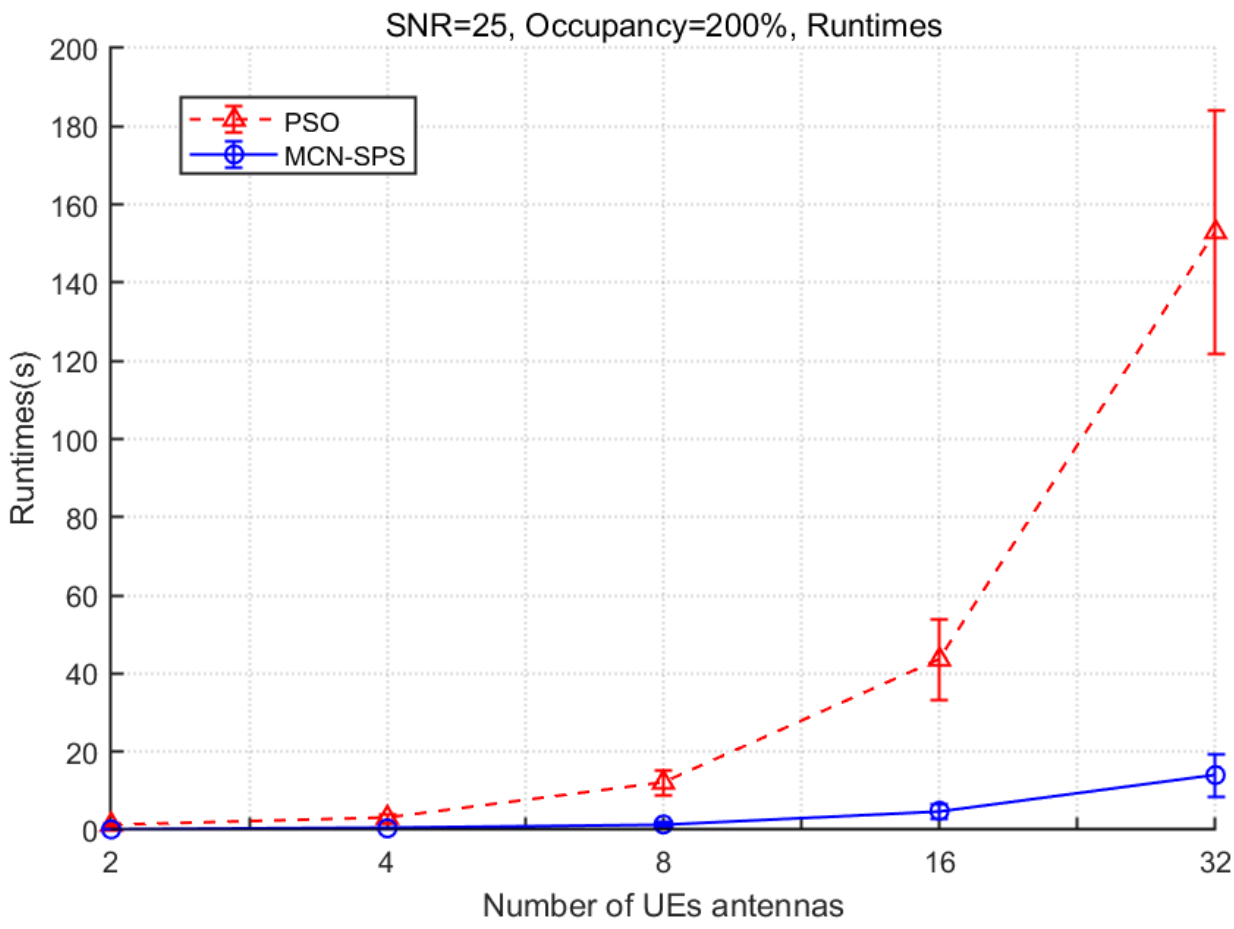}}
	
	
	\subfigure[]{\includegraphics[width=.35\textwidth]{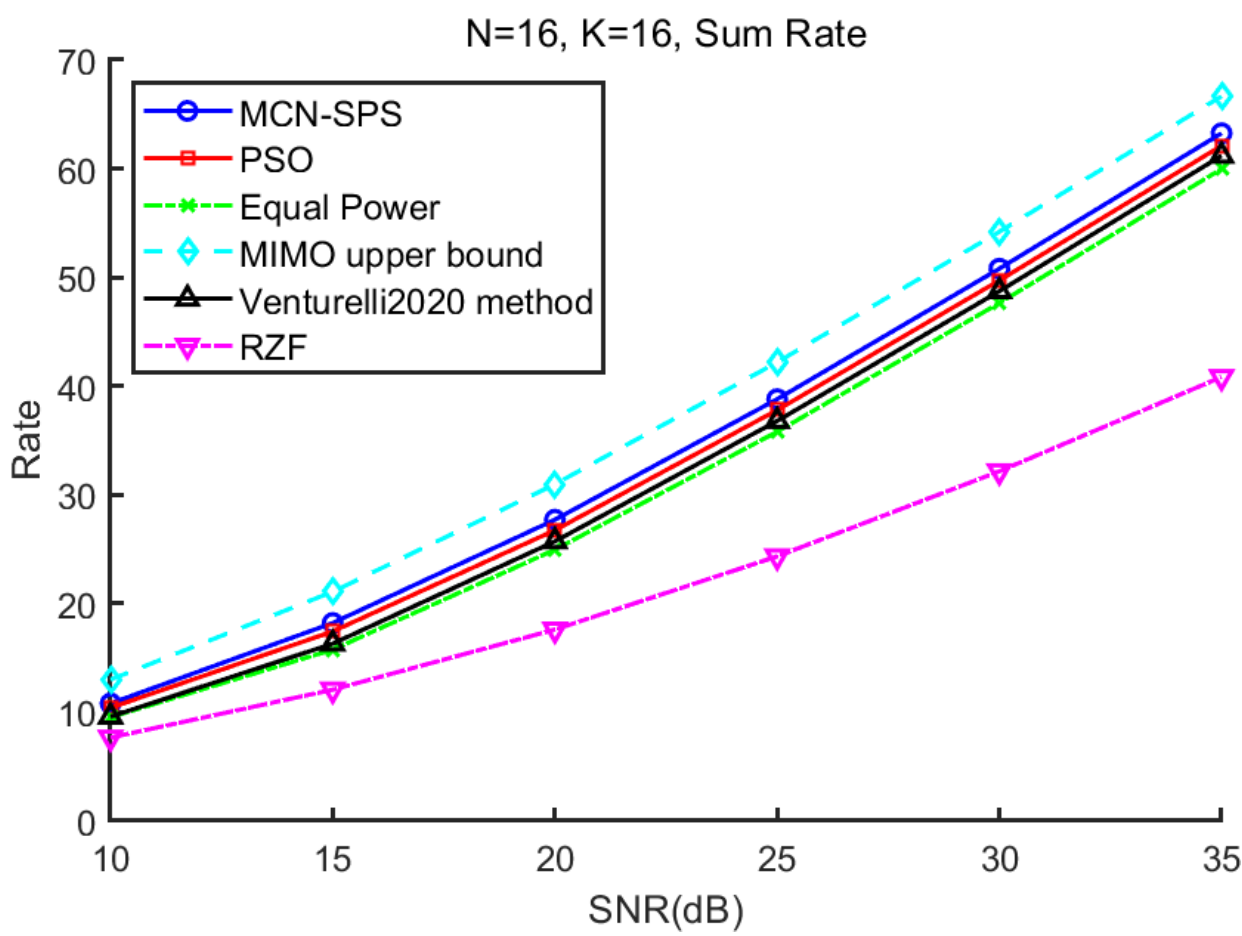}}
	\subfigure[]{\includegraphics[width=.35\textwidth]{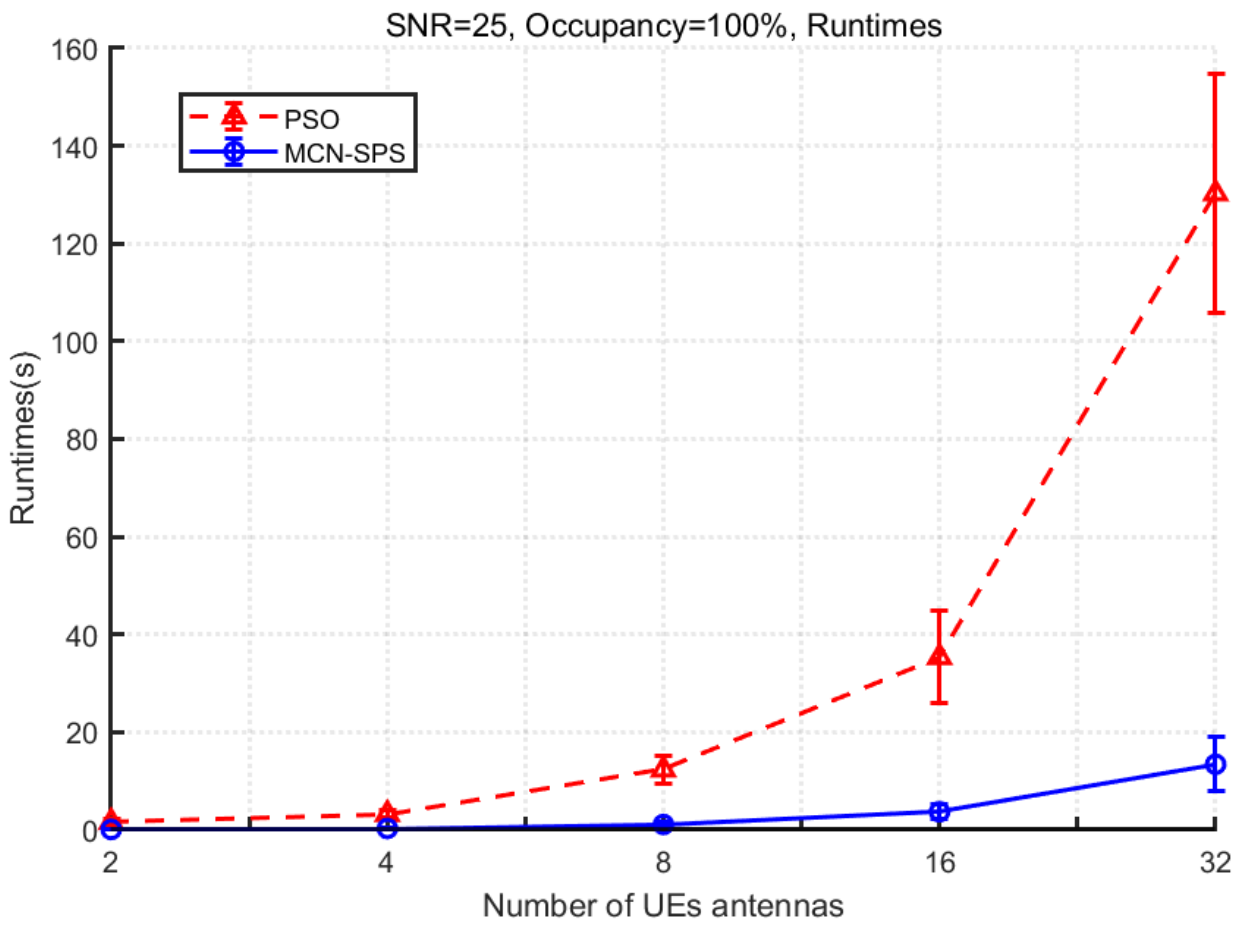}}
	\caption{Sum rate comparisons in different SNR with: (a) $N=8,~K=16$, (c) $N=K=16$, and runtime from $K=2$ to $32$ are shown in (b), (d) respectively.}
	\label{Fig. SNR v.s. Rate}
\end{figure*}


\noindent \textbf{Metrics: }
\textcolor{blue}{The sum rate and runtime are used as the performance and complexity metrics, respectively. Each plotted point is the sample mean over $N_{\rm MC}=1000$ independently generated channel realizations. All methods are evaluated on the same realizations at each operating point, and the curves are plotted directly from these means without smoothing or interpolation.} The sum rate of each realization is calculated by Eq.~(\ref{Eq. high SNR sum rate}).
At a fixed point, the Hilbert residual between two consecutive iterates is zero.
When the processing is converged into a fixed point, this metric is zero due to the same from input to output. In addition, we define the occupancy of the MIMO system by 
\begin{equation}
	\text{Occupancy}=\frac{K}{N}\times100\%,
\end{equation}
in which the situation of $\text{Occupancy}>100\%$ is referred to as overload MIMO.

\noindent \textbf{Benchmark and setting: }In the comparisons, we adopt PSO-based method \cite{qiu2024lattice}, the \cite{venturelli2020optimization}'s method , the IF precoding without $\mathbf{D}$ optimization, RZF as the benchmarks, referring to as "PSO", "Venturelli", "Equal Power" and "RZF" respectively. 
\textcolor{blue}{LLL is adopted to generate the unimodular coefficient candidate in the practical implementation.}
In PSO, the maximum of iteration and population size are set as $300$ and $50$ respectively, and the velocity is within $[-10^{-1},10^{-1}]$. Considering Equal Power and RZF, the power allocation matrix $\mathbf{D}$ is both set as $\mathbf{I}$. To MCN-SPS method, we employ RIF scheme, and set $\text{max\_iter}$ in Alg. \ref{Alg. Alternating Optimization} as $100$ and $r_0$ in Alg. \ref{Alg. MCN-SPS} as $10$. 
\textcolor{blue}{In each AJO iteration, a candidate producing a lower achievable rate is rejected, and the last accepted pair is returned.}
\textcolor{blue}{All runtime measurements were performed using MATLAB R2023a on Ubuntu 22.04 on a server equipped with two Intel Xeon 8582C processors and 128 GB of RAM. Both MCN-SPS and PSO used double-precision arithmetic and the same 50-worker parallel execution setting. For a fair comparison, both implementations were vectorized and preallocated, while plotting, file I/O, and diagnostic output were disabled during timing. Runtime was measured from algorithm initialization to the generation of the final $(\mathbf A,\mathbf D)$ pair, excluding channel generation, and is reported in seconds.} We calculate the upper bound of MIMO system by 
\begin{equation}
	R_{\text{ub}}=\frac{1}{2}\log\det\left(\mathbf{I}+\frac{\rho}{K}\mathbf{H}\mathbf{W}\mathbf{H}^{\top}\right),
\end{equation}
where $\mathbf{W}$ represents the $K\times K$ diagonal matrix outputted by the famous water-filling algorithm \cite{tse2005fundamentals}.

\subsection{Convergence verification}
To verify the convergence, Fig. \ref{Fig. Convergence verification} illustrates the trajectory of the moment vector $\mathbf{d}$ in $\varOmega$ under a sample channel matrix $\mathbf{H}$, along with the corresponding trend of the Hilbert metric across iterations, for different SNR values.
At SNR$=0$dB, a period‑2 oscillation is clearly observed in Fig. \ref{Fig. Convergence verification}(a): $\mathbf{d}^{(t)}$ jumps between only two positions in $\varOmega$, while the Hilbert distance $d_{H}(\mathbf{d}^{(t+1)},\mathbf{d}^{(t)})$ remains constant. Geometrically, the cavity in this case appears as a cylinder.
In Fig. \ref{Fig. Convergence verification}(b), a similar oscillatory behavior persists, but the two locations gradually converge toward each other. This results in a cylindro‑conical cavity shape and is reflected in the Hilbert metric, which decreases yet converges to a non‑zero value.
As the SNR increases further, shown in Fig. \ref{Fig. Convergence verification}(c), the cavity takes on a tapered conical form, indicating that the iteration eventually converges. Correspondingly, $d_{H}(\mathbf{d}^{(t+1)},\mathbf{d}^{(t)})$ converges to zero.
Overall, these observations are consistent with Lemma \ref{Lem. LLL mu in M matrix}.

\subsection{Performance comparison}

To address the massive communication and ubiquitous application scenarios envisioned for 6G, we conduct a series of comparative experiments to validate the effectiveness and advantages of MCN-SPS in large-scale MIMO systems. Fig. \ref{Fig. N v.s. Rate_256} presents the results under different occupancy levels ($100\%$ and $150\%$) and SNR values ranging from $15$dB to $30$dB. Since the PSO method becomes inapplicable when $K\geq32$, we compare the sum-rate performance of MCN-SPS, Venturelli's method \cite{venturelli2020optimization}, and RZF for various numbers of user antennas $K$ (from $16$ to $256$). The following observations can be drawn from Fig. \ref{Fig. N v.s. Rate_256}:
\begin{itemize}
	\item [i)] Across all considered values of $K$ and occupancy levels, MCN-SPS achieves a higher sum rate than the considered benchmarks in all evaluated settings. Moreover, the sum rate gain of MCN-SPS increases as the occupancy level rises, highlighting its applicability in large-scale MIMO systems and its robustness in overloaded MIMO scenarios.
	\item[ii)] As $K$ grows, the performance gap between MCN-SPS and RZF stabilizes, with their curves becoming nearly parallel. In contrast, the performance of Venturelli’s method gradually falls below that of RZF, leading to an intersection between their curves. This indicates that a jointly optimized $(\mathbf{A},\mathbf{D})$ pair is crucial for the effectiveness of IF precoding in large-scale MIMO.
	\item[iii)] With increasing SNR, the performance gap between MCN-SPS and RZF widens further, while the intersection point between Venturelli’s method and RZF shifts toward smaller $K$. These trends suggest that relaxation-based joint optimization methods for IF precoding exhibit reduced applicability in systems with many users or lower SNR. Such limitations are not observed with MCN-SPS, demonstrating the strong applicability of the MCN-SPS-based IF precoder.
\end{itemize}

Regarding the performance at different SNR values, we establish two comparisons on sum rate and runtime between the proposed method and benchmarks with $N=8,~K=16$ and $N=K=16$. The results are depicted in Fig. \ref{Fig. SNR v.s. Rate}, and we make the following observations.
\begin{itemize}
	\item [i)] In overload MIMO depicted in Fig. \ref{Fig. SNR v.s. Rate}(a), RZF fails to work due to the non-eliminated multi-user interference (MUI). But IF based methods (MCN-SPS, PSO, Venturelli, Equal Power) still work, which reveals RIF has the effectiveness to reduce MUI.
	\item[ii)] Compared to PSO, Venturelli, Equal Power and RZF, our proposed method achieve a higher sum rate in all SNR$\geq10$dB in Fig. \ref{Fig. SNR v.s. Rate}(a) and (c), which gain $1$dB when $N=K=16$ and $2$dB when $N=8,~K=16$ compared to PSO method. Meanwhile, the proposed method run in a nearly double lower time compared to PSO method.
	\item[iii)] The gap between the MIMO upper bound and the proposed method shrinks as occupancy increases, which reveals the applicability of MCN-SPS method in overload MIMO.
	\item[iv)] From Fig. \ref{Fig. SNR v.s. Rate}(b) and Fig. \ref{Fig. SNR v.s. Rate}(d), MCN-SPS achieves an nearly doubled runtime decreasing compared to PSO, which reflects the lower complexity of MCN-SPS method.
\end{itemize}
\color{blue}
\section{Conclusion}

This paper developed a coefficient-labelled geometric formulation for the unimodularly restricted integer-forcing precoder design. The exact single-output coefficient mapping organizes the normalized scaling domain into coefficient-dependent regions. The resulting analysis identifies the NP-hard coefficient-evaluation bottleneck and shows that the exact self-consistent solution set is finite and nonempty. Guided by this structure, MCN-SPS combines stochastic cross-region exploration with rate-safeguarded alternating refinement. Its practical implementation uses reciprocal approximation for the scaling update and LLL to generate a feasible unimodular coefficient candidate. The acceptance rule guarantees non-decreasing accepted rates, but does not imply local or global optimality of the returned pair. Under the adopted arithmetic-complexity model, the LLL-based implementation has complexity $\mathcal O\!\left(K^4\log K\log_2(r_0)\right)$. Numerical results show a favorable empirical performance--complexity tradeoff in the considered MIMO configurations.
\color{black}
\appendices

\section{Proof of Theorem \ref{The. similar lattice reduction}}\label{App. similar lattice reduction proof}
We prove Theorem \ref{The. similar lattice reduction} by contradiction. First, we assume the original proposition does not hold under the original background, i.e., when $\|\mathbf{D}_{A}-\mathbf{I}\|_{2}<\epsilon$ for a minimum value $\epsilon>0$, 
\begin{equation}\label{Eq. Appendix B proposition}
	\forall \mathbf{R}_{B}=\mathbf{G}_{B}\mathbf{U}\in \mathcal{S}_B,~\forall\mathbf{R}_{A}=\mathbf{G}_{A}\mathbf{V}\in \mathcal{S}_A,~\text{s.t.},~\mathbf{U}\neq\mathbf{V}.
\end{equation}
Since $\mathbf{R}_{A}$ and $\mathbf{R}_{B}$ are the SIVP solutions of $\mathbf{G}_{A}$ and $\mathbf{G}_{B}$ respectively, based on Definition \ref{Def. Successive Minima}, their basis vectors satisfy
\begin{align}
	\|\mathbf{R}_{A}\mathbf{e}_i\|_2&=\|\mathbf{G}_{A}\mathbf{V}\mathbf{e}_i\|_2\leq\lambda_N(\Lambda_{A}),\\
	\|\mathbf{R}_{B}\mathbf{e}_i\|_2&=\|\mathbf{G}_{B}\mathbf{U}\mathbf{e}_i\|_2\leq\lambda_N(\Lambda_{B}),~i\in\{1,\cdots,K\},
\end{align}
where $\mathbf{e}_i$ denotes the $i$-th column of identity matrix $\mathbf{I}$. Since $\mathbf{U}$ is the unimodular matrix, $\mathbf{G}_{A}\mathbf{U}$ is also a generator matrix of $\Lambda_{A}$. Due to $\mathbf{G}_{B}=\mathbf{G}_{A}\mathbf{D}_{A}$, we have
\begin{align}
	\mathbf{G}_{B}\mathbf{U}\mathbf{e}_i-\mathbf{G}_{A}\mathbf{U}\mathbf{e}_i&=\mathbf{G}_{A}\mathbf{D}_{A}\mathbf{U}\mathbf{e}_i-\mathbf{G}_{A}\mathbf{U}\mathbf{e}_i
	=\mathbf{G}_{A}\left(\mathbf{D}_{A}\mathbf{U}\mathbf{e}_i-\mathbf{U}\mathbf{e}_i\right).
\end{align}
Since $\|\cdot\|_{2}$ is a matrix norm compatible with the vector 2-norm, we have
\begin{align}
	\|\mathbf{G}_{B}\mathbf{U}\mathbf{e}_i-\mathbf{G}_{A}\mathbf{U}\mathbf{e}_i\|_2&=\|\mathbf{G}_{A}\left(\mathbf{D}_{A}\mathbf{U}\mathbf{e}_i-\mathbf{U}\mathbf{e}_i\right)\|_2
	\leq\|\mathbf{G}_{A}\|_{2}\|\mathbf{D}_{A}\mathbf{U}\mathbf{e}_i-\mathbf{U}\mathbf{e}_i\|_2.
\end{align}
Because $\mathbf{G}_{A}$ is known, $\leq\|\mathbf{G}_{A}\|_{2}$ can be regarded as a known value. Meanwhile, since $\|\mathbf{D}_{A}-\mathbf{I}\|_{2}<\epsilon$, we have
\begin{align}
	\|\mathbf{D}_{A}\mathbf{U}\mathbf{e}_i-\mathbf{U}\mathbf{e}_i\|_2&=\|\left(\mathbf{D}_{A}-\mathbf{I}\right)\mathbf{U}\mathbf{e}_i\|_2
	\leq\|\mathbf{D}_{A}-\mathbf{I}\|_{2}\|\mathbf{U}\mathbf{e}_i\|_2
	<\epsilon\|\mathbf{U}\mathbf{e}_i\|_2.
\end{align}
By regarding $\|\mathbf{U}\mathbf{e}_i\|_2$ a constant, $\|\mathbf{D}_{A}\mathbf{U}\mathbf{e}_i-\mathbf{U}\mathbf{e}_i\|_2$ is bounded and approach to zero when $\|\mathbf{D}_{A}-\mathbf{I}\|_{2}<\epsilon$, i.e.,
\begin{align}
	&\forall \epsilon_2>0,~\exists\epsilon>0,
	~s.t.,~\|\mathbf{D}_{A}-\mathbf{I}\|_{2}<\epsilon,~\|\mathbf{D}_{A}\mathbf{U}\mathbf{e}_i-\mathbf{U}\mathbf{e}_i\|_2<\epsilon_2.
\end{align}
Since $\epsilon_2>0$ can be any positive minimum value, we assume that $\epsilon_2=\frac{\epsilon}{\|\mathbf{G}_{A}\|_{2}}$, and we have
\begin{align}
	\|\mathbf{G}_{B}\mathbf{U}\mathbf{e}_i-\mathbf{G}_{A}\mathbf{U}\mathbf{e}_i\|_2&\leq\|\mathbf{G}_{A}\|_{2}\|\mathbf{D}_{A}\mathbf{U}\mathbf{e}_i-\mathbf{U}\mathbf{e}_i\|_2
	<\|\mathbf{G}_{A}\|_{2}\times\frac{\epsilon}{\|\mathbf{G}_{A}\|_{2}}=\epsilon.
\end{align}
According to the triangle inequality and Eq. (\ref{Eq. The. 1 condition 2}), we can derive that
\begin{align}\label{Eq. Appendix B inequality}
	\|\mathbf{G}_{A}\mathbf{U}\mathbf{e}_i\|_2&=\|\mathbf{G}_{B}\mathbf{U}\mathbf{e}_i-\left(\mathbf{G}_{B}\mathbf{U}\mathbf{e}_i-\mathbf{G}_{A}\mathbf{U}\mathbf{e}_i\right)\|_2\notag\\
	&\leq\|\mathbf{G}_{B}\mathbf{U}\mathbf{e}_i\|_2+\|\mathbf{G}_{B}\mathbf{U}\mathbf{e}_i-\mathbf{G}_{A}\mathbf{U}\mathbf{e}_i\|_2\notag\\
	&<\lambda_K\left(\Lambda_{B}\right)+\epsilon
	<\lambda_K\left(\Lambda_{A}\right)+2\epsilon.
\end{align}
There are two cases for the relationship between $\|\mathbf{G}_{A}\mathbf{U}\mathbf{e}_i\|_2$ and $\lambda_K\left(\Lambda_{A}\right)$:  $\|\mathbf{G}_{A}\mathbf{U}\mathbf{e}_i\|_2>\lambda_K\left(\Lambda_{A}\right)$ and $\|\mathbf{G}_{A}\mathbf{U}\mathbf{e}_i\|_2\leq\lambda_K\left(\Lambda_{A}\right)$. If $\|\mathbf{G}_{A}\mathbf{U}\mathbf{e}_i\|_2>\lambda_K\left(\Lambda_{A}\right)$, we let 
\begin{equation}
	\exists \mathbf{e}_i,~s.t.,~\epsilon\in\left(0,\frac{\|\mathbf{G}_{A}\mathbf{U}\mathbf{e}_i\|_2-\lambda_K\left(\Lambda_{A}\right)}{2}\right],
\end{equation}
so we have 
\begin{equation}
	\lambda_K\left(\Lambda_{A}\right)+2\epsilon\in\left(\lambda_K\left(\Lambda_{A}\right),\|\mathbf{G}_{A}\mathbf{U}\mathbf{e}_i\|_2\right].
\end{equation}
However, it is contradictory between $\lambda_K\left(\Lambda_{A}\right)+2\epsilon\leq\|\mathbf{G}_{A}\mathbf{U}\mathbf{e}_i\|_2$ and Eq. (\ref{Eq. Appendix B inequality}), $\|\mathbf{G}_{A}\mathbf{U}\mathbf{e}_i\|_2\leq\lambda_K\left(\Lambda_{A}\right)$ hold true for all $\epsilon>0$. According to the Definition \ref{Def. Successive Minima},  $\mathbf{G}_{A}\mathbf{U}$ is one of the SIVP solution of $\Lambda_{A}$. There exists a unimodular matrix $\mathbf{U}$ such that $\mathbf{V}=\mathbf{U}$, which is contradictory to Eq. (\ref{Eq. Appendix B proposition}). Thus, the original proposition in Theorem \ref{The. similar lattice reduction} is true, and Theorem \ref{The. similar lattice reduction} has been proved.	 

\color{blue}
\section{Proof of Theorem \ref{Pro. NP-hard}}\label{App. unimodular NP-hardness proof}

The proof consists of two parts. We first establish the exact SP1-SP4 decomposition of the unimodularly restricted problem. We then prove that the exact coefficient evaluation required in SP2 is NP-hard.

\subsection{Exact subproblem decomposition}

For each $\mathbf A\in\mathcal A_{\mathrm{uni}}$, define the fixed-coefficient objective
\begin{equation}
	\label{Eq. appendix fixed coefficient objective}
	F_{\mathbf A}(\mathbf d)
	=
	\mathbf d^{\top}
	\left[
	\left(\mathbf A\mathbf A^{\top}\right)\circ\mathbf M
	\right]
	\mathbf d,
	\qquad
	\mathbf d\in\varOmega.
\end{equation}
Under the exact single-output mapping in Eq.~(\ref{Eq. d2A mapping}), the objective of the restricted joint problem is
\begin{equation}
	\label{Eq. appendix restricted objective}
	f(\mathbf d)
	=
	F_{\mathbf A(\mathbf d)}(\mathbf d)
	=
	\min_{\mathbf A\in\mathcal A_{\mathrm{uni}}}
	F_{\mathbf A}(\mathbf d).
\end{equation}
The last equality follows from the definition of the exact UMSBP oracle. The deterministic tie-breaking rule only selects one minimizer and does not change the minimum value.
For a fixed coefficient matrix $\mathbf A$, let
\begin{equation}
	\label{Eq. appendix fixed coefficient minimizers}
	\mathcal D(\mathbf A)
	=
	\underset{\mathbf d\in\varOmega}{\arg\min}
	F_{\mathbf A}(\mathbf d)
\end{equation}
denote its fixed-coefficient minimizer set. SP1 determines $\mathcal D(\mathbf A)$. SP2 retains only the minimizers that are consistent with the coefficient label, i.e.,
\begin{equation}
	\label{Eq. appendix consistency intersection}
	\mathcal D(\mathbf A)\cap\mathcal P(\mathbf A)
	=
	\left\{
	\mathbf d\in\mathcal D(\mathbf A)
	\,\middle|\,
	\mathbf A(\mathbf d)=\mathbf A
	\right\}.
\end{equation}
Accordingly, the region-wise candidate set obtained by SP1-SP3 is
\begin{equation}
	\label{Eq. appendix region candidate set}
	\mathcal C
	=
	\left\{
	(\mathbf A,\mathbf d)
	\,\middle|\,
	\mathbf A\in\mathcal A_{\mathrm{uni}},
	\;
	\mathbf d\in
	\mathcal D(\mathbf A)\cap\mathcal P(\mathbf A)
	\right\}.
\end{equation}
We next show that $\mathcal C$ contains a global solution of the restricted problem. Let $\mathbf d^{\star}$ be a global minimizer of Eq.~(\ref{Eq. appendix restricted objective}) and set $\mathbf A^{\star}=\mathbf A(\mathbf d^{\star})$. For every $\mathbf d\in\varOmega$, Eq.~(\ref{Eq. appendix restricted objective}) gives
\begin{equation}
	f(\mathbf d) \leq F_{\mathbf A^{\star}}(\mathbf d).
	\label{Eq. appendix lower envelope inequality}
\end{equation}
Since $\mathbf d^{\star}$ globally minimizes $f(\mathbf d)$, we have
\begin{align}
	F_{\mathbf A^{\star}}(\mathbf d^{\star})
	&
	=
	f(\mathbf d^{\star})
	\leq
	f(\mathbf d)
	\leq
	F_{\mathbf A^{\star}}(\mathbf d),
	\qquad
	\forall\,\mathbf d\in\varOmega.
	\label{Eq. appendix global fixed coefficient inequality}
\end{align}
Therefore, $\mathbf d^{\star}\in\mathcal D(\mathbf A^{\star})$.
Moreover, the definition $\mathbf A^{\star}=\mathbf A(\mathbf d^{\star})$ implies $\mathbf d^{\star}\in\mathcal P(\mathbf A^{\star})$.
Consequently,
\begin{equation}
	(\mathbf A^{\star},\mathbf d^{\star})
	\in
	\mathcal C.
	\label{Eq. appendix global solution in candidate set}
\end{equation}
Conversely, every $(\mathbf A,\mathbf d)\in\mathcal C$ satisfies $\mathbf A(\mathbf d)=\mathbf A$, and hence $F_{\mathbf A}(\mathbf d)=f(\mathbf d)$.
It is therefore feasible for the restricted joint problem. Combining this observation with Eq.~(\ref{Eq. appendix global solution in candidate set}) yields
\begin{equation}
	\min_{(\mathbf A,\mathbf d)\in\mathcal C}
	F_{\mathbf A}(\mathbf d)
	=
	\min_{\mathbf d\in\varOmega}
	f(\mathbf d).
	\label{Eq. appendix candidate equivalence}
\end{equation}
Thus, SP3 identifies the nonempty consistency intersections in Eq.~(\ref{Eq. appendix consistency intersection}), while SP4 selects the smallest objective value over the resulting candidates. This proves the exact SP1-SP4 decomposition.

\subsection{NP-hardness of the exact UMSBP}



Before proving the NP-hardness of SP2, we recall the exact Euclidean closest vector problem (CVP), which is NP-hard \cite{micciancio2002complexity}.

\begin{definition}[Exact Euclidean CVP]
	Given a nonsingular rational lattice basis $\mathbf B\in\mathbb Q^{n\times n}$ and a target $\mathbf t\in\mathbb Q^n$, find an integer vector $\mathbf z_{\mathrm{CVP}}\in\mathbb Z^n$ satisfying
	\begin{equation}
		\label{Eq. appendix CVP definition}
		\mathbf z_{\mathrm{CVP}}\in\underset{\mathbf z\in\mathbb Z^n}{\arg\min}
		\left\|\mathbf B\mathbf z-\mathbf t\right\|_2^2.
	\end{equation}
\end{definition}
Although the system model is formulated over the real field, the complexity reduction is restricted to rationally encoded instances to ensure finite input encoding in the conventional Turing-machine model. Since $\mathbb Q\subset\mathbb R$, hardness on this subclass establishes worst-case hardness for the real-domain formulation.

For a nonsingular rational lattice basis $\mathbf G\in\mathbb Q^{m\times m}$, the exact UMSBP is
\begin{equation}
	\label{Eq. appendix UMSBP definition}
	\mathbf U^{\star}
	\in
	\underset{\substack{
			\mathbf U\in\mathbb Z^{m\times m}\\
			|\det(\mathbf U)|=1
	}}{\arg\min}
	\left\|
	\mathbf G\mathbf U
	\right\|_F^2.
\end{equation}
We construct a polynomial-time reduction from the CVP in Eq.~(\ref{Eq. appendix CVP definition}) to the UMSBP in Eq.~(\ref{Eq. appendix UMSBP definition}). For the given CVP instance, define $C_0=\|\mathbf B\|_F^2+\|\mathbf t\|_2^2$ and choose $\beta=C_0+1$. Since $C_0\geq0$, it follows that
\begin{equation}
	\beta^2>C_0.
	\label{Eq. appendix beta condition}
\end{equation}
Construct the augmented rational basis
\begin{equation}
	\label{Eq. appendix augmented basis}
	\widehat{\mathbf B}
	=
	\begin{bmatrix}
		\mathbf B &-\mathbf t\\
		\mathbf 0_n^{\top}&\beta
	\end{bmatrix}
	\in\mathbb Q^{(n+1)\times(n+1)}.
\end{equation}
This construction uses only polynomially many rational arithmetic operations, and its encoding length is polynomial in the encoding length of $(\mathbf B,\mathbf t)$.
Let
\begin{equation}
	\label{Eq. appendix augmented UMSBP optimizer}
	\widehat{\mathbf A}^{\star}
	\in
	\underset{\substack{
			\mathbf A\in\mathbb Z^{(n+1)\times(n+1)}\\
			|\det(\mathbf A)|=1
	}}
	{\arg\min}
	\left\|
	\widehat{\mathbf B}\mathbf A
	\right\|_F^2.
\end{equation}
be an exact UMSBP output. 
Denote the last row of $\widehat{\mathbf A}^{\star}$ by $\mathbf r\triangleq(\widehat{\mathbf A}^{\star})^{\top}\mathbf e_{n+1}=[r_1,\ldots,r_{n+1}]^{\top}$. Thus, $\mathbf r^{\top}$ is the last row of $\widehat{\mathbf A}^{\star}$.
Since $\widehat{\mathbf A}^{\star}$ is unimodular, it is nonsingular, and hence $\breve{\mathbf r}\neq\mathbf 0_{n+1}^{\top}$. Because $\breve{\mathbf r}$ is an integer vector,
\begin{equation}
	\|\mathbf r\|_2^2\in\{1,2,\ldots\}.
	\label{Eq. appendix integer last row norm}
\end{equation}
The last row of $\widehat{\mathbf B}\widehat{\mathbf A}^{\star}$ is $\beta\mathbf r^{\top}$. Therefore,
\begin{equation}
	\left\|
	\widehat{\mathbf B}\widehat{\mathbf A}^{\star}
	\right\|_F^2
	\geq
	\beta^2\|\mathbf r\|_2^2.
	\label{Eq. appendix augmented lower bound}
\end{equation}
On the other hand, $\mathbf I_{n+1}$ is a feasible unimodular matrix, and
\begin{align}
	\left\|
	\widehat{\mathbf B}\mathbf I_{n+1}
	\right\|_F^2
	&=
	\|\mathbf B\|_F^2
	+
	\|\mathbf t\|_2^2
	+
	\beta^2
	=
	C_0+\beta^2.
	\label{Eq. appendix identity upper bound}
\end{align}
If $\|\breve{\mathbf r}\|_2^2\geq2$, then Eq.~(\ref{Eq. appendix beta condition}) and Eq.~(\ref{Eq. appendix augmented lower bound}) give
\begin{equation}
	\left\|
	\widehat{\mathbf B}\widehat{\mathbf A}^{\star}
	\right\|_F^2
	\geq
	2\beta^2
	>
	\beta^2+C_0.
\end{equation}
This contradicts the feasible upper bound in Eq.~(\ref{Eq. appendix identity upper bound}). Hence,
\begin{equation}
	\|\mathbf r\|_2^2=1.
	\label{Eq. appendix unit last row}
\end{equation}
An integer vector has squared Euclidean norm one if and only if it is a signed standard basis vector. Thus, there exist a unique $j\in\{1,\ldots,n+1\}$ and an $s\in\{-1,1\}$ such that
\begin{equation}
	r_j=s,
	\qquad
	r_i=0,\quad i\neq j.
	\label{Eq. appendix signed unit row}
\end{equation}

Let the $j$-th column of $\widehat{\mathbf A}^{\star}$ be
\begin{equation}
	\widehat{\mathbf a}^{\star}_j
	=
	\begin{bmatrix}
		\mathbf q\\
		s
	\end{bmatrix},
	\qquad
	\mathbf q\in\mathbb Z^n,
\end{equation}
and define
\begin{equation}
	\label{Eq. appendix extracted CVP vector}
	\mathbf z^{\star}=s\mathbf q.
\end{equation}
We next prove that $\mathbf z^{\star}$ solves the original CVP. Let $\mathbf P$ be a signed permutation matrix that moves the $j$-th column to the last position and multiplies that column by $s$. Then
\begin{equation}
	\mathbf P^{\top}\mathbf P=\mathbf I_{n+1},
	\qquad
	|\det(\mathbf P)|=1.
\end{equation}
Define
\begin{equation}
	\widetilde{\mathbf A}
	=
	\widehat{\mathbf A}^{\star}\mathbf P.
\end{equation}
Since right multiplication by a signed permutation only permutes columns and changes their signs,
\begin{equation}
	\left\|
	\widehat{\mathbf B}\widetilde{\mathbf A}
	\right\|_F^2
	=
	\left\|
	\widehat{\mathbf B}\widehat{\mathbf A}^{\star}
	\right\|_F^2.
	\label{Eq. appendix permutation invariance}
\end{equation}
Thus, $\widetilde{\mathbf A}$ remains an exact UMSBP optimizer. By Eq.~(\ref{Eq. appendix signed unit row}), its last row is $[0,\ldots,0,1]$, so it has the block form
\begin{equation}
	\label{Eq. appendix block optimizer}
	\widetilde{\mathbf A}
	=
	\begin{bmatrix}
		\mathbf V&\mathbf z^{\star}\\
		\mathbf 0_n^{\top}&1
	\end{bmatrix}
\end{equation}
for some $\mathbf V\in\mathbb Z^{n\times n}$. Moreover,
\begin{equation}
	|\det(\mathbf V)|
	=
	|\det(\widetilde{\mathbf A})|
	=
	1,
\end{equation}
and hence $\mathbf V$ is unimodular.
With Eq.~(\ref{Eq. appendix augmented basis}) and Eq.~(\ref{Eq. appendix block optimizer}), we obtain
\begin{equation}
	\widehat{\mathbf B}\widetilde{\mathbf A}
	=
	\begin{bmatrix}
		\mathbf B\mathbf V&
		\mathbf B\mathbf z^{\star}-\mathbf t\\
		\mathbf 0_n^{\top}&\beta
	\end{bmatrix}.
\end{equation}
Therefore,
\begin{equation}
	\label{Eq. appendix separated UMSBP objective}
	\left\|
	\widehat{\mathbf B}\widetilde{\mathbf A}
	\right\|_F^2
	=
	\|\mathbf B\mathbf V\|_F^2
	+
	\|\mathbf B\mathbf z^{\star}-\mathbf t\|_2^2
	+
	\beta^2.
\end{equation}
By contradiction, we suppose that $\mathbf z^{\star}$ is not an exact CVP solution. Then there exists a vector $\mathbf z'\in\mathbb Z^n$ such that
\begin{equation}
	\|\mathbf B\mathbf z'-\mathbf t\|_2^2
	<
	\|\mathbf B\mathbf z^{\star}-\mathbf t\|_2^2.
	\label{Eq. appendix better CVP vector}
\end{equation}
Construct
\begin{equation}
	\mathbf A'
	=
	\begin{bmatrix}
		\mathbf V&\mathbf z'\\
		\mathbf 0_n^{\top}&1
	\end{bmatrix}.
\end{equation}
Since
\begin{equation}
	|\det(\mathbf A')|
	=
	|\det(\mathbf V)|
	=
	1,
\end{equation}
$\mathbf A'$ is a feasible UMSBP coefficient matrix. Its objective
value is
\begin{equation}
	\left\|
	\widehat{\mathbf B}\mathbf A'
	\right\|_F^2
	=
	\|\mathbf B\mathbf V\|_F^2
	+
	\|\mathbf B\mathbf z'-\mathbf t\|_2^2
	+
	\beta^2.
\end{equation}
By Eq.~(\ref{Eq. appendix better CVP vector}), this gives
\begin{equation}
	\left\|
	\widehat{\mathbf B}\mathbf A'
	\right\|_F^2
	<
	\left\|
	\widehat{\mathbf B}\widetilde{\mathbf A}
	\right\|_F^2,
\end{equation}
contradicting the optimality of $\widetilde{\mathbf A}$. Therefore,
\begin{equation}
	\mathbf z^{\star}
	\in
	\underset{\mathbf z\in\mathbb Z^n}{\arg\min}
	\|\mathbf B\mathbf z-\mathbf t\|_2^2.
	\label{Eq. appendix extracted vector solves CVP}
\end{equation}

The extraction requires scanning $\breve{\mathbf r}$, locating its unique nonzero entry $s\in\{-1,1\}$, reading the first $n$ entries of the corresponding column, and returning $\mathbf z^{\star}=s\mathbf q$.
It therefore takes polynomial time. We have consequently established
\begin{equation}
	\mathrm{CVP}
	\leq_{\mathrm p}
	\mathrm{UMSBP}.
	\label{Eq. appendix CVP UMSBP reduction}
\end{equation}
Since exact Euclidean CVP is NP-hard, the exact UMSBP is NP-hard.

\subsection{UMSBP as the computational bottleneck of SP2}

It remains to connect the preceding UMSBP instance to the exact
coefficient mapping in SP2. For the augmented basis in
Eq.~(\ref{Eq. appendix augmented basis}), define $\widehat{\mathbf M}=\widehat{\mathbf B}^{\top}\widehat{\mathbf B}$ and select $\widehat{\mathbf d}=\mathbf 1_{n+1}$. Since all entries of $\widehat{\mathbf d}$ are positive and their product equals one, we have $\widehat{\mathbf d}\in\varOmega$. For every $\mathbf A\in\mathcal A_{\mathrm{uni}}$, we have
\begin{align}
	\left\|
	\widehat{\mathbf M}^{\frac12}
	\operatorname{diag}(\widehat{\mathbf d})
	\mathbf A
	\right\|_F^2
	&=
	\operatorname{Tr}
	\left(
	\mathbf A^{\top}\widehat{\mathbf M}\mathbf A
	\right)
	=
	\operatorname{Tr}
	\left(
	\mathbf A^{\top}
	\widehat{\mathbf B}^{\top}
	\widehat{\mathbf B}
	\mathbf A
	\right)
	=
	\left\|
	\widehat{\mathbf B}\mathbf A
	\right\|_F^2.
	\label{Eq. appendix UMSBP mapping equivalence}
\end{align}
Hence, evaluating the exact coefficient mapping $\mathbf A(\widehat{\mathbf d})$ in Eq.~(\ref{Eq. d2A mapping}) solves the UMSBP instance constructed above. Because the two objectives agree for every feasible $\mathbf A$, they also have the same minimizer set. The deterministic tie-breaking rule therefore does not affect the reduction. 
Eq.~(\ref{Eq. appendix UMSBP mapping equivalence}) shows that an exact UMSBP instance can be solved by evaluating the exact coefficient mapping in Eq.~(\ref{Eq. d2A mapping}). Since exact UMSBP is NP-hard, the general exact coefficient oracle invoked in SP2 is NP-hard. Therefore, the exact SP1--SP4 procedure contains an NP-hard computational bottleneck. This completes the proof of Theorem~\ref{Pro. NP-hard}.

\color{black}
\section{Proof of Theorem \ref{Pro. solution set}}\label{App. the charteristic of solution set}
the characteristic of $\mathcal{S}$ consists of the boundedness, discreteness, finiteness and nonempty property, which we will elaborate them respectively:

\noindent\textbf{i) Boundedness:} 
Since $\mathbf{M}$ is a PDM under the condition in Lemma \ref{Lem. extreme existence and uniqueness}, there must exist a constant $\phi>0$, s.t., 
\begin{equation}
	\mathbf{v}^{\top}\mathbf{M}\mathbf{v}\geq\phi\|\mathbf{v}\|^2,~\forall\mathbf{v}\in\mathbb{R}^K.
\end{equation}
Considering the target function in Eq. (\ref{Eq. optimization problem}), we have 
\begin{equation}
	\mathrm{Tr}\left(\mathbf{A}^{\top}\mathbf{D}^{\top}\mathbf{M}\mathbf{D}\mathbf{A}\right)\geq\phi\sum_{i=1}^{K}\|\mathbf{D}\mathbf{a}_i\|^2,
\end{equation}
where $\mathbf{a}_i$ represents the $i$-th column of $\mathbf{A}$. By the constrain $\prod_{i=1}^{K} d_i=1$, we can derive the following solutions:
\begin{itemize}
	\item[1] when a $d_i\to0$, it must have a $d_j\to\infty$ to remain the constrain.
	\item[2] when a $d_j\to0$, and the $j$-th parameter of corresponded column $\mathbf{a}_i$ is nonzero, we have $\|\mathbf{D}\mathbf{a}_i\|\to\infty$, and $ \mathrm{Tr}\left(\mathbf{A}^{\top}\mathbf{D}^{\top}\mathbf{M}\mathbf{D}\mathbf{A}\right)\to\infty$, which is impossible to as a minimum value of $\mathrm{Tr}\left(\mathbf{A}^{\top}\mathbf{D}^{\top}\mathbf{M}\mathbf{D}\mathbf{A}\right)$. 
\end{itemize} 
Thus, it exists a constant $C$, s.t., $\forall d_i\in[m,M]$, $m>0$, $M<\infty$ when $\mathrm{Tr}\left(\mathbf{A}^{\top}\mathbf{D}^{\top}\mathbf{M}\mathbf{D}\mathbf{A}\right)\leq C$, which reveals the boundedness of $\mathcal{S}$.

\noindent\textbf{ii) Discreteness:} 
We prove this property by contradiction. Firstly, we assume there exists a dissimilarity point sequence $\{(\mathbf{A}_n,\mathbf{d}_n)\}\subset\mathcal{S}$ satisfies
\begin{equation}
	(\mathbf{A}_n,\mathbf{d}_n)\to(\mathbf{A}^*,\mathbf{d}^*).
\end{equation}
Since $\mathbf{A}_n$ is an integer-valued matrix, it must be $\mathbf{A}_n=\mathbf{A}^*$ when $n$ is sufficiently large, which implies that $\mathbf{d}_n$ is the solution of Eq. (\ref{Eq. fixed A optimization problem}). According to Lemma \ref{Lem. extreme existence and uniqueness} and Remark \ref{Rem. convex problem in fixed A}, Eq. (\ref{Eq. fixed A optimization problem}) can be regarded as a convex problem whose solution represented by $\mathbf{d}_{\mathbf{A}}^*$ is unique \cite{boyd2004convex}. Based on Eq. (\ref{Eq. d2A mapping}) and Remark \ref{Rem. convex problem in fixed A}, any $(\mathbf{A},\mathbf{d})\in\mathcal{S}$ requires consistency, which reveals that $(\mathbf{A}^*,\mathbf{d}_{\mathbf{A}}^*)\in\mathcal{S}$. Thus, when $n$ is large enough, $(\mathbf{A}_n,\mathbf{d}_n)$ becomes constant sequence which is contradictory to the definition of dissimilarity point sequence. So there cannot exist a dissimilarity point sequence consisted in $\mathcal{S}$, and $\mathcal{S}$ is discrete.

\noindent\textbf{iii) Finiteness:} 
Since $\mathcal{S}$ is a bounded discrete set in finite dimensional Euclidean space, $\mathcal{S}$ is a compact set. According \cite[Theorem 2.4.1]{rudin1976principles}, $\mathcal{S}$ is a finite set.

\noindent\textbf{iv) Nonempty property:} 
We prove this property by contradiction. Firstly, we assume that none of regions in $\varOmega$ have the extreme. Since $\varOmega$ fulfill the whole $\mathbb{R}_{+}^{K-1}$, the problem in Eq. (\ref{Eq. fixed A optimization problem}) have no solution for every $\mathbf{A}$, which is contradictory to Lemma \ref{Lem. extreme existence and uniqueness}. Thus, there must be at least one region with extreme, and $\mathcal{S}$ is nonempty.

\section{Proof of Lemma \ref{Lem. Uniqueness of intersection}}\label{App. Uniqueness of intersection proof}
We prove Lemma \ref{Lem. Uniqueness of intersection} by contradiction. Firstly, we assume
\begin{equation}
	\exists\mathbf{r}=[r_1,\cdots,r_{K-1},\beta]^{\top}\in\mathcal{D}^{+} s.t., \mathbf{r}\notin\varOmega,~\mathrm{and},~\mathbf{B}\mathbf{r}\notin\varOmega_2.
\end{equation} 
Since $d_i>0,i\in {1,\cdots,K}$, it is impossible to find a $\mathbf{r}$ in the positive orthant of $\mathcal{D}$ with $\beta<0$. Thus, $\mathbf{r}$ must belong to $\varOmega$, and $\mathbf{B}\mathbf{r}\notin\varOmega_2$. In other words, any ray toward the positive orthant of $\mathcal{D}$ must have intersection to $\varOmega$, and so as for $\mathcal{B}$ to $\varOmega_2$.

Regarding the uniqueness of intersection, since there are two propositions in this lemma, we will prove them respectively.
\begin{itemize}
	\item \textit{$\varOmega$: }Assuming that there are two intersections for a ray to $\varOmega$, i.e.,
	\begin{align}\label{Eq. constract proposition on Omega}
		&\forall\mathbf{d}_1=[d_{1,1},\cdots,d_{1,K}]^{\top},\mathbf{d}_2=\eta\mathbf{d}_1\in\varOmega,\eta\neq1,
		~\text{s.t.},~\prod_{i=1}^{K}d_{1,i}=1,\prod_{i=1}^{K}d_{2,i}=1.
	\end{align}
	However, due to
	\begin{equation}
		\prod_{i=1}^{K}d_{2,i}=\eta^K\prod_{i=1}^{K}d_{1,i}=\eta^K,
	\end{equation}
	it is contradictory to Eq. (\ref{Eq. constract proposition on Omega}), and any intersection number more than $2$ is contradictory similarly for the same reason. Thus, there is only one intersection between any ray toward the positive orthant of $\mathcal{D}$ and $\varOmega$.
	\item \textit{$\varOmega_2$: }For any $\mathbf{d}\in\varOmega$, when $\prod_{i=1}^{K}d_i=1$, we have
	\begin{equation}
		\prod_{i=1}^{K}\frac{1}{d_i}=1
	\end{equation}
	so the set 
	\begin{equation}
		\varOmega_3=\left\{(\mathbf{1}_K\oslash\mathbf{d})|\prod_{i=1}^{K}d_i=1,d_i>0,i\in {1,\cdots,K}\right\}.
	\end{equation}
	corresponds to $\varOmega$ one-by-one. Similarly, when we assume that
	\begin{align}
		&\forall\mathbf{x}_1,\mathbf{x}_2=\eta\mathbf{x}_1\in\varOmega_2,\eta\neq1,
		\text{s.t.},\prod_{i=1}^{K}\left[\mathbf{B}^{-1}\mathbf{x}_1\right]_i=1,\prod_{i=1}^{K}\left[\mathbf{B}^{-1}\mathbf{x}_2\right]_i=1,
	\end{align} 
	it is contradictory since
	\begin{equation}
		\prod_{i=1}^{K}\left[\mathbf{B}^{-1}\mathbf{x}_2\right]_i=\eta^K\prod_{i=1}^{K}\left[\mathbf{B}^{-1}\mathbf{x}_1\right]_i=\eta^K.
	\end{equation}
	According the one-by-one corresponding, there is only one intersection between any ray toward the positive orthant of $\mathcal{B}$ and $\varOmega_2$.
\end{itemize}
Combining aforementioned contents, the lemma has been proved.

\section{Proof of Theorem \ref{The. Convergation of proposed method}}\label{App. Convergation of proposed method proof}
Before proving Theorem \ref{The. Convergation of proposed method}, there is a lemma about the characteristic of Hilbert metric on the vector $\mathbf{1}_K\oslash\mathbf{d}$
\begin{lemma}\label{Lem. Hilbert metric for vector inverse}
	In the complete metric space with Hilbert metric, we have
	\begin{equation}
		\forall \mathbf{d}_{x},\mathbf{d}_{y}\in\mathcal{D},~s.t.,~d_H\left(\mathbf{1}_K\oslash\mathbf{d}_{x},\mathbf{1}_K\oslash\mathbf{d}_{y}\right)=d_H\left(\mathbf{d}_{x},\mathbf{d}_{y}\right).
	\end{equation}
\end{lemma}
\begin{proof}
	Regarding $\forall \mathbf{d}_{x}=[d_{x,1},\cdots,d_{x,K}]^{\top},\mathbf{d}_{y}=[d_{y,1},\cdots,d_{y,K}]^{\top}\in\varOmega$, since the negative correlation of reciprocal, we have
	\begin{align}
		&d_H\left(\mathbf{1}_K\oslash\mathbf{d}_{x},\mathbf{1}_K\oslash\mathbf{d}_{y}\right)\notag\\=&\log\left(\frac{\max_{i=1}^{K}\left(\frac{d_{y,i}}{d_{x,i}}\right)}{\min_{i=1}^{K}\left(\frac{d_{y,i}}{d_{x,i}}\right)}\right)
		=\log\left(\frac{\frac{1}{\min_{i=1}^{K}\left(\frac{d_{x,i}}{d_{y,i}}\right)}}{\frac{1}{\max_{i=1}^{K}\left(\frac{d_{x,i}}{d_{y,i}}\right)}}\right)
		=\log\left(\frac{\max_{i=1}^{K}\left(\frac{d_{x,i}}{d_{y,i}}\right)}{\min_{i=1}^{K}\left(\frac{d_{x,i}}{d_{y,i}}\right)}\right)
		=d_H\left(\mathbf{d}_{x},\mathbf{d}_{y}\right).
	\end{align}
	Thus, the lemma has been proved.
\end{proof}
Then we prove the Theorem \ref{The. Convergation of proposed method}.
Towards the $p$-PAM symbols in Rayleigh channel with domain $[0,+\infty)$, the channel matrix $\mathbf{H}$ remains non-negative. Due to the difference on $\mathbf{M}$ in DIF and RIF, we analyze them respectively. 

\noindent \textbf{DIF:} When $N\geq K$, $\mathbf{H}$ is a non-negative matrix with full rank in row, which means
\begin{equation}\label{Eq. non-negative 1}
	\left[\mathbf{H}\mathbf{H}^{\top}\right]_{i,j}\geq 0,~i,j\in\left\{1,\cdots,K\right\}.
\end{equation}
Meanwhile, Lemma \ref{Lem. extreme existence and uniqueness} reveals that $\mathbf{H}\mathbf{H}^{\top}$ is a PDM when $N\geq K$. Based on the properties of M matrices \cite{berman1994nonnegative}, when $\mathbf M$ is a positive-definite $Z$-matrix, the inverse of a positive-defined M matrix is the positive-defined non-negative matrix which is a necessary and sufficient condition. So we learn that $\mathbf{M}=\left(\mathbf{H}\mathbf{H}^{\top}\right)^{-1}$ is a M matrix whose elements can be represented by
\begin{equation}
	m_{i,j}=\left[\left(\mathbf{H}\mathbf{H}^{\top}\right)^{-1}\right]_{i,j}=\left\{
	\begin{aligned}
		&\text{non-negative},&i=j;\\
		&\text{non-positive},&i\neq j.
	\end{aligned}
	\right.
\end{equation} 

\noindent \textbf{RIF:} In RIF, since $\frac{K}{\rho}\geq0$ and the situation $\frac{K}{\rho}=0$ happened only when $\rho\rightarrow+\infty$, we learn that $\frac{K}{\rho}\mathbf{I}$ is non-negative. Combining Eq. (\ref{Eq. non-negative 1}), it reveals that
\begin{equation}
	\left[\frac{K}{\rho}\mathbf{I}+\mathbf{H}\mathbf{H}^{\top}\right]_{i,j}=\left[\frac{K}{\rho}\mathbf{I}\right]_{i,j}+\left[\mathbf{H}\mathbf{H}^{\top}\right]_{i,j}\geq 0,
\end{equation}
where $~i,j\in\left\{1,\cdots,K\right\}$. As the same reason in DIF, we learn that $\mathbf{M}=\left(\frac{K}{\rho}\mathbf{I}+\mathbf{H}\mathbf{H}^{\top}\right)^{-1}$ is a M matrix whose elements can be represented by
\begin{equation}
	m_{i,j}=\left[\left(\frac{K}{\rho}\mathbf{I}+\mathbf{H}\mathbf{H}^{\top}\right)^{-1}\right]_{i,j}=\left\{
	\begin{aligned}
		&\text{non-negative},&i=j;\\
		&\text{non-positive},&i\neq j.
	\end{aligned}
	\right.
\end{equation} 

Since $\mathbf{A}\mathbf{A}^{\top}$ is non-negative in Theorem \ref{The. Convergation of proposed method}, we have
\begin{align}
	\left[\mathbf{G}\right]_{i,j}&=\left[\left(\mathbf{A}\mathbf{A}^{\top}\right)\circ\mathbf{M}\right]_{i,j}\notag\\
	&=\left[\mathbf{A}\mathbf{A}^{\top}\right]_{i,j}\times m_{i,j}
	=\left\{
	\begin{aligned}
		&\text{non-negative}\times \text{non-negative},&i=j;\\
		&\text{non-negative}\times \text{non-positive},&i\neq j,
	\end{aligned}
	\right.
	=\left\{
	\begin{aligned}
		&\text{non-negative},&i=j;\\
		&\text{non-positive},&i\neq j.
	\end{aligned}
	\right.
\end{align}
According to Lemma \ref{Lem. extreme existence and uniqueness}, $\mathbf{G}$ is a PDM, and $\mathbf{G}^{-1}$ is non-negative due to the properties of M matrix \cite{berman1994nonnegative}. Thus, $\mathbf{G}^{-1}$ can be regarded as a positive transformation. Based on the Perron-Frobenius theorem in direction and Hilbert metric, we have
\begin{equation}
	\exists \gamma\in\left(0,1\right),~d_H\left(\mathbf{G}^{-1}\mathbf{d}_{x},\mathbf{G}^{-1}\mathbf{d}_{y}\right)\leq\gamma d_H\left(\mathbf{d}_{x},\mathbf{d}_{y}\right)
\end{equation}
for all $\mathbf{d}_{x},\mathbf{d}_{y}\in \varOmega$. Define the function $g(\mathbf{d})\triangleq\Gamma\left(\mathbf{G}^{-1}\left(\mathbf{1}_K\oslash\mathbf{d}\right)\right)$, in which $\Gamma(\mathbf{x})$ can be regarded as a scaling operation on $\mathbf{x}$. Combining Lemma \ref{Lem. Hilbert metric for vector inverse}, we have
\begin{align}
	&d_H\left(g(\mathbf{d}_{x}),g(\mathbf{d}_{y})\right)\notag\\
	=&d_H\left(\Gamma\left(\mathbf{G}^{-1}\left(\mathbf{1}_K\oslash\mathbf{d}_{x}\right)\right),\Gamma\left(\mathbf{G}^{-1}\left(\mathbf{1}_K\oslash\mathbf{d}_{y}\right)\right)\right)\notag\\
	=&d_H\left(\mathbf{G}^{-1}\left(\mathbf{1}_K\oslash\mathbf{d}_{x}\right),\mathbf{G}^{-1}\left(\mathbf{1}_K\oslash\mathbf{d}_{y}\right)\right)\notag\\
	\leq& \gamma d_H\left(\mathbf{1}_K\oslash\mathbf{d}_{x},\mathbf{1}_K\oslash\mathbf{d}_{y}\right)
	=\gamma d_H\left(\mathbf{d}_{x},\mathbf{d}_{y}\right).
\end{align}
$d_H\left(g(\mathbf{d}_{x}),g(\mathbf{d}_{y})\right)\leq\gamma d_H\left(\mathbf{d}_{x},\mathbf{d}_{y}\right)$ reflects that the function $g(\mathbf{d})$ is a contraction within $\varOmega$, which means Eq. (\ref{Eq. iterative function}) will converge to a fixed direction satisfied Eq. (\ref{Eq. iterative terminate}). The theorem has been proved.

\section{Proof of Theorem \ref{The. upper bound of RA iterative number}}\label{App. upper bound of RA iterative number proof}
According to the process in Appendix \ref{App. Convergation of proposed method proof}, when Alg. \ref{Alg. Reciprocal Approximation} converges, we obtain
\begin{equation}\label{Eq. Convergence on g(x)}
	d_H\left(g(\mathbf{d}_{x}),g(\mathbf{d}_{y})\right)\leq\gamma d_H\left(\mathbf{d}_{x},\mathbf{d}_{y}\right).
\end{equation}
Thus, we have an inference chain denoted by
\begin{align}
	&d_H\left(\mathbf{d}^{(t+1)},\mathbf{d}^{(t)}\right)\notag\\
	=&d_H\left(\Gamma\left(\mathbf{G}^{-1}\left(\mathbf{1}_K\oslash\mathbf{d}^{(t)}\right)\right),\Gamma\left(\mathbf{G}^{-1}\left(\mathbf{1}_K\oslash\mathbf{d}^{(t-1)}\right)\right)\right)\notag\\
	=&d_H\left(\mathbf{G}^{-1}\left(\mathbf{1}_K\oslash\mathbf{d}^{(t)}\right),\mathbf{G}^{-1}\left(\mathbf{1}_K\oslash\mathbf{d}^{(t-1)}\right)\right)\notag\\
	\leq& \gamma d_H\left(\mathbf{1}_K\oslash\mathbf{d}^{(t)},\mathbf{1}_K\oslash\mathbf{d}^{(t-1)}\right)\notag\\
	=&\gamma d_H\left(\mathbf{d}^{(t)},\mathbf{d}^{(t-1)}\right)\notag\\
	\leq&\gamma^2 d_H\left(\mathbf{d}^{(t-1)},\mathbf{d}^{(t-2)}\right)\notag\\
	&\cdots\notag\\
	\leq&\gamma^{t} d_H\left(\mathbf{d}^{(1)},\mathbf{d}^{(0)}\right)\notag\\
	=&\gamma^{t} d_H\left(\mathbf{G}^{-1}(\mathbf{1}_K\oslash\mathbf{d}^{(0)}),\mathbf{d}^{(0)}\right).
\end{align}
In Alg. \ref{Alg. Reciprocal Approximation}, the algorithm terminates when $d_H\left(\mathbf{d}^{(t+1)},\mathbf{d}^{(t)}\right)\leq\epsilon$. Let $d_H\left(\mathbf{d}^{(t+1)},\mathbf{d}^{(t)}\right)=\epsilon$, we can derive that
\begin{align}
	t_{\mathrm{RA}}\log\left(\gamma\right)&\geq\log\left(\frac{d_H\left(\mathbf{d}^{(t+1)},\mathbf{d}^{(t)}\right)}{d_H\left(\mathbf{d}^{(1)},\mathbf{d}^{(0)}\right)}\right)
	=\log\epsilon-\log\left(d_H\left(\mathbf{G}^{-1}(\mathbf{1}_K\oslash\mathbf{d}^{(0)}),\mathbf{d}^{(0)}\right)\right).
\end{align}
As $\gamma\in\left(0,1\right)$, $\log\left(\gamma\right)$ is within $\left(-\infty,0\right)$, so
\begin{equation}\label{Eq. upper t}
	t_{\mathrm{RA}}\leq\frac{\log\epsilon-\log\left(d_H\left(\mathbf{G}^{-1}(\mathbf{1}_K\oslash\mathbf{d}^{(0)}),\mathbf{d}^{(0)}\right)\right)}{\log\left(\gamma\right)}.
\end{equation}

Regarding the bounds of $\gamma$, we employ the contraction ratio in Birkhoff-Hopf theorem which can be calculated by
\begin{equation}\label{Eq. contraction ratio}
	\kappa(\mathbf{\mathbf{G}^{-1}})=\tanh\left(\frac{\Delta(\mathbf{G}^{-1})}{4}\right).
\end{equation}
In Eq. (\ref{Eq. contraction ratio}), the projective diameter $\Delta(\mathbf{G}^{-1})$ can be calculated in \cite{lemmens2012nonlinear} by
\begin{align}\label{Eq. gamma bounds}
	\Delta(\mathbf{G}^{-1})&=\max_{1\leq i,j\leq K}d_H\left(\mathbf{G}^{-1}\mathbf{e}_i,\mathbf{G}^{-1}\mathbf{e}_j\right)
	=\log\left(\max_{i,j,p,q}\frac{g^{'}_{p,i}g^{'}_{q,j}}{g^{'}_{p,j}g^{'}_{q,i}}\right),
\end{align}
where $\mathbf{e}_i,i\in\{1,\cdots,K\}$ represents the $i$-th column of $\mathbf{I}$ and $g^{'}_{i,j},~i,j\in\{1,\cdots,K\}$ the $\{i,j\}$ element of $\mathbf{G}^{-1}$. Combining Eq. (\ref{Eq. upper t}) and Eq. (\ref{Eq. gamma bounds}), we can obtain Eq. (\ref{Eq. upper bound of tRA}) and Eq. (\ref{Eq. projective diameter}). The theorem has been proved.

\bibliographystyle{IEEEtran}
\bibliography{ref}

@Article{LLL,
  author  = {Arjen K. Lenstra and Hendrik W. Lenstra and L{\'a}szl{\'o} Mikl{\'o}s Lov{\'a}sz},
  journal = {Mathematische Annalen},
  title   = {Factoring polynomials with rational coefficients},
  year    = {1982},
  pages   = {515-534},
  volume  = {261},
}

@Book{hampton2013MIMO,
  author    = {Hampton, Jerry R.},
  publisher = {Cambridge University Press},
  title     = {Introduction to MIMO Communications},
  year      = {2013},
  place     = {Cambridge},
}

@Article{hassibi2003muchmassivemimo,
  author    = {Babak Hassibi and Bertrand M. Hochwald},
  journal   = {{IEEE} Trans. Inf. Theory},
  title     = {How much training is needed in multiple-antenna wireless links?},
  year      = {2003},
  number    = {4},
  pages     = {951--963},
  volume    = {49},
  doi       = {10.1109/TIT.2003.809594},
}

@Article{wang2023XLMIMO,
  author    = {Zhe Wang and Jiayi Zhang and Hongyang Du and Wei E. I. Sha and Bo Ai and Dusit Niyato and M{\'{e}}rouane Debbah},
  journal   = {{IEEE} Wirel. Commun.},
  title     = {Extremely Large-Scale {MIMO:} Fundamentals, Challenges, Solutions, and Future Directions},
  year      = {2024},
  number    = {3},
  pages     = {117--124},
  volume    = {31},
  doi       = {10.1109/MWC.132.2200443},
}

@Article{silva2017integer,
  author    = {Danilo Silva and Gabriel Fernando Pivaro and Gustavo Fraidenraich and Behnaam Aazhang},
  journal   = {{IEEE} Trans. Wirel. Commun.},
  title     = {On Integer-Forcing Precoding for the Gaussian {MIMO} Broadcast Channel},
  year      = {2017},
  number    = {7},
  pages     = {4476--4488},
  volume    = {16},
  doi       = {10.1109/TWC.2017.2699178},
}

@Article{venturelli2020optimization,
  author    = {Ricardo Bohaczuk Venturelli and Danilo Silva},
  journal   = {{IEEE} Wirel. Commun. Lett.},
  title     = {Optimization of Integer-Forcing Precoding for Multi-User {MIMO} Downlink},
  year      = {2020},
  number    = {11},
  pages     = {1860--1864},
  volume    = {9},
  doi       = {10.1109/LWC.2020.3006393},
}

@Article{he2018uplink,
  author    = {Wenbo He and Bobak Nazer and Shlomo Shamai Shitz},
  journal   = {{IEEE} Trans. Inf. Theory},
  title     = {Uplink-Downlink Duality for Integer-Forcing},
  year      = {2018},
  number    = {3},
  pages     = {1992--2011},
  volume    = {64},
  doi       = {10.1109/TIT.2018.2791589},
}

@Book{nguyen2010lll,
  author    = {Nguyen, Phong Q and Vall{\'e}e, Brigitte},
  publisher = {Springer},
  title     = {The LLL algorithm},
  year      = {2010},
}

@Article{qiu2024lattice,
  author    = {Xinzhe Qiu and Tao Yang and John Thompson},
  journal   = {{IEEE} Trans. Wirel. Commun.},
  title     = {On Lattice-Based Broadcasting for Massive-User {MIMO:} Practical Algorithms and Optimization},
  year      = {2024},
  number    = {11},
  pages     = {16544--16558},
  volume    = {23},
  doi       = {10.1109/TWC.2024.3442787},
}

@Article{kohlberg1982contraction,
  author    = {Elon Kohlberg and John W. Pratt},
  journal   = {Math. Oper. Res.},
  title     = {The Contraction Mapping Approach to the Perron-Frobenius Theory: Why Hilbert's Metric?},
  year      = {1982},
  number    = {2},
  pages     = {198--210},
  volume    = {7},
  doi       = {10.1287/MOOR.7.2.198},
}

@Article{birkhoff1957extensions,
  author    = {Birkhoff, Garrett},
  journal   = {Transactions of the American Mathematical Society},
  title     = {Extensions of Jentzsch's theorem},
  year      = {1957},
  number    = {1},
  pages     = {219--227},
  volume    = {85},
  publisher = {JSTOR},
}

@Article{carroll2004birkhoff,
  author    = {Carroll, Joseph E},
  journal   = {Linear algebra and its applications},
  title     = {Birkhoff's contraction coefficient},
  year      = {2004},
  pages     = {227--234},
  volume    = {389},
  publisher = {Elsevier},
}

@Article{rajatheva2020white,
  author  = {Rajatheva, Nandana and Atzeni, Italo and Björnson, Emil and Bourdoux, André and Buzzi, Stefano and Doré, Jean-Baptiste and Erkucuk, Serhat and Fuentes, Manuel and Guan, Ke and Hu, Yuzhou and Huang, Xiaojing and Hulkkonen, Jari and Jornet, Josep Miquel and Katz, Marcos1 and Nilsson, Rickard and Panayirci, Erdal and Rabie, Khaled and Rajapaksha, Nuwanthika and Salehi, Mohammad Javad and Sarieddeen, Hadi and Shahabuddin, Shahriar and Svensson, Tommy and Tervo, Oskari and Tölli, Antti and Wu, Qingqing and Xu, Wen},
  journal = {6G Research Visions},
  title   = {White paper on broadband connectivity in 6G},
  year    = {2020},
  volume  = {10},
  editor  = {null},
  place   = {Country unknown/Code not available},
  url     = {https://par.nsf.gov/biblio/10223732},
}

@Article{hochwald2004multiple,
  author    = {Bertrand M. Hochwald and Thomas L. Marzetta and Vahid Tarokh},
  journal   = {{IEEE} Trans. Inf. Theory},
  title     = {Multiple-Antenna Channel Hardening and Its Implications for Rate Feedback and Scheduling},
  year      = {2004},
  number    = {9},
  pages     = {1893--1909},
  volume    = {50},
  doi       = {10.1109/TIT.2004.833345},
}

@Article{yu2004sum,
  author    = {Wei Yu and John M. Cioffi},
  journal   = {{IEEE} Trans. Inf. Theory},
  title     = {Sum Capacity of Gaussian Vector Broadcast Channels},
  year      = {2004},
  number    = {9},
  pages     = {1875--1892},
  volume    = {50},
  doi       = {10.1109/TIT.2004.833336},
}

@Article{hochwald2005vector,
  author    = {Bertrand M. Hochwald and Christian B. Peel and A. Lee Swindlehurst},
  journal   = {{IEEE} Trans. Commun.},
  title     = {A vector-perturbation technique for near-capacity multiantenna multiuser communication-part {II:} perturbation},
  year      = {2005},
  number    = {3},
  pages     = {537--544},
  volume    = {53},
  doi       = {10.1109/TCOMM.2004.841997},
}

@Article{sun2014tomlinson,
  author    = {Liang Sun and Matthew R. McKay},
  journal   = {{IEEE} Trans. Signal Process.},
  title     = {Tomlinson-Harashima Precoding for Multiuser {MIMO} Systems With Quantized {CSI} Feedback and User Scheduling},
  year      = {2014},
  number    = {16},
  pages     = {4077--4090},
  volume    = {62},
  doi       = {10.1109/TSP.2014.2336633},
}

@Article{costa1983writing,
  author    = {Max H. M. Costa},
  journal   = {{IEEE} Trans. Inf. Theory},
  title     = {Writing on dirty paper},
  year      = {1983},
  number    = {3},
  pages     = {439--441},
  volume    = {29},
  doi       = {10.1109/TIT.1983.1056659},
}

@Book{micciancio2002complexity,
  author    = {Daniele Micciancio, Shafi Goldwasser (auth.)},
  publisher = {Springer},
  title     = {Complexity of Lattice Problems: A Cryptographic Perspective},
  year      = {2002},
  edition   = {1},
  isbn      = {9781461352938; 1461352932; 9781461508977; 1461508975},
  series    = {The Springer International Series in Engineering and Computer Science No.671},
}

@Article{vikalo2005sphere,
  author    = {Haris Vikalo and Babak Hassibi},
  journal   = {{IEEE} Trans. Signal Process.},
  title     = {On the sphere-decoding algorithm {II.} Generalizations, second-order statistics, and applications to communications},
  year      = {2005},
  number    = {8-1},
  pages     = {2819--2834},
  volume    = {53},
  doi       = {10.1109/TSP.2005.850350},
}

@Article{zhang2021local,
  author    = {Jiayi Zhang and Jing Zhang and Emil Bj{\"{o}}rnson and Bo Ai},
  journal   = {{IEEE} Trans. Commun.},
  title     = {Local Partial Zero-Forcing Combining for Cell-Free Massive {MIMO} Systems},
  year      = {2021},
  number    = {12},
  pages     = {8459--8473},
  volume    = {69},
  doi       = {10.1109/TCOMM.2021.3110214},
}

@Article{chen2022generalized,
  author    = {Li Chen and Zhiqin Wang and Ying Du and Yunfei Chen and F. Richard Yu},
  journal   = {{IEEE} J. Sel. Areas Commun.},
  title     = {Generalized Transceiver Beamforming for {DFRC} With {MIMO} Radar and {MU-MIMO} Communication},
  year      = {2022},
  number    = {6},
  pages     = {1795--1808},
  volume    = {40},
  doi       = {10.1109/JSAC.2022.3155515},
}

@Article{wang2024efficient,
  author    = {Zheng Wang and Le Liang and Shanxiang Lyu and Yili Xia and Yongming Huang and Derrick Wing Kwan Ng},
  journal   = {{IEEE} Trans. Wirel. Commun.},
  title     = {Efficient Statistical Linear Precoding for Downlink Massive {MIMO} Systems},
  year      = {2024},
  number    = {10},
  pages     = {14805--14818},
  volume    = {23},
  doi       = {10.1109/TWC.2024.3419137},
}

@Article{albreem2021overview,
  author    = {Mahmoud A. M. Albreem and Alaa H. Al Habbash and Ammar M. Abu{-}Hudrouss and Salama S. Ikki},
  journal   = {{IEEE} Access},
  title     = {Overview of Precoding Techniques for Massive {MIMO}},
  year      = {2021},
  pages     = {60764--60801},
  volume    = {9},
  doi       = {10.1109/ACCESS.2021.3073325},
}

@InProceedings{stern2016advanced,
  author    = {Sebastian Stern and Robert F. H. Fischer},
  booktitle = {{IEEE} International Symposium on Information Theory, {ISIT} 2016, Barcelona, Spain, July 10-15, 2016},
  title     = {Advanced factorization strategies for lattice-reduction-aided preequalization},
  year      = {2016},
  pages     = {1471--1475},
  publisher = {{IEEE}},
  doi       = {10.1109/ISIT.2016.7541543},
}

@InProceedings{hong2012reverse,
  author    = {Song{-}Nam Hong and Giuseppe Caire},
  booktitle = {Proceedings of the 2012 {IEEE} International Symposium on Information Theory, {ISIT} 2012, Cambridge, MA, USA, July 1-6, 2012},
  title     = {Reverse compute and forward: {A} low-complexity architecture for downlink distributed antenna systems},
  year      = {2012},
  pages     = {1147--1151},
  publisher = {{IEEE}},
  doi       = {10.1109/ISIT.2012.6283033},
}

@Article{ordentlich2014precoded,
  author    = {Or Ordentlich and Uri Erez},
  journal   = {{IEEE} Trans. Inf. Theory},
  title     = {Precoded Integer-Forcing Universally Achieves the {MIMO} Capacity to Within a Constant Gap},
  year      = {2015},
  number    = {1},
  pages     = {323--340},
  volume    = {61},
  doi       = {10.1109/TIT.2014.2370047},
}

@Book{boyd2004convex,
  author    = {Boyd, Stephen P and Vandenberghe, Lieven},
  publisher = {Cambridge university press},
  title     = {Convex optimization},
  year      = {2004},
}

@Book{Tao2012,
  author    = {Tao, Terence},
  publisher = {American Mathematical Soc.},
  title     = {Topics in random matrix theory},
  year      = {2012},
  volume    = {132},
}

@Article{Lyu2017boost,
  author    = {Shanxiang Lyu and Cong Ling},
  journal   = {{IEEE} Trans. Signal Process.},
  title     = {Boosted {KZ} and {LLL} Algorithms},
  year      = {2017},
  number    = {18},
  pages     = {4784--4796},
  volume    = {65},
  doi       = {10.1109/TSP.2017.2708020},
}

@Book{lemmens2012nonlinear,
  author    = {Lemmens, Bas and Nussbaum, Roger},
  publisher = {Cambridge University Press},
  title     = {Nonlinear Perron-Frobenius Theory},
  year      = {2012},
  volume    = {189},
}

@Book{berman1994nonnegative,
  author    = {Berman, Abraham and Plemmons, Robert J},
  publisher = {SIAM},
  title     = {Nonnegative matrices in the mathematical sciences},
  year      = {1994},
}

@Article{hopf1963inequality,
  author    = {Hopf, Eberhard},
  journal   = {Journal of Mathematics and Mechanics},
  title     = {An inequality for positive linear integral operators},
  year      = {1963},
  number    = {5},
  pages     = {683--692},
  volume    = {12},
  publisher = {JSTOR},
}

@Book{tse2005fundamentals,
  author    = {Tse, David and Viswanath, Pramod},
  publisher = {Cambridge university press},
  title     = {Fundamentals of wireless communication},
  year      = {2005},
}

@Article{avner2015vector,
  author    = {Yuval Avner and Benjamin M. Zaidel and Shlomo Shamai (Shitz)},
  journal   = {{IEEE} Trans. Inf. Theory},
  title     = {On Vector Perturbation Precoding for the {MIMO} Gaussian Broadcast Channel},
  year      = {2015},
  number    = {11},
  pages     = {5999--6027},
  volume    = {61},
  doi       = {10.1109/TIT.2015.2462340},
}

@Article{trelea2003particle,
  author    = {Ioan Cristian Trelea},
  journal   = {Inf. Process. Lett.},
  title     = {The particle swarm optimization algorithm: convergence analysis and parameter selection},
  year      = {2003},
  number    = {6},
  pages     = {317--325},
  volume    = {85},
  doi       = {10.1016/S0020-0190(02)00447-7},
}

@Article{dirk2011latticereduction,
  author    = {Dirk W{\"{u}}bben and Dominik Seethaler and Joakim Jald{\'{e}}n and Gerald Matz},
  journal   = {{IEEE} Signal Process. Mag.},
  title     = {Lattice Reduction},
  year      = {2011},
  number    = {3},
  pages     = {70--91},
  volume    = {28},
  doi       = {10.1109/MSP.2010.938758},
}

@Article{sayed2014adaptation,
  author    = {Ali H. Sayed},
  journal   = {Found. Trends Mach. Learn.},
  title     = {Adaptation, Learning, and Optimization over Networks},
  year      = {2014},
  number    = {4-5},
  pages     = {311--801},
  volume    = {7},
  doi       = {10.1561/2200000051},
}

@Article{rudin1976principles,
  author  = {Rudin, Walter},
  journal = {3rd ed.},
  title   = {Principles of mathematical analysis},
  year    = {1976},
}

@InProceedings{NIPS2017_491442df,
  author    = {Jason M. Altschuler and Jonathan Weed and Philippe Rigollet},
  booktitle = {Advances in Neural Information Processing Systems 30: Annual Conference on Neural Information Processing Systems 2017, December 4-9, 2017, Long Beach, CA, {USA}},
  title     = {Near-linear time approximation algorithms for optimal transport via Sinkhorn iteration},
  year      = {2017},
  editor    = {Isabelle Guyon and Ulrike von Luxburg and Samy Bengio and Hanna M. Wallach and Rob Fergus and S. V. N. Vishwanathan and Roman Garnett},
  pages     = {1964--1974},
}

@Article{SakzadHV13,
  author    = {Amin Sakzad and Jagadeesh Harshan and Emanuele Viterbo},
  journal   = {{IEEE} Trans. Wirel. Commun.},
  title     = {Integer-Forcing {MIMO} Linear Receivers Based on Lattice Reduction},
  year      = {2013},
  number    = {10},
  pages     = {4905--4915},
  volume    = {12},
  doi       = {10.1109/TWC.2013.090513.121465},
}

@InProceedings{HarshanSV17,
  author    = {Jagadeesh Harshan and Amin Sakzad and Emanuele Viterbo},
  booktitle = {2017 {IEEE} Wireless Communications and Networking Conference, {WCNC} 2017, San Francisco, CA, USA, March 19-22, 2017},
  title     = {Integer-Forcing Linear Receivers: {A} Design Criterion for Full-Diversity STBCs},
  year      = {2017},
  pages     = {1--6},
  publisher = {{IEEE}},
  bibsource = {dblp computer science bibliography, https://dblp.org},
  doi       = {10.1109/WCNC.2017.7925751},
  url       = {https://doi.org/10.1109/WCNC.2017.7925751},
}

@Article{KrunzS23,
  author    = {Marwan Krunz and Peyman Siyari},
  journal   = {{IEEE} Trans. Commun.},
  title     = {Secure Linear Precoding in Overloaded {MU-MIMO} Wireless Networks},
  year      = {2023},
  number    = {7},
  pages     = {4050--4061},
  volume    = {71},
  bibsource = {dblp computer science bibliography, https://dblp.org},
  doi       = {10.1109/TCOMM.2023.3267852},
  url       = {https://doi.org/10.1109/TCOMM.2023.3267852},
}

@Article{LiangQSB06,
  author    = {Jing J. Liang and A. K. Qin and Ponnuthurai N. Suganthan and S. Baskar},
  journal   = {{IEEE} Trans. Evol. Comput.},
  title     = {Comprehensive learning particle swarm optimizer for global optimization of multimodal functions},
  year      = {2006},
  number    = {3},
  pages     = {281--295},
  volume    = {10},
  doi       = {10.1109/TEVC.2005.857610},
}

\end{document}